%% file: main.tex
\documentclass[a4paper,UKenglish,cleveref,autoref,numberwithinsect]{lipics-v2021}
\nolinenumbers
\usepackage{mathtools,stmaryrd,nicefrac}
\usepackage{dsfont}
\allowdisplaybreaks[3] 
\usepackage{bm} 
\usepackage{multicol}
\usepackage{fancybox,framed} 
\usepackage{float}
\usepackage{proof}
\usepackage{xr}
\usepackage{tikz-cd}
\usepackage{subfiles}
\usepackage{adjustbox}
\usepackage{vwcol}
\usepackage{tikzit}

\usepackage{macros}
\providecommand{\Description}[1]{}

\input{figures/tikzit.tikzstyles}
\input{figures/tikzit.tikzdefs}

\title{A Complete Equational Presentation of Qudit Circuits via Polycontrolled PROPs}
\titlerunning{Complete Equational Presentation of Qudit Circuits}
\author{Colin Blake}
  {Universit\'e de Lorraine, CNRS, Inria, LORIA, Nancy, France}
  {colin.blake@inria.fr}
  {https://orcid.org/0009-0000-4045-8145}
  {}
\authorrunning{C. Blake}
\Copyright{Colin Blake}
\funding{This work is supported by the Plan France 2030 through the PEPR integrated project EPiQ (ANR-22-PETQ-0007) and the HQI platform (ANR-22-PNCQ-0002); by the European Union through the MSCA Staff Exchanges project QCOMICAL (Grant Agreement ID: 101182520); and by the Maison du Quantique MaQuEst.}
\acknowledgements{I thank my PhD advisors, Simon Perdrix and Miriam Backens, for their helpful feedback and for reviewing earlier versions of this article.}

\begin{document}

\maketitle

\begin{abstract}
High-dimensional quantum computation needs a native circuit-level equational
theory for qudits. We give the first finite schematic equational theory that is
sound and complete for exact unitary qudit circuits in every finite dimension
at least two. Circuits are built from local gates, sequential and parallel
composition, and value-controls; equality is derivable exactly when the
standard unitary denotations agree. For each dimension, a finite list of local
bounded-arity axiom schemata presents the theory, and the diagrammatic shapes
do not depend on \(d\). Primitive value-control makes control on a chosen
basis value part of the language, so local rules generate the internal algebra
of controlled operations within the circuit PROP. This gives a finite,
dimension-uniform basis for exact equational reasoning about qudit circuits.
\end{abstract}
\clearpage

\subfile{introduction}
\subfile{polyprop_mfcs}
\subfile{set_of_rules} 
\subfile{embedding_mfcs}
\subfile{conclusion}

\clearpage
\bibliography{qudit-complete-theory}
\clearpage          
\appendix    
\subfile{appendix_proof_guide}
\subfile{appendix_derived_gates}
\subfile{appendix_anglesrelations}

\subfile{appendix_derivations}
\subfile{appendix_compatibility}
\subfile{appendix_commuting}
\subfile{appendix_exhaustivity}

\subfile{appendix_derivations_after_commuting}
\subfile{appendix_gray}
\subfile{appendix_encodingdecoding}
\subfile{appendix_mimicking}

\subfile{appendix_dec_swap}
\subfile{appendix_soundness}

\end{document}

%% file: figures/tikzit.tikzstyles
\tikzstyle{boitecircuit}=[fill=white, draw, shape=rectangle, minimum height=0.6cm, inner sep=0.1em, minimum width=0.5cm]
\tikzstyle{dot}=[fill=black, draw=none, shape=circle, inner sep=0pt, outer sep=0pt, scale=0.891]
\tikzstyle{operator}=[draw, fill=white, minimum size=1.5em]
\tikzstyle{control}=[fill=white, draw=black, inner sep=0.14em, shape=rounded rectangle, minimum height=0.32cm, minimum width=0.28cm]
\tikzstyle{vdotsnode}=[inner sep=0pt, outer sep=0pt]
\tikzstyle{phase}=[fill=white, draw=black, inner sep=0.15em, shape=rounded rectangle, minimum height=0.4cm]

\tikzstyle{new}=[-, tikzit draw=magenta]
\tikzstyle{tirets}=[-, draw=black, dashed]
\tikzstyle{boxed}=[-, draw=gray, thick, dashed]
\tikzstyle{productframe}=[-, draw=gray, densely dashed]
\tikzstyle{bangbox}=[-, draw=gray, densely dashed]

%% file: introduction.tex
\section{Introduction}

Qudit circuits make higher-dimensional quantum data visible at circuit level.
Basis values of a \(d\)-level carrier can be named, controlled, and rewritten
directly, which matters both for hardware where extra levels are native and for
compilers that exploit those levels as structure.
For \(d>2\), high-dimensional carriers are a native resource
\cite{qudit_high_dim}, and multilevel entangling operations have already been
demonstrated, for example, in trapped-ion processors
\cite{native_qudit_entanglement}. The
extra levels also matter when the algorithm is logically qubit-based. Qutrit
workspace can implement generalized Toffoli gates without external ancillas and
with lower depth, while ququart encodings can shorten three-qubit constructions
\cite{qutrit_asymptotic,ququart_waltz}; the small example in
Figure~\ref{fig:ququart-compression} shows the mechanism. Packing two qubits
into one four-level system turns two binary controls into one value-control and
turns a two-qubit target subspace into an adjacent two-level ququart subspace.

\begin{figure}[H]
\centering
\begingroup
\begin{adjustbox}{max width=.95\linewidth}
\begin{tabular}{c@{\qquad}c@{\qquad}c}
\input{./figures/ququart-compression-qubits.tikz}
&
\begin{tabular}{c}
\(\xrightarrow{\ \textsf{compress}\ }\)\\[-0.2ex]
\(\scriptstyle A=2a_1+a_0\)\\[-0.2ex]
\(\scriptstyle B=2b_1+b_0\)
\end{tabular}
&
\input{./figures/ququart-compression-ququarts.tikz}
\end{tabular}
\end{adjustbox}
\endgroup
\caption{Compressing four qubits into two ququarts via
\(\ket{a_1a_0}\ket{b_1b_0}\mapsto\ket{A}\ket{B}\). The branch
\(a_1=a_0=1\) becomes \(A=3\), while the target subspace
\(\mathrm{span}\{\ket{01},\ket{10}\}\) becomes
\(\mathrm{span}\{\ket1,\ket2\}\), giving \(\ctrl_3(H^{(1,2)})\).}
\label{fig:ququart-compression}
\end{figure}

Such compression is useful only if the packed circuits can still be
reasoned about locally. A compiler may introduce value-controls, adjacent
two-level gates, or swaps between encoded levels; a verifier then has to show
that these local changes preserve the intended unitary. This paper studies
circuit-level rewriting for qudit circuits.

Quantum circuits provide the standard low-level syntax for finite-dimensional
quantum computation \cite{quantum-circuits}: wires carry $d$-level quantum
data, boxes represent unitary gates, and an $n$-wire circuit denotes a unitary
operator on $(\mathbb{C}^d)^{\otimes n}$. A
rewrite system for such circuits should therefore be stable under the two
ways circuits are built, sequential composition and parallel composition. In
PROP language, where objects are wire counts, this asks for local equations
between diagrams, closed under those compositions and under the structural
symmetries. We ask whether equality of unitary qudit circuits can be
axiomatized by a \emph{finite schematic}, syntax-directed equational theory in
the circuit PROP.

For qubits ($d=2$), the answer is known to be yes
\cite{qubit-theory}, with later simplifications and minimal presentations
\cite{qubit-theory-simplified,qubit-theory-minimal}. The qubit proofs also
identify the pressure point. Control cannot be treated as an afterthought.
Without it in the syntax, completeness tends to reintroduce it through rules
with arbitrarily many explicit control wires. Controlled PROPs avoid
that blow-up by treating control functorially, which restores finite,
bounded-arity axiom sets \cite{delorme-perdrix-control} and connects with the
categorical account of control constructors in \cite{lemonier-control}.

Graphical calculi provide complementary complete languages for quantum
processes, including qudit and finite-dimensional variants of ZX/ZW-style
systems
\cite{zx-calculus,zx-complete,zh-calculus,qufinite-zx,qudit-zh,qufinite-zxw-complete,zx-fdhilb-complete,zw-fd-complete}.
These calculi usually describe a broad class of linear maps. When the endpoint
is an ordinary unitary circuit, a proof may finish with a separate
reconstruction of a circuit from the whole diagram, or from the matrix it
denotes, as in circuit extraction from ZX diagrams \cite{zx-to-circuits}.
Existing circuit-level qutrit and qudit results are more specialised:
phase-gadget methods treat qutrit diagonal gates
\cite{qutrit-diagonal-gates}, complete rule sets cover multi-qutrit Clifford
fragments \cite{qutrit-clifford}, and generalized-Clifford formalisms handle
multi-qudit computations
\cite{Lin2021algebraicframework,Lin2025graphicalcalculus}. The case left open
is the full unitary qudit circuit category itself, where wire arities are
objects, circuits are morphisms, and, for every \(d\ge2\), one wants a
bounded-arity list of local axiom schemata whose shape is uniform in \(d\).

Beyond qubits, the difficulty is to find a qudit-native circuit presentation on
which completeness can still be proved by finite local rules. Several binary
shortcuts in the qubit completeness proof of \cite{qubit-theory} become
case-sensitive for adjacent two-level gates. A local optical relation may decode
either to a single adjacent-level chain or to a mixed row/column configuration
guarded by basis controls. The proof therefore needs structured control
management while remaining inside a native qudit circuit language.

Our answer is to make each value-control primitive. For every basis value
$k\in[d]$, a control functor $\ctrl_k$ sends an endomorphism $f:n\to n$ to its
$k$-controlled extension $\ctrl_k(f):1+n\to1+n$. The axioms may mention
$\ctrl_k(-)$ directly, so the rule set does not need a variable number of
control wires.

\textbf{Main result.}
For each finite $d\ge 2$, we define a polycontrolled circuit PROP $\CQC$ and a
finite schematic equational theory $\mathrm{QC}_d$ generated by local axiom schemata
uniform in $d$ whose instances involve at most three wires; compatibility,
commutativity, and exhaustivity are derived as the global control algebra in
\(\QCeq\).  The primitive structure of \(\CQC\) contains
only the same-control functor laws.
The main completeness theorem states that $\QCeq$ is sound and complete
for the standard unitary semantics $\interp{-}:\CQC\to\Qudit$: for any
$n$-qudit circuits $C_1,C_2:n\to n$, $\interp{C_1}=\interp{C_2}$ if and only if
$\QCeq \vdash C_1=C_2$.

The completeness proof has three steps. Finite bounded-arity axiom schemata
supply the local rules; the control/support schemata derive the control-algebra laws
(Theorems~\ref{thm:compatibility}--\ref{thm:exhaustive}); Gray-ordered
linear-optical transfer proves Theorem~\ref{thm:main-completeness}.  The
transfer encodes \(n\)-qudit
circuits as \(d^n\)-mode single-photon \(\LOPP\) circuits, serialises tensor to
match the optical direct-sum tensor, decodes optical derivations back into
\(\CQC\), and proves an encoding/decoding retraction. When \(d=2\), this
specializes to the qubit transfer pattern of \cite{qubit-theory}; for larger
\(d\), reflected Gray order keeps optical neighbours adjacent in one qudit
coordinate.

%% file: polyprop_mfcs.tex
\section{Qudit circuits as a polycontrolled PROP}\label{sec:qudit-circuits}
Fix a finite dimension \(d\ge 2\) and write \([d]=\{0,\ldots,d-1\}\).
A PROP is a strict symmetric monoidal category whose objects are natural
numbers and whose tensor on objects is addition~\cite{Saunders}.  PROPs give
the ambient circuit calculus used throughout the paper: objects are wire counts
and morphisms are diagrams modulo strict symmetric-monoidal coherence. Since
all gates considered here are unitary, the only non-empty hom-sets are
endomorphism hom-sets.

The extra structure is value control. In ordinary circuit notation one freely
draws a control and reads it by its action on a branch; here this operation is
part of the categorical syntax. For a detailed analysis of controlled and polycontrolled
PROPs and their circuit applications, see \cite{delorme-perdrix-control}; here
we use only the control-functor axioms below.

\begin{definition}\label{def:polycontrolled-PROP}
For a PROP \(\cat{P}\), let \(\cat{P}_{\mathrm{endo}}\) be the wide subcategory
with the same objects and only endomorphism hom-sets. A \emph{control functor}
is a functor \(\ctrl:\cat{P}_{\mathrm{endo}}\to\cat{P}_{\mathrm{endo}}\) such
that \(\ctrl(n)=1+n\) on objects.  For every \(f,g:n\to n\), \(m\ge0\), and
permutation \(\pi:n\to n\), it satisfies
\[
  \begin{array}{r@{\;}c@{\;}l@{\quad}r@{\;}c@{\;}l@{\quad}r@{\;}c@{\;}l}
  \ctrl(\id_n)&=&\id_{1+n},&
  \ctrl(g\circ f)&=&\ctrl(g)\circ\ctrl(f),&
  \ctrl(f\otimes\id_m)&=&\ctrl(f)\otimes\id_m,\\[0.3ex]
  \multicolumn{9}{c}{
  \ctrl(\pi^{-1}\circ f\circ\pi)
  =(\id_1\otimes\pi^{-1})\circ\ctrl(f)\circ(\id_1\otimes\pi),
  }\\[0.3ex]
  \multicolumn{9}{c}{
  \ctrl(\ctrl(f))\circ(\sigma_{1,1}\otimes\id_n)
  =(\sigma_{1,1}\otimes\id_n)\circ\ctrl(\ctrl(f)).
  }
\end{array}
\]

A \(d\)-ary \emph{polycontrolled PROP} is a PROP equipped with control functors
\((\ctrl_k)_{k\in[d]}\).
\end{definition}

In this syntax, \(\ctrl_k(f):1+n\to 1+n\) is a circuit with one additional
distinguished control wire. The equations say that unused wires remain unused,
target-wire permutations commute with adding a control, and two nested controls
of the same value may be swapped with the two control wires.

For a word \(u=u_1\cdots u_m\in[d]^*\) we write
\(\ctrl_u(f):=\ctrl_{u_1}(\cdots\ctrl_{u_m}(f)\cdots)\), with
\(\ctrl_\epsilon(f)=f\). Beyond the single-value control functor axioms, the
later proofs use three global principles, compatibility between different
values, commutation of different branches, and exhaustivity over all values.
These principles are valid for the unitary interpretation of \(\CQC\) defined
below.  The structural congruence of circuits contains only the same-control
functor laws, and Section~\ref{subsec:derived-control-algebra} derives the
global principles from the finite axiom system.

\begin{definition}\label{def:control-algebra}
Let $(\ctrl_k)_{k\in[d]}$ be a family of control functors. For well-typed
endomorphisms \(f,g:n\to n\), we use the following three control-algebra laws.

\begin{enumerate}
\item The family is \emph{compatible} when, for any \(a,b\in[d]\), the two
ways of nesting \(a\)- and \(b\)-controls around the same circuit agree up to
the control-wire symmetry \(\sigma_{1,1}:2\to2\), i.e.\
\(\ctrl_a(\ctrl_b(f))\circ(\sigma_{1,1}\otimes\id_n)
=
(\sigma_{1,1}\otimes\id_n)\circ\ctrl_b(\ctrl_a(f))\).

\item The family is \emph{commutative} when, for distinct \(a,b\in[d]\),
operations guarded by different values of the same control wire commute,
i.e.\ \(\ctrl_a(f)\circ\ctrl_b(g)=\ctrl_b(g)\circ\ctrl_a(f)\).

\item The family is \emph{exhaustive} when all value-controlled branches of
the same circuit compose to the uncontrolled target operation,
i.e.\ \(\ctrl_0(f)\circ\ctrl_1(f)\circ\cdots\circ\ctrl_{d-1}(f)=\id_1\otimes f\).
The product in the exhaustivity law is composed in increasing order of the
control value; once commutativity is available, this order is immaterial.
\end{enumerate}
\end{definition}
In the matrix model, these branch laws say that on each computational-basis
value of the control wire exactly one guarded operation is active.

\begin{definition}\label{def:support-sensitive}
We define a \(d\)-ary polycontrolled PROP to be \emph{support-sensitive} when each
endomorphism \(f:n\to n\) is assigned a basis-support envelope
\(\mathrm{supp}(f)\subseteq[d]^n\) such that endomorphisms with disjoint
supports commute: if \(\mathrm{supp}(f)\cap\mathrm{supp}(g)=\varnothing\), then
\(f\circ g=g\circ f\).  Supports must also be stable under the circuit
constructors: for \(f,g:n\to n\) and \(h:m\to m\),
\(\mathrm{supp}(f\circ g)\subseteq \mathrm{supp}(f)\cup\mathrm{supp}(g)\),
\(\mathrm{supp}(f\otimes h)\subseteq
(\mathrm{supp}(f)\times[d]^m)\cup([d]^n\times\mathrm{supp}(h))\),
wire symmetries transport supports by the corresponding permutation of basis
words, and \(\mathrm{supp}(\ctrl_k(f))=\{k\}\times\mathrm{supp}(f)\).
\end{definition}
\begin{remark}
Support-sensitivity is stronger than commutativity:
if \(a\neq b\), then \(\mathrm{supp}(\ctrl_a(f))\) and
\(\mathrm{supp}(\ctrl_b(g))\) are disjoint, so the controlled operations
commute.
\end{remark}

The structural PROP supplies the empty diagram, identity wires, symmetries,
sequential composition, and tensor.
Set \(\mathsf G_d :=
\{\input{./figures/phase.tikz}:0\to0 \mid \theta\in\mathbb{R}\}\cup
\{\input{./figures/Hadiip.tikz}:1\to1 \mid 0\le r<d-1\}\). This non-structural gate
family consists of scalar phases indexed by \(\theta\) and adjacent two-level
Hadamards indexed by \(r\).
A phase on one basis level is expressed as \(\ctrl_i(\input{./figures/phase.tikz})\), which
semantically multiplies \(\ket i\) by \(e^{i\theta}\) and fixes the other basis
states.

\begin{remark}
The choice of adjacent two-level generators comes from the transfer. Under
reflected Gray order, an adjacent optical beam splitter connects two basis words
that differ in one digit by \(\pm1\) (Lemma~\ref{lem:gray-neighbours}). Its
qudit decoding acts on two neighbouring levels, possibly under controls, so the
transferred equations stay local and bounded in arity. The same choice is close
to common qudit gate descriptions, where single-qudit gates are often expressed
as rotations on selected pairs of levels
\cite{qudit_high_dim,qutrit_asymptotic,ququart_waltz}. Permutations of basis
values can still be built from adjacent swaps.
\end{remark}

\begin{definition}\label{def:raw-controlled-circuits}
Let \(\RawCQC_d\) be the raw endomorphism syntax generated by the structural
PROP operations, the gates in \(\mathsf G_d\), and the unary constructors
\(\ctrl_k(-)\) for \(k\in[d]\), with no non-endomorphism homs.
\end{definition}

\begin{definition}\label{def:structural-congruence-CQC_d}
Let \(\equiv\) be the least congruence, closed under \(\circ\), \(\otimes\),
and every \(\ctrl_k(-)\), containing the strict PROP coherence laws and the
control-functor equations for each \(\ctrl_k\).
\end{definition}

\begin{definition}\label{def:CQC_d}
Set \(\CQC(n,n):=\RawCQC_d(n,n)/{\equiv}\), with
\(\CQC(m,n):=\varnothing\) for \(m\neq n\).
\end{definition}
The quotient \(\equiv\) contains only structural equations. Compatibility,
commutativity, and exhaustivity must still be proved inside \(\QCeq\).

Let \(\Qudit\) be the semantic polycontrolled PROP with
\(\Qudit(n,n)=U(d^n)\) and no non-endomorphism homs. Symmetries are interpreted
by wire swaps, and the remaining clauses are
\[
\begin{array}{rcl@{\qquad}rcl@{\qquad}rcl}
\interp{C_2\circ C_1} &=& \interp{C_2}\interp{C_1},&
\interp{C_1\otimes C_2} &=& \interp{C_1}\otimes\interp{C_2},&
\interp{\input{./figures/phase.tikz}} &=& e^{i\theta},\\[0.4ex]
\interp{\input{./figures/Hadiip.tikz}}&=&\multicolumn{7}{l}{
\text{Hadamard on }\mathrm{span}\{\ket r,\ket{r+1}\}
\text{ and identity elsewhere,}}\\[0.4ex]
\interp{\ctrl_k(C)}&=&\multicolumn{7}{l}{
\ket{k}\bra{k}\otimes\interp{C}+
\sum_{\ell\in[d],\,\ell\neq k}\ket{\ell}\bra{\ell}\otimes I.}
\end{array}
\]
\begin{proposition}\label{PROP:interp-well-defined}
The interpretation respects the structural congruence and hence factors through
a strict symmetric monoidal functor \(\interp{-}:\CQC\to\Qudit\).
\end{proposition}
\begin{proof}
The clauses are the standard matrix interpretation of strict monoidal
circuits, so identities, composition, tensor, and symmetries respect the PROP
coherence laws. The displayed block formula for \(\interp{\ctrl_k(C)}\) is
functorial in \(C\). It is also stable under padding and target-wire
permutations, and it makes two equal-value controls symmetric after swapping
the control wires.
Thus every generator of the structural congruence has equal denotation, and the
interpretation descends to a strict symmetric monoidal functor on \(\CQC\).
\end{proof}

This small signature still reaches all finite-dimensional unitaries.
\begin{proposition}\label{PROP:CQC-universal}
For every $n\ge0$, \(\interp{-}\) is surjective on \(n\)-wire endomorphisms.
\end{proposition}
\begin{proof}
For \(n=0\), scalar phases realise \(U(1)\). For one qudit, adjacent two-level
decompositions reduce an arbitrary element of \(U(d)\) to adjacent supports
\cite{nielsen-chuang-decomp,Clements,su3-decomp,su3-decomp2}, each realised by
Euler decompositions into the adjacent Hadamard, level phases, and a scalar
phase \cite{nielsen-chuang-decomp,qubit-universal-set}. For \(n>1\), these
one-qudit unitaries together with the controlled scalar
\(\interp{\ctrl_0(\ctrl_0(\scalebox{0.8}{\input{./figures/piphase.tikz}}))}\) give the
exact universal family of \cite{qudit-universal}: the controlled scalar is
entangling for \(d\ge2\), scalar phases supply global phases, and tensors with
identities and symmetries place the generators on any chosen wires.
\end{proof}

\begin{definition}\label{def:eq-theory-CQC}
An equational theory $\Gamma$ on $\CQC$ is a set of well-typed equations closed
under equivalence closure, sequential composition, tensor, and every control
constructor $\ctrl_k(-)$. We write $\Gamma\vdash f=g$ for the generated
congruence.
\end{definition}
\begin{definition}\label{def:QC_d}
For each $d\ge2$, $\QCeq$ is the equational theory on $\CQC$ generated by the
local axiom schemata displayed in Figures~\ref{fig:axioms_QC_part1}
and~\ref{fig:axioms_QC_part2}.
\end{definition}
Here \(\QCeq\) is the congruence generated by the displayed rules. The transfer
argument of Section~\ref{sec:embedding-lopp} proves that this congruence is
exact: Theorem~\ref{thm:main-completeness} identifies it with equality of
unitary denotations.

%% file: set_of_rules.tex
\section{A Finite Axiom System for Qudit Circuits}\label{sec:axioms}

The presentation below consists of the gate algebra of the qudit signature and
the finite branch-control rules used by the later normalisation arguments. All equations
in Figures~\ref{fig:axioms_QC_part1}--\ref{fig:axioms_QC_part2} are read in
\(\CQC\), where the strict PROP laws and same-control coherence equations have
already been quotiented out.

\subsection{Reading the axiom schemata}\label{subsec:derived-gates}

The schemata below use a slightly higher-level notation, but no extra
generators are being added. The non-structural generators are the gates
\(\mathsf G_d\) of Section~\ref{sec:qudit-circuits}; identities, symmetries,
and the control constructors \(\ctrl_k(-)\) come from the polycontrolled PROP
structure. Here \(X^{(i,j)}\), \(H^{(i,j)}\), and \(R_x^{(i,j)}(\theta)\) are
derived one-qudit circuits: semantically they are, respectively, a
transposition of levels \(i,j\), a Hadamard on
\(\mathrm{span}\{\ket i,\ket j\}\), and the corresponding two-level
beam-splitter rotation.  Appendix~\ref{app:derived-notation} gives their
recursive definitions.
For instance, when \(d=3\), the non-adjacent swap is the adjacent walk
\(X^{(0,2)}=X^{(1,2)}\circ X^{(0,1)}\circ X^{(1,2)}\); the same pattern gives
\(H^{(0,2)}\) and \(R_x^{(0,2)}(\theta)\).

The figures use \emph{indexed product frames} as meta-notation for finite
sequential products. A frame labelled \(\vec k:\varphi(\vec k)\) binds the
displayed index tuple and expands to the composite of the enclosed diagram over
all satisfying tuples, in lexicographic increasing order:
\begin{center}
\setlength{\tabcolsep}{0pt}
\begin{tabular}{@{}c@{}}
\adjustbox{max width=0.92\linewidth}{%
\(\input{./figures/bangbox-p.tikz}\qquad\Longrightarrow\qquad
\input{./figures/bangbox-unfold.tikz}\)}\\[0.45em]
\adjustbox{max width=0.86\linewidth}{%
\(\text{for }d=3:\quad
\input{./figures/bangbox-example.tikz}=
\input{./figures/bangbox-example-unfold.tikz}\)}
\end{tabular}
\end{center}
A tuple label such as \((a,b):0\le a<b<d\) is read lexicographically, with
\(a\) the outer index and \(b\) the inner index. Product frames expand on the
same wires, so they add repetitions without increasing arity.
Read the figures as local support schemata.  Increasing \(d\)
changes the admissible level labels and the finite products expanded by frames,
not the arity of the displayed diagrams.

\subsection{The axiom schemata}

\begin{theorem}\label{thm:bounded-arity-axioms}
For every \(d\ge2\), \(\QCeq\) is a finite schematic equational theory on
\(\CQC\).  The diagrammatic shape of the rules is independent of \(d\), and
every expanded instance acts on at most three wires.
\end{theorem}
\begin{proof}
The only \(d\)-dependent data are level indices and product frames.  Product
frames expand to sequential products on the same wires, while the displayed
diagrams themselves contain at most three wires.  Each angle variable ranges
over \(\mathbb R\); finiteness refers to the finite list of rule shapes.
\end{proof}

The constructor \(\ctrl_k(-)\) is primitive in \(\CQC\), while the structural
congruence contains only the functor laws for a fixed control value.  The
equations below must therefore generate the algebra relating different control
values.  Figure~\ref{fig:axioms_QC_part1} gives the gate algebra.
Figure~\ref{fig:axioms_QC_part2} lists the bounded control/support schemata
used to derive the control laws of Definition~\ref{def:control-algebra}.

\begin{figure*}[!t]
\centering
\begingroup
\setlength{\fboxsep}{2pt}
\fbox{
\begin{minipage}{\dimexpr0.985\textwidth-2\fboxsep-2\fboxrule\relax}
\centering
\setlength{\abovedisplayskip}{2pt}
\setlength{\belowdisplayskip}{2pt}
\setlength{\abovedisplayshortskip}{2pt}
\setlength{\belowdisplayshortskip}{2pt}
\begin{minipage}[c]{0.34\linewidth}
\centering
\schemarule{Sum}{sum}{\input{./figures/two-phases.tikz} = \input{./figures/phase-sum.tikz}}
\end{minipage}
\hspace{0.05\linewidth}
\begin{minipage}[c]{0.24\linewidth}
\centering
\schemarule{$2\pi$}{2pi}{\input{./figures/phase2pi.tikz} = \input{./figures/empty_diagram.tikz}}
\end{minipage}\\[0.12em]
\begin{minipage}[c]{0.48\linewidth}
\centering
\schemarule{XH}{XH}{\input{./figures/XH.tikz} = \input{./figures/HX.tikz}}
\end{minipage}
\hspace{0.04\linewidth}
\begin{minipage}[c]{0.38\linewidth}
\centering
\schemarule{H$^{2}$}{Hiip-Hiip}{\input{./figures/Hiip-Hiip.tikz} = \input{./figures/line.tikz}}
\end{minipage}\\[0.12em]
\begin{minipage}[c]{0.88\linewidth}
\centering
\schemarule{EH}{eulerH}{\input{./figures/eulerHL.tikz} = \input{./figures/eulerHR.tikz}}
\end{minipage}\\[0.12em]
\begin{minipage}[c]{0.96\linewidth}
\centering
\schemarule{3Rx}{3Rx}{\input{./figures/3bsL.tikz} = \input{./figures/3bsR.tikz}}
\end{minipage}\\[0.12em]
\begin{minipage}[c]{0.985\linewidth}
\centering
\schemarule{CXC}{CX-XC-CX}{\input{./figures/CX-XC-CX.tikz} = \input{./figures/XC-CX-XC.tikz}}
\end{minipage}\\[0.12em]
\begin{minipage}[c]{0.64\linewidth}
\centering
\schemarule{S}{swap-decomp}{\input{./figures/CX-XC-CX-bangbox.tikz} = \input{./figures/swap.tikz}}
\end{minipage}
\end{minipage}}
\endgroup
\caption{
\textbf{Core local axiom schemata for \(\mathrm{QC}_{d}\).}
Each displayed equality is a schema: level labels range over \([d]\), and
product frames abbreviate finite products over the indices written in their
labels.
In \eqref{eulerH} and \eqref{3Rx} the right-hand-side angles are any
admissible choices for the indicated Euler decompositions; Appendix~\ref{appendix:relations_angles}
describes this convention.
}
\Description{Framed collection of circuit-diagram rewrite axioms defining the core rules of the equational
theory QC_d: phase arithmetic, basic two-level identities involving swaps and Hadamards,
a swap decomposition into controlled operations, and the local decomposition, rotation, and classical schemata
\eqref{eulerH}, \eqref{3Rx}, and \eqref{CX-XC-CX}. Each axiom is shown as an equality between a left-hand and a
right-hand circuit diagram.}
\label{fig:axioms_QC_part1}
\end{figure*}

\begin{figure}[!t]
\centering
\begingroup
\setlength{\fboxsep}{2pt}
\fbox{
\begin{minipage}{\dimexpr0.985\textwidth-2\fboxsep-2\fboxrule\relax}
\centering
\setlength{\abovedisplayskip}{2pt}
\setlength{\belowdisplayskip}{2pt}
\setlength{\abovedisplayshortskip}{2pt}
\setlength{\belowdisplayshortskip}{2pt}
\newcommand{\axiomsep}{\par\vspace{0.01em}\noindent\rule{0.96\linewidth}{0.25pt}\par\vspace{0.01em}}
\begin{minipage}[c]{0.43\linewidth}
\centering
\schemarule{ExP}{axiom-total-phase}{\input{./figures/axiom-total-phase-L.tikz} = \input{./figures/axiom-total-phase-R.tikz}}
\end{minipage}
\hspace{0.04\linewidth}
\begin{minipage}[c]{0.43\linewidth}
\centering
\schemarule{ExH}{axiom-total-hadamard}{\input{./figures/axiom-total-hadamard-L.tikz} = \input{./figures/axiom-total-hadamard-R.tikz}}
\end{minipage}\\[-0.02em]
\axiomsep
\begin{minipage}[c]{0.58\linewidth}
\centering
\schemarule{SHH}{HHcomm}{\input{./figures/HHcommA.tikz} = \input{./figures/HHcommB.tikz}}
\end{minipage}\\[0.03em]
\begin{minipage}[c]{0.43\linewidth}
\centering
\schemarule{SHP}{Hphasecomm}{\input{./figures/HphasecommA.tikz} = \input{./figures/HphasecommB.tikz}}
\end{minipage}
\hspace{0.04\linewidth}
\begin{minipage}[c]{0.43\linewidth}
\centering
\schemarule{SH$\pi$}{Hcnot}{\input{./figures/HcnotA.tikz} = \input{./figures/HcnotB.tikz}}
\end{minipage}\\[-0.02em]
\axiomsep
\begin{minipage}[c]{0.43\linewidth}
\centering
\schemarule{BPP}{cp-cp}{\input{./figures/comm-phasesL.tikz} = \input{./figures/comm-phasesR.tikz}}
\end{minipage}
\hspace{0.04\linewidth}
\begin{minipage}[c]{0.43\linewidth}
\centering
\schemarule{BP$\pi$}{cp-cnot}{\input{./figures/comm-phasesbasepiL.tikz} = \input{./figures/comm-phasesbasepiR.tikz}}
\end{minipage}\\[0.03em]
\begin{minipage}[c]{0.43\linewidth}
\centering
\schemarule{B$\pi\pi$d}{cnot-cnot-diff}{\input{./figures/comm-phasescpipi-diffL.tikz} = \input{./figures/comm-phasescpipi-diffR.tikz}}
\end{minipage}
\hspace{0.04\linewidth}
\begin{minipage}[c]{0.43\linewidth}
\centering
\schemarule{BH$\pi$}{ch-ccnot}{\input{./figures/Hphase-commpi-sameL.tikz} = \input{./figures/Hphase-commpi-sameR.tikz}}
\end{minipage}\\[0.03em]
\begin{minipage}[c]{0.46\linewidth}
\centering
\schemarule{BHH}{ch-ch}{\input{./figures/CHCH-comm-sameL.tikz} = \input{./figures/CHCH-comm-sameR.tikz}}
\end{minipage}
\hspace{0.04\linewidth}
\begin{minipage}[c]{0.43\linewidth}
\centering
\schemarule{BHP}{ch-ccp}{\input{./figures/HphasecommcA.tikz} = \input{./figures/HphasecommcB.tikz}}
\end{minipage}\\[0.03em]
\begin{minipage}[c]{0.43\linewidth}
\centering
\schemarule{B$\pi\pi$s}{cnot-cnot}{\input{./figures/comm-phasescpi-sameL.tikz} = \input{./figures/comm-phasescpi-sameR.tikz}}
\end{minipage}
\\[-0.02em]
\axiomsep
\begin{minipage}[c]{0.36\linewidth}
\centering
\schemarule{CoP}{axiom-compat}{\input{./figures/swap-cacbpi.tikz} = \input{./figures/cbcapi-swap.tikz}}
\end{minipage}
\hspace{0.04\linewidth}
\begin{minipage}[c]{0.36\linewidth}
\centering
\schemarule{Co$\pi$}{axiom-compat-pi}{\input{./figures/swap-cacbpi2.tikz} = \input{./figures/cbcapi-swap2.tikz}}
\end{minipage}
\end{minipage}}
\endgroup
\caption{
\textbf{Control-support axiom schemata for \(\mathrm{QC}_{d}\).}
The prefixes Ex, S, B, and Co mark finite schemata for totalisation,
support-locality, branch-commutation, and compatibility.
The \(S\)-rules use pairwise distinct level labels: \(i,j,k\), and
\(i,j,k,\ell\) in \eqref{HHcomm}.  The \(B\)-rules use distinct outer control
values \(k\neq\ell\).
}
\Description{Framed collection of control-support circuit-diagram rewrite axioms for QC_d. It includes the totalisation schemata
\eqref{axiom-total-phase} and \eqref{axiom-total-hadamard} used to derive exhaustivity of control, support-locality rules for Hadamards acting on disjoint basis levels, a finite factor-commutation basis for
operations controlled on different basis values (the cases with \(k\neq \ell\)), and the bounded compatibility
schemata \eqref{axiom-compat} and \eqref{axiom-compat-pi} for nested controls. Each schema is displayed
as an equality between two circuit diagrams.}
\label{fig:axioms_QC_part2}
\end{figure}

Every displayed angle parameter in Figure~\ref{fig:axioms_QC_part2} ranges
over \(\mathbb R\).  The inequalities in that figure restrict labels: the
\(S\)-rules require disjoint active levels, and the \(B\)-rules require
distinct outer control values \(k\neq\ell\).

Before deriving the control algebra, we compare the schemata in
Figure~\ref{fig:axioms_QC_part2} with the global laws they generate.  If the
corresponding control-algebra law, or semantic support-sensitivity, were
already present in the syntax, the matching schemata would follow as
contextual instances.  The next subsection proves the direction needed for
completeness: the bounded schemata generate the three control-algebra laws.

\begin{lemma}\label{lem:commutativity-rules-admissible}
If commutativity from Definition~\ref{def:control-algebra} were admitted as an
admissible rewrite in \(\CQC\), then the schemata
\eqref{cp-cp}, \eqref{cp-cnot}, \eqref{cnot-cnot-diff}, \eqref{ch-ccnot},
\eqref{ch-ch}, \eqref{ch-ccp}, and \eqref{cnot-cnot} would be derivable.
\end{lemma}

\begin{lemma}\label{lem:support-sensitive-rules}
If \(\CQC\) were assumed to be support-sensitive, with the displayed generators
carrying their semantic supports, then the branch-commutation schemata listed
in Lemma~\ref{lem:commutativity-rules-admissible} would be derivable.  The
same assumption would also give \eqref{HHcomm}, \eqref{Hphasecomm}, and
\eqref{Hcnot}.
\end{lemma}

\begin{lemma}\label{lem:compatibility-rules-admissible}
If mixed-control compatibility from Definition~\ref{def:control-algebra} were
admitted as an admissible rewrite in \(\CQC\), then the schemata
\eqref{axiom-compat} and \eqref{axiom-compat-pi} would be derivable.
\end{lemma}

\begin{lemma}\label{lem:exhaustivity-rules-admissible}
If exhaustivity from Definition~\ref{def:control-algebra} were admitted as an
admissible rewrite in \(\CQC\), then the totalisation schemata
\eqref{axiom-total-phase} and \eqref{axiom-total-hadamard} would be derivable.
\end{lemma}
\begin{proof}[Proof of Lemmas~\ref{lem:commutativity-rules-admissible}--\ref{lem:exhaustivity-rules-admissible}]
The branch-commutation rules are instances of commutativity, up to tensoring
with identities and PROP coherence.  Support-sensitivity gives
the Hadamard commutations because the displayed supports are disjoint:
\(\{i,j\}\) from \(\{k,\ell\}\) or \(\{k\}\), and
\(\{i,j\}\times[d]\) from \(\{k\}\times\{m\}\).
Compatibility gives \eqref{axiom-compat} and \eqref{axiom-compat-pi} by
applying it to a scalar phase and to a once-controlled \(\pi\)-phase.
Finally, exhaustivity gives \eqref{axiom-total-phase} and
\eqref{axiom-total-hadamard} by taking, respectively, the scalar phase and the
derived two-level Hadamard as the controlled circuit.
\end{proof}

\begin{proposition}\label{prop:qudit-support-sensitive}
The semantic polycontrolled PROP \((\Qudit,(\Ctrl_k)_{k\in[d]})\) is
support-sensitive.
\end{proposition}
\begin{proof}
For \(U\in\Qudit(n,n)\), set
\(\mathrm{supp}(U):=\{x\in[d]^n\mid U\ket{x}\neq\ket{x}\}\). The support
inclusions for composition and tensor, and the formulas for symmetries and
controls, follow from their matrix semantics. If \(U\) and \(V\) have disjoint
supports \(S\) and \(T\), then they act on different direct summands in the
decomposition by \(S\), \(T\), and the remaining basis states, hence commute.
\end{proof}

\subsection{Derived control-algebra laws}\label{subsec:derived-control-algebra}

We now derive the three control-algebra laws
(Definition~\ref{def:control-algebra}) from the finite schemata in
Figures~\ref{fig:axioms_QC_part1}--\ref{fig:axioms_QC_part2}.

\begin{definition}\label{def:G-factors}
Let \(\mathcal G=\{\sigma_{1,1},\ \input{./figures/phase.tikz},\ \input{./figures/Hadiip.tikz},
\ \ctrl_k(\input{./figures/phase.tikz}),\ \ctrl_k(\ctrl_\ell(\input{./figures/piphase.tikz}))\}\),
with the adjacent Hadamard indexed by \(0\le r<d-1\) and the controls indexed
by \(k,\ell\in[d]\). A contextual \(\mathcal G\)-factor is a morphism
\(\id_p\otimes G\otimes\id_q\), where \(G\in\mathcal G\) and \(p,q\ge 0\).
\end{definition}

\begin{lemma}[Factor normalisation]\label{lemma-Greducible}
Every morphism of \(\CQC\) is provably equal in \(\QCeq\) to a possibly empty
finite sequential product of contextual \(\mathcal G\)-factors.
\end{lemma}
\begin{proof}
Apply \(\Separate{}\) from Appendix~\ref{app:separate-normalisation}.  It
expands the derived one-qudit notation, serialises tensor, and pushes controls
through composites.  When a controlled generator is not already a contextual
\(\mathcal G\)-factor, derived equations reduce it to a finite product with
fewer outer controls on each factor.  Lemma~\ref{lem:separate-terminates}
gives the termination measure; the controlled-generator reductions are the
derived rules \eqref{axiom-controlled-hadamard-decompose}, \eqref{axiom-ccp},
and \eqref{axiom-cccp} of Appendix~\ref{app:derivations}.  An identity map is
sent to the empty product at the same ambient arity, i.e.\ the PROP unit.
\end{proof}

\begin{theorem}\label{thm:compatibility}
In the equational theory \(\QCeq\) on \(\CQC\), the family
\((\ctrl_k)_{k=0}^{d-1}\) satisfies compatibility.
\end{theorem}

\begin{theorem}\label{thm:commuting}
In the equational theory \(\QCeq\) on \(\CQC\), the family
\((\ctrl_k)_{k=0}^{d-1}\) is commutative.
\end{theorem}

\begin{theorem}\label{thm:exhaustive}
In the equational theory \(\QCeq\) on \(\CQC\), the family
\((\ctrl_k)_{k=0}^{d-1}\) satisfies exhaustivity.
\end{theorem}
\begin{proof}[Proof of Theorems~\ref{thm:compatibility}--\ref{thm:exhaustive}]
Apply factor normalisation to the controlled composites that occur in the
three laws.  The resulting products consist of contextual \(\mathcal G\)-factors.
Plain symmetry factors are handled by structural naturality; the remaining
comparisons are the controlled-factor checks of Lemma~\ref{lem:finite-factor-laws}.

For compatibility, the two nested-control sides give matching separated
products after the outer control wires are swapped.  Factors meeting at most
one of those wires pass the swap by naturality; factors meeting both wires use
the nested-control check of Lemma~\ref{lem:finite-factor-laws}.

For commutativity, the separated products for \(\ctrl_k(f)\) and
\(\ctrl_\ell(g)\), with \(k\neq\ell\), are interchanged one neighbouring pair
at a time.  Each pair is a branch check from
Lemma~\ref{lem:finite-factor-laws}; induction on the two product lengths moves
all factors across.

For exhaustivity, the totalisation check gives the phase and
adjacent-Hadamard generators.  The identity case follows from the control
functor laws, and structural swaps use the fixed presentation
\eqref{swap-decomp}.  The compatibility and commutativity parts give closure
under composition, tensor product, and added controls.  Appendix~\ref{app:commuting}
contains the factor commutation table, and Appendix~\ref{app:exhaustivity}
contains the closure induction for exhaustivity.
\end{proof}

%% file: embedding_mfcs.tex
\section{Embedding Qudit Circuits into LOPP Circuits}\label{sec:embedding-lopp}
Single-photon linear optics gives a complete equational calculus for finite
unitary matrices generated by phases, beam splitters, and swaps
\cite{LOPP,LOPP-min}.  Completeness is proved by encoding qudit circuits into
that calculus and decoding the resulting optical derivations.  We place the
\(d^n\) computational basis states
of an \(n\)-qudit register in reflected Gray order, so adjacent optical
generators decode to basis-controlled phases or adjacent two-level qudit gates.
The tensor operations do not match. \(\CQC\) uses Hilbert-space tensor product,
while single-photon \(\LOPP\) uses block-diagonal direct sum.  The translations
are therefore contextual maps on a fixed \(d^n\)-mode space. They preserve
sequential composition, serialise tensor, and are tied together by an
encode-decode retraction in \(\QCeq\).
We prove the three required statements in that order. Encoding preserves
denotations after the Gray-basis change; decoded optical rewrites can be
mimicked in \(\QCeq\); decoding an encoding retracts to the original qudit
circuit.

\subsection{Linear optics and Gray-ordered semantics}\label{subsec:lopp-semantic-model}
The decoding map later inspects raw optical subterms, including the consecutive
mode interval a subcircuit occupies.  The quotient \(\LOPP\) is the language to
which the imported completeness theorem applies.
\begin{definition}\label{def:raw-lopp}
Let \(\RawLOPP\) be the free endomorphism PROP generated by the empty diagram,
identity wires, swap, one-mode phases \(\input{./figures/lov-phase.tikz}\), and two-mode
beam splitters \(\scalebox{0.7}{\input{./figures/beamsplitter.tikz}}\), with
\(\RawLOPP(m,n)=\varnothing\) for \(m\neq n\).  Let \(\LOPP\) be the quotient
of \(\RawLOPP\) by strict PROP coherence.
\end{definition}

\begin{definition}\label{def:lopp-semantics-sp}
The single-photon interpretation
\(\interp{-}_{\mathrm{sp}}:\RawLOPP(m,m)\to U(m)\) is defined by
\[
\begin{aligned}
\interp{C_2\circ C_1}_{\mathrm{sp}}
  &= \interp{C_2}_{\mathrm{sp}}\interp{C_1}_{\mathrm{sp}},&
\interp{C_1\otimes C_2}_{\mathrm{sp}}
  &= \interp{C_1}_{\mathrm{sp}}\oplus\interp{C_2}_{\mathrm{sp}},\\
\interp{\input{./figures/lov-phase.tikz}}_{\mathrm{sp}}
  &= (e^{i\theta}),&
\interp{\scalebox{0.7}{\input{./figures/beamsplitter.tikz}}}_{\mathrm{sp}}
  &=
  \begin{pmatrix}
  \cos\theta&i\sin\theta\\
  i\sin\theta&\cos\theta
  \end{pmatrix}.
\end{aligned}
\]
The structural generators are interpreted by
\(\interp{\gempty}_{\mathrm{sp}}=I_0\),
\(\interp{\gI}_{\mathrm{sp}}=I_1\), and
\(\interp{\gswap}_{\mathrm{sp}}=\begin{psmallmatrix}0&1\\1&0\end{psmallmatrix}\).
More generally, identities and swaps use the corresponding identity and
permutation matrices.  This interpretation respects strict PROP coherence and
therefore descends to \(\LOPP\).
\end{definition}

The non-structural equations of the calculus are listed in
Figure~\ref{fig:axiom_lopp} of Appendix~\ref{app:mimicking-rules}; we write
\(\mathrm{LOPP}\vdash C_1=C_2\) for the congruence generated by those equations
together with strict PROP coherence.
\begin{theorem}[\cite{LOPP,LOPP-min}]\label{thm:lopp-complete}
The equational theory \(\mathrm{LOPP}\) is sound and complete for the
single-photon semantics: for \(C_1,C_2:m\to m\),
\(\interp{C_1}_{\mathrm{sp}}=\interp{C_2}_{\mathrm{sp}}\) if and only if
\(\mathrm{LOPP}\vdash C_1=C_2\).
\end{theorem}

To compare \(d^n\)-mode optical circuits with \(n\)-qudit circuits, we identify
optical modes with computational basis states using a reflected \(d\)-ary Gray
code.
\begin{definition}\label{def:gray-code}
The reflected Gray order \(G_n^d:\{0,\ldots,d^n-1\}\to[d]^n\) is defined by
\(G_0^d(0)=\epsilon\) and, writing \(k=q d^{n-1}+r\), by
\(G_n^d(k)=q\cdot G_{n-1}^d(r)\) if \(q\) is even and
\(G_n^d(k)=q\cdot G_{n-1}^d(d^{n-1}-1-r)\) if \(q\) is odd.
\end{definition}
\begin{lemma}\label{lem:gray-neighbours}
Consecutive Gray words differ in exactly one digit, and that digit changes by
\(\pm1\). In particular, two adjacent optical modes always correspond to an
adjacent two-level qudit transition, possibly under controls on the remaining
digits.
\end{lemma}
\begin{proof}
Induct on \(n\), the case \(n=0\) being empty. Inside a leading block
\(q d^{n-1},\ldots,(q+1)d^{n-1}-1\), the suffix is listed either in
\(G_{n-1}^d\)-order or in reverse order, so the induction hypothesis gives a
single adjacent change in the suffix. At the boundary between the \(q\)- and
\((q+1)\)-blocks, reflection makes the terminal suffix of one block equal to
the initial suffix of the next; hence only the leading digit changes, from
\(q\) to \(q+1\). Thus every consecutive pair differs in one digit, by
\(\pm1\).
\end{proof}
For \(d=3,n=2\), the order is
\(00,01,02,12,11,10,20,21,22\).  The edges \(02\to12\) and \(12\to11\)
show the two cases: a fixed suffix or prefix guards the move, and odd blocks
reverse the orientation, while the changed digit stays adjacent.
Let \(\mathfrak G_{n,d}:\mathbb{C}^{d^n}\to(\mathbb{C}^{d})^{\otimes n}\) be
the permutation unitary determined by
\(\mathfrak G_{n,d}(\ket t)=\ket{G_n^d(t)}\). The native single-photon
semantics \(\interp{-}_{\mathrm{sp}}\) is the one used by
Theorem~\ref{thm:lopp-complete}; for comparison with qudit circuits we
conjugate it by this fixed Gray-code basis change.
\begin{definition}\label{def:gray-permutation}\label{def:lopp-semantics}
For \(C:d^n\to d^n\) in \(\LOPP\), define its Gray-ordered qudit semantics by
\(\interp{C}_{\mathrm{LOPP}}:=\mathfrak G_{n,d}\circ
\interp{C}_{\mathrm{sp}}\circ\mathfrak G_{n,d}^{-1}\).
\end{definition}

\subsection{Encoding and decoding}\label{subsec:encoding}\label{subsec:decoding}
Encoding and decoding are defined before quotienting and carry context
parameters: encoding records where a qudit circuit sits inside a larger
register, and decoding records where a raw optical subcircuit sits in the mode
list.  These raw choices must survive quotienting;
Lemmas~\ref{lem:encoding-respects-equiv}
and~\ref{lem:decoding-respects-equiv-maintext} give the required invariance.

Three optical gadget families are used by the encoding. The block permutation
\(\sigma^d_{a,n,b}\in\RawLOPP(d^{a+n+b},d^{a+n+b})\) reorders Gray-coded modes
so that, after decoding, it becomes the qudit wire permutation
\(\id_a\otimes\sigma_{n,b}\).  The lift
\(\Lift_d^p(C)\) copies a \(d^n\)-mode circuit \(C\) through \(p\) additional
Gray digits, mirroring the copy on odd reflected blocks.  The selected lift
\(\Lift_{d,j}^p(C)\) copies \(C\) only on the branch where the added
distinguished digit has value \(j\), with identities on other branches. Their raw
constructions are in Appendices~\ref{app:encoding-details}
and~\ref{app:swap-encoding}; their decoding interface is
Lemma~\ref{lem:gadget-decoding-interface}.

\begin{definition}\label{def:encoding}
For a raw qudit circuit \(C:n\to n\), the contextual encoding
\(\mathrm E^a_b(C)\in\RawLOPP(d^{a+n+b},d^{a+n+b})\) represents \(C\) on the
middle \(n\) qudits of an \((a+n+b)\)-qudit register. The recursive clauses are
\[
\begin{array}{rcl}
\mathrm E^a_b(C_2\circ C_1)&:=&
\mathrm E^a_b(C_2)\circ\mathrm E^a_b(C_1),\\
\mathrm E^a_b(C_1\otimes C_2)&:=&
\mathrm E^{a+n_1}_b(C_2)\circ\mathrm E^a_{b+n_2}(C_1)
\quad(C_i:n_i\to n_i),\\
\mathrm E^a_b(\id_n)&:=&\id_{d^{a+n+b}},\\
\mathrm E^a_b(\input{./figures/phase.tikz})&:=&
\bigl(\input{./figures/lov-phase.tikz}\bigr)^{\otimes d^{a+b}},\\
\mathrm E^a_b(\input{./figures/Hadiip.tikz})&:=&
\sigma^d_{a,b,1}\circ\Lift_d^{a+b}\!\bigl(B_d^{(i,i+1)}\bigr)\circ
\sigma^d_{a,1,b},\\
\mathrm E^a_b(\gswap)&:=&
\sigma^d_{a,b,2}\circ\sigma^d_{a+b,1,1}\circ\sigma^d_{a,2,b},\\
\mathrm E^a_b(\ctrl_j(C))&:=&
\sigma^d_{a,b,m+1}\circ \Lift_{d,j}^{a+b}(\mathrm E^0_0(C))\circ
\sigma^d_{a,m+1,b}\quad(C:m\to m).
\end{array}
\]
Here \(B_d^{(i,i+1)}\) denotes the \(d\)-mode optical network implementing the
adjacent Hadamard, fixed in Definition~\ref{def:hadamard-bs-gadget} of
Appendix~\ref{app:encoding-details}.
We write \(\mathrm E(C):=\mathrm E^0_0(C)\).
\end{definition}
For example, if \(g,h:1\to1\), then
\(\mathrm E^0_0(g\otimes h)=\mathrm E^1_0(h)\circ \mathrm E^0_1(g)\) on the
same \(d^2\) modes: tensor is serialised over context branches.

\begin{lemma}\label{lem:encoding-respects-equiv}
For all \(a,b\ge0\), if raw qudit circuits \(C,C':n\to n\) satisfy
\(C\equiv C'\) in the structural congruence defining \(\CQC\), then
\(\mathrm{LOPP}\vdash \mathrm E^a_b(C)=\mathrm E^a_b(C')\).
\end{lemma}
\begin{proof}
The structural congruence is generated by strict PROP coherence and by the
value-control functor laws.  The recursive clauses for \(\mathrm E^a_b\) send
the unit, associativity, interchange, and symmetry equations of strict PROP
coherence to the corresponding raw \(\LOPP\) coherence equations, with the
qudit symmetry implemented by the block permutations \(\sigma^d_{a,n,b}\).
For the control-functor equations, selected lifts implement the same projector
algebra on the distinguished Gray digit.  The clause-by-clause calculation is
Lemma~\ref{lem:encoding-respects-equiv-app} in
Appendix~\ref{app:encoding-decoding}.
\end{proof}

\begin{lemma}\label{lem:encoding-correctness}
The contextual encoding preserves semantics: for \(C:n\to n\) in \(\CQC\),
\(\interp{\mathrm E^a_b(C)}_{\mathrm{LOPP}}=
I_{d^a}\otimes \interp{C}\otimes I_{d^b}\).  In particular,
\(\interp{\mathrm E(C)}_{\mathrm{LOPP}}=\interp{C}\).
\end{lemma}
\begin{proof}
By Lemma~\ref{lem:encoding-respects-equiv}, choose any raw representative and
argue by induction on it.  The identity, composition, and phase clauses are
immediate from the definition of \(\mathrm E^a_b\).  The tensor clause
serialises \(C_1\otimes C_2\) as two actions on disjoint context branches, so
the Gray-ordered matrix is
\(I_{d^a}\otimes\interp{C_1}\otimes\interp{C_2}\otimes I_{d^b}\).
Ordinary lifts duplicate a local gate on every Gray branch; selected lifts keep
only the branch with the required control value; and block permutations move
the relevant Gray blocks to the qudit wire on which the generator acts.
Appendix~\ref{app:encoding-details} gives the lift constructions, and
Appendix~\ref{app:swap-encoding} gives the block-permutation construction.
\end{proof}

Decoding goes in the opposite direction, reading a raw \(\LOPP\) subcircuit in
a consecutive mode interval \(t,\ldots,t+\ell-1\) inside a \(d^n\)-mode system.
When \(G_n^d(t)=u\,v\,w\) and
\(G_n^d(t+1)=u\,(v+\varepsilon)\,w\), with
\(\varepsilon\in\{-1,+1\}\), write
\(\Lambda^u_w(g):=\sigma_{|u|,|w|,1}\circ\ctrl_{uw}(g)\circ
\sigma_{|u|,1,|w|}\)
for the one-qudit gate \(g\) placed on the changed digit and controlled by the
surrounding basis word.
\begin{definition}\label{def:decoding}
For \(0\le t\) and \(t+\ell\le d^n\), the contextual decoding
\(\mathrm D^t_n:\RawLOPP(\ell,\ell)\to\RawCQC_d(n,n)\) reads a raw optical
subcircuit as occupying global modes \(t,\ldots,t+\ell-1\). Its recursive
clauses are
\[
\begin{array}{rcl}
\mathrm D^t_n(C_2\circ C_1)&:=&
\mathrm D^t_n(C_2)\circ\mathrm D^t_n(C_1),\\
\mathrm D^t_n(C_1\otimes C_2)&:=&
\mathrm D^{t+\ell_1}_n(C_2)\circ \mathrm D^t_n(C_1)
\quad(C_1:\ell_1\to\ell_1),\\
\mathrm D^t_n(\gempty)&:=&\id_n,\qquad
\mathrm D^t_n(\gI):=\id_n,\\
\mathrm D^t_n(\input{./figures/lov-phase.tikz})&:=&
\ctrl_{G_n^d(t)}(\input{./figures/phase.tikz}),\\
\mathrm D^t_n(\gswap)&:=&
\Lambda^u_w(\input{./figures/Xvveps.tikz}),\\
\mathrm D^t_n(\scalebox{0.7}{\input{./figures/beamsplitter.tikz}})&:=&
\Lambda^u_w(\input{./figures/Rxvveps.tikz}).
\end{array}
\]
In the last two clauses, \(u,v,w,\varepsilon\) are determined by
Lemma~\ref{lem:gray-neighbours} for the consecutive modes \(t,t+1\).
The figures \(X^{(v,v+\varepsilon)}\) and
\(R_x^{(v,v+\varepsilon)}(\theta)\) use \(\varepsilon=+1\) for the adjacent
transition \(v\leftrightarrow v+1\) and \(\varepsilon=-1\) for
\(v\leftrightarrow v-1\); the displayed order records the direction of the
Gray edge.
We write \(\mathrm D(C):=\mathrm D^0_n(C)\) for a circuit on all \(d^n\) modes.
\end{definition}

For \(d=3,n=2\), mode \(4\) has Gray word \(11\), so its phase decodes to a
phase controlled on both qutrits being \(1\); beam splitters on modes \(2,3\)
and \(3,4\) decode to adjacent transitions along \(02\leftrightarrow12\) and
\(12\leftrightarrow11\), with the unchanged digit as control.

The lemmas below keep track of the interval \(t,\ldots,t+\ell-1\) and its Gray
orientation.  Tensor is read on consecutive intervals as
\(\mathrm D^{t+\ell_1}_n(C_2)\circ\mathrm D^t_n(C_1)\); odd Gray blocks reverse
the local order.  The remaining checks are local phases, swaps, and beam
splitters.

\begin{lemma}\label{lem:decoding-locality}
For every \(n,t\), the decoding of a one-mode phase is supported on the single
basis word \(G_n^d(t)\), and the decoding of a swap or beam splitter on
consecutive modes \(t,t+1\) is supported on the two Gray-neighbouring basis
words \(G_n^d(t)\) and \(G_n^d(t+1)\).
\end{lemma}
\begin{proof}
The one-mode phase clause of Definition~\ref{def:decoding} is
\(\ctrl_{G_n^d(t)}(\input{./figures/phase.tikz})\), so its support is exactly the branch word
\(G_n^d(t)\). For swaps and beam splitters, Lemma~\ref{lem:gray-neighbours}
writes the adjacent Gray words as \(u\,v\,w\) and
\(u\,(v+\varepsilon)\,w\). The decoding clause is then
\(\Lambda^u_w(g)=\sigma_{|u|,|w|,1}\circ\ctrl_{uw}(g)\circ
\sigma_{|u|,1,|w|}\), where \(g\) is the adjacent \(X\)- or \(R_x\)-gate on
the changed digit. Hence the only non-identity action is on those two basis
words.
\end{proof}
\begin{lemma}\label{lem:decoding-respects-equiv-maintext}
Let \(P[-]\) be a raw \(\LOPP\) context whose hole occupies a consecutive
mode interval \(t,\ldots,t+\ell-1\) inside a \(d^n\)-mode circuit. If
\(C,C':\ell\to\ell\) are equal by strict PROP coherence in \(\RawLOPP\), then
\(\QCeq\vdash \mathrm D^0_n(P[C])=\mathrm D^0_n(P[C'])\).
In particular, \(\QCeq\vdash\mathrm D(C)=\mathrm D(C')\) for coherent
\(d^n\)-mode circuits \(C,C'\).
\end{lemma}
\begin{proof}
The proof is an induction on the strict-PROP coherence step inside the fixed
raw context.  Units and associativity follow directly from the recursive
clauses for \(\mathrm D\).  For interchange, the two optical factors occupy
disjoint consecutive intervals; by Lemma~\ref{lem:decoding-locality}, their
decodings act on disjoint Gray supports, so Theorem~\ref{thm:commuting}
commutes the two decoded factors.  The symmetry equations reduce to the
decoding of adjacent optical swaps, which gives the controlled adjacent
transposition of the two corresponding Gray labels.  Lemma~\ref{decodingtoporules}
in Appendix~\ref{app:mimicking-rules} gives the induction over raw contexts.
\end{proof}

\begin{lemma}\label{lem:gadget-decoding-interface}
The optical gadgets used in the encoding decode as follows.  For every
\(n,p\ge0\), every raw \(C:d^n\to d^n\) in \(\LOPP\), and every \(j\in[d]\),
\(\QCeq\vdash\mathrm D^0_{n+p}(\Lift_d^p(C))=
\id_p\otimes\mathrm D^0_n(C)\) and
\(\QCeq\vdash\mathrm D^0_{n+1+p}(\Lift_{d,j}^p(C))=
\id_p\otimes\ctrl_j(\mathrm D^0_n(C))\).  For all \(a,n,b\ge0\),
\(\QCeq\vdash\mathrm D^0_{a+n+b}(\sigma^d_{a,n,b})=\id_a\otimes\sigma_{n,b}\).
\end{lemma}
\begin{proof}
For \(\Lift_d^p(C)\), the \(d^p\) Gray blocks contain copies of \(C\), with
the local order reflected on odd blocks.  Decoding one full block gives
\(\mathrm D^0_n(C)\) guarded by the corresponding \(p\)-digit Gray word; the
reflected-block calculation gives the same decoded circuit.  Distinct-branch
commutation reorders these guarded copies, and exhaustivity removes the full
family of guards, giving \(\id_p\otimes\mathrm D^0_n(C)\).  The selected lift
keeps only the blocks with outer digit \(j\).  The detailed derivations are
Lemmas~\ref{lem:dec-lift} and~\ref{lem:dec-selected-lift} in
Appendix~\ref{app:encoding-decoding}; the block-permutation equation is
Lemma~\ref{lem:dec-swap} in Appendix~\ref{app:swap-encoding}.
\end{proof}

\begin{theorem}\label{thm:mimicking-rules}
If \(\mathrm{LOPP}\vdash C_1=C_2\) for \(C_1,C_2:d^n\to d^n\), then
\(\QCeq\vdash \mathrm D(C_1)=\mathrm D(C_2)\).
\end{theorem}
\begin{proof}
Induct on the \(\mathrm{LOPP}\)-derivation. Strict PROP steps use
Lemma~\ref{lem:decoding-respects-equiv-maintext}.  The non-structural axioms
decode to the matching controlled phase, swap, and Euler rules of \(\QCeq\).
For the three-beam-splitter axiom, same-digit Gray moves give \eqref{3Rx},
while two-coordinate corners use its mixed controlled form \eqref{3CRx}.
Appendix~\ref{app:mimicking-rules} gives the axiom-by-axiom derivations.
\end{proof}

\subsection{Completeness via transfer}
The transfer uses the raw context kept by \(E^a_b\) and \(D^t_n\).  The
encoding parameters \(a,b\) place the qudit circuit inside a larger Gray
register; the decoding parameter \(t\) records the optical interval being read
inside \(d^n\) modes.  These data are needed before quotienting, since optical
tensor is read as consecutive mode intervals.

\begin{theorem}\label{thm:encoding-decoding}
For all \(n,a,b\ge0\) and all \(C:n\to n\) in \(\CQC\),
\(\QCeq\vdash
\mathrm D^0_{a+n+b}(\mathrm E^a_b(C))=\id_a\otimes C\otimes\id_b\).
\end{theorem}
This internal retraction is the reason the non-monoidal translation is usable:
in a fixed surrounding register, encoding and then decoding returns the original
circuit with only identity padding.
\begin{proof}
Choose a raw representative by Lemma~\ref{lem:encoding-respects-equiv} and
argue by structural induction.  Appendix~\ref{app:encoding-decoding} gives the
case analysis.  Empty diagrams and identities are immediate, scalar phases
collapse by exhaustivity, adjacent Hadamards use Lemma~\ref{lem:gadget-decoding-interface},
and wire symmetries use the block-permutation equation.  Composition and tensor
follow from the recursive clauses. For \(\ctrl_j(C')\), decoding the selected
lift gives \(\id_a\otimes\ctrl_j(\mathrm D(\mathrm E(C')))\otimes\id_b\), which
the induction hypothesis reduces to \(\id_a\otimes\ctrl_j(C')\otimes\id_b\).
\end{proof}
\begin{corollary}\label{PROP:encoding-decoding}
For every \(C:n\to n\) in \(\CQC\), \(\QCeq\vdash \mathrm D(\mathrm E(C))=C\).
\end{corollary}
\begin{proof}
Take \(a=b=0\) in Theorem~\ref{thm:encoding-decoding}; the identity padding is
then the monoidal unit on both sides.
\end{proof}

Completeness is now the encode-mimic-decode round trip.

\begin{theorem}\label{thm:main-completeness}
For every \(d\ge2\), \(\QCeq\) is sound and complete for exact unitary qudit
semantics: for all \(C_1,C_2:n\to n\), \(\interp{C_1}=\interp{C_2}\) if and
only if \(\QCeq\vdash C_1=C_2\).
\end{theorem}
\begin{proof}
Soundness is checked axiom by axiom against the displayed unitary semantics
in Appendix~\ref{app:axiom-soundness}:
phase arithmetic, \(2\pi\)-periodicity, and the Euler/rotation rules are the
corresponding one- or two-level matrix identities.  For the support and
branch-commutation cells, the controlled summands are selected by disjoint
basis projectors; compatibility permutes the projectors on the two control
wires; totalisation sums the \(d\) branch projectors to the identity.  If
\(\interp{C_1}=\interp{C_2}\), then
Lemma~\ref{lem:encoding-correctness} gives equal single-photon semantics for the
encodings after Gray conjugation, so
\(\mathrm{LOPP}\vdash \mathrm E(C_1)=\mathrm E(C_2)\).  Theorem~\ref{thm:mimicking-rules}
and Corollary~\ref{PROP:encoding-decoding} yield
\(\QCeq\vdash C_1=\mathrm D(\mathrm E(C_1))=\mathrm D(\mathrm E(C_2))=C_2\).
\end{proof}

%% file: conclusion.tex
\section{Conclusion}

We have given a finite schematic, bounded-arity equational theory for exact unitary qudit
circuits, uniform in the dimension and including global phases.  Primitive
value-control keeps the schemata within three wires, while compatibility,
commutativity, and exhaustivity are derived internally.  The
Gray-ordered \(\LOPP\) transfer preserves circuits through mimicking and
encode-decode retraction.  Reflected Gray order, selected lifts, and derived
control algebra keep optical locality compatible with native tensor structure.

The displayed diagrams have dimension-independent arity.  The dependence on
\(d\) is carried by level labels, product-frame ranges, and the reflected Gray
order used by the transfer.  For higher-dimensional compilation, the calculus
gives direct equations for value controls and adjacent two-level operations.
Natural next steps include orienting rule fragments as rewrite systems and
specialising the presentation to fault-tolerant or hardware-native qudit gate
sets.

%% file: appendix_proof_guide.tex
\section{Proof dependency guide}\label{app:proof-dependency-guide}

The proof dependencies are acyclic in the following order.  This appendix is
only a reference point for checking where each later argument gets its inputs.

The raw circuit syntax and the structural control-functor laws are fixed in
Section~\ref{sec:qudit-circuits}.  No compatibility, commutativity, or
exhaustivity principle for distinct controls is assumed there.

The local schemata of
Figures~\ref{fig:axioms_QC_part1}--\ref{fig:axioms_QC_part2} generate
\(\QCeq\).  Appendix~\ref{app:derived-notation} expands the derived one-qudit
notation, Appendix~\ref{appendix:relations_angles} fixes the angle conventions
for the rotation schemata, and Appendix~\ref{app:axiom-soundness} records the
semantic soundness checks for the displayed axiom schemata.

Appendix~\ref{app:derivations} proves the first support catalogue, including
the factor-expansion and generator-commutation rules used by normalisation and
commutativity.  Appendix~\ref{app:separate-normalisation} then proves the
factor normalisation used in Lemma~\ref{lemma-Greducible} and
Corollary~\ref{corollary-Greducible}, reducing arbitrary controlled circuits to
the finite factor family \(\mathcal G\) named in
Section~\ref{subsec:derived-control-algebra}.

The control algebra is derived from those normalised factors:
Appendix~\ref{app:compatibility} proves compatibility,
Appendix~\ref{app:commuting} proves commutativity of differently controlled
circuits, and Appendix~\ref{app:exhaustivity} proves exhaustivity using those
two principles together with the totalisation support rules.

Once the control algebra is available, Appendix~\ref{app:derivations-p2} proves
the post-control-algebra identities, including the mixed controlled
three-rotation rule used in the LOPP mimicking step.

Appendices~\ref{app:encoding-details}, \ref{app:swap-encoding},
and~\ref{app:encoding-decoding} establish the Gray-code interface for encoding
and decoding.  Appendix~\ref{app:swap-encoding} proves the block-permutation
decoding interface used both by the retraction proof and by structural
congruence in the mimicking proof.  At this point the derived control algebra
is available, so products over Gray branches can be reordered and collapsed by
commutativity and exhaustivity.

Appendix~\ref{app:mimicking-rules} checks that every LOPP axiom is preserved by
decoding.  The mixed three-mode optical case uses the mixed controlled
three-rotation rule derived in Appendix~\ref{app:derivations-p2}.

Section~\ref{sec:embedding-lopp} combines the pieces: encoding preserves
semantics, LOPP completeness gives an optical derivation,
Theorem~\ref{thm:mimicking-rules} decodes it into \(\QCeq\), and
Theorem~\ref{thm:encoding-decoding} removes the encode-decode round trip.

%% file: appendix_derived_gates.tex
\section{Derived notation for the axiom schemata}\label{app:derived-notation}

The axiom figures use several derived one-qudit circuits as notation.  The
definitions below pin down those abbreviations inside \(\CQC\); no generator is
added.

\begin{definition}\label{def:derived-adjacent-swap}
Write \(H^{(r,r+1)}\) for the adjacent Hadamard generator
\(\input{./figures/Hadiip.tikz}\).  For each adjacent pair \((r,r+1)\), the adjacent swap is
\(X^{(r,r+1)}
:= H^{(r,r+1)} \circ \ctrl_{r+1}(\input{./figures/piphase.tikz}) \circ
H^{(r,r+1)}\).
\end{definition}

\begin{definition}\label{def:derived-transpositions}
Set \(X^{(i,i)}:=\id_1\).  For the adjacent-walk conjugator, set
\(S_{r,r}:=\id_1\) and, if \(i<j\),
\(S_{i,j}:=\prod_{k=0}^{j-i-1} X^{(i+k,i+k+1)}\), where
\(\prod_{k=0}^{m} f_k := f_m\circ\cdots\circ f_0\), and define
\(X^{(i,j)}:= S_{i+1,j}\circ X^{(i,i+1)}\circ S_{i+1,j}^{-1}\).
For \(i>j\), set \(X^{(i,j)}:=X^{(j,i)}\).
\end{definition}

\begin{definition}\label{def:derived-two-level-gates}
The derived Hadamards are obtained by sliding the relevant levels into an
adjacent position.  Set \(H^{(i,i)}:=\id_1\); if \(i<j\), set
\(H^{(i,j)}:=X^{(j,i+1)} \circ H^{(i,i+1)} \circ X^{(j,i+1)}\); and if
\(i>j\), set \(H^{(i,j)}:=X^{(j,i)} \circ H^{(j,i)} \circ X^{(j,i)}\).
For distinct \(i,j\in[d]\), define
\(R_x^{(i,j)}(\theta):=H^{(i,j)}\circ \ctrl_i(\input{./figures/phase.tikz})\circ
\ctrl_j(\input{./figures/negphase.tikz})\circ H^{(i,j)}\).
\end{definition}

Under \(\interp{-}\), \(X^{(i,j)}\) swaps \(\ket{i}\) and \(\ket{j}\),
\(H^{(i,j)}\) acts as the usual Hadamard on
\(\mathrm{span}\{\ket{i},\ket{j}\}\), and \(R_x^{(i,j)}(\theta)\) acts on
that same two-dimensional subspace as
\(\left(\begin{smallmatrix}\cos\theta&i\sin\theta\\
i\sin\theta&\cos\theta\end{smallmatrix}\right)\), while all three act as the
identity on the orthogonal complement.

\begin{definition}\label{def:indexed-boxes}
The product frames used in the axiom figures are also meta-notation.  A frame
labelled \(\vec k:\varphi(\vec k)\), where
\(\vec k=(k_1,\ldots,k_r)\), denotes the sequential composite of the enclosed
circuit over all tuples \(\vec k\in[d]^r\) satisfying \(\varphi\), in
lexicographic increasing order.  Equivalently, \(k_1\) is the outer loop
variable and \(k_r\) is the inner loop variable.  For a single variable this is
just the increasing order of that index; if the predicate is omitted, the
default range is \(0\le k<d\).
\end{definition}

%% file: appendix_anglesrelations.tex
\section{Angle relations for rotation axioms}\label{appendix:relations_angles}

The angle-bearing schemata \eqref{eulerH}, \eqref{3Rx}, and the derived mixed
rule \eqref{3CRx} need
explicit right-hand-side parameter relations.  This appendix fixes those
\emph{meta-level parameter relations}: the left-hand side carries freely chosen
real parameters, while the right-hand side is chosen from the stated admissible
solutions.  Away from the usual Euler degeneracies the solution is deterministic;
at degenerate points we record a canonical choice, but the axiom schema may be
instantiated with any admissible degenerate representative.

In Delorme's one-qubit Euler parametrisations, the decomposition of a
two-level unitary naturally comes with a block scalar factor
(\emph{global phase} in \(U(2)\)).
In the present paper we apply such decompositions \emph{inside a fixed
two-dimensional subspace} \(\mathrm{span}\{\ket{i},\ket{j}\}\) of a \(d\)-level system.
Viewed in the full \(d\)-dimensional space, a phase factor that is ``global''
for the \(2\times 2\) block multiplies both \(\ket{i}\) and \(\ket{j}\) by the
same phase while leaving the other basis levels unchanged.
In our syntax this is realised by applying the same level-phase on both levels,
i.e.\ by \(\ctrl_i(\alpha)\circ \ctrl_j(\alpha)\).
On the \((i,j)\)-subspace this acts as \(e^{\mathrm{i}\alpha}I_2\), hence it
\emph{commutes} with every gate supported on \((i,j)\) (in particular with
\(H^{(i,j)}\) and with all level phases).
For this reason, whenever the qubit formulas isolate a scalar phase, we are free
to slide the corresponding two-level phase and absorb it into neighbouring level
phases. This accounts for the difference between our explicit angle formulas and
the qubit ones; semantically they implement the same
underlying \(U(2)\) parametrisation on \(\mathrm{span}\{\ket{i},\ket{j}\}\).

Throughout, equalities between angles are understood modulo \(2\pi\).
We write \(\arg(z)\) for the principal argument of a non-zero complex number
\(z\in\mathbb{C}\), taking values in \((-\pi,\pi]\).
When we restrict an angle to a range (e.g.\ \(\beta\in[0,2\pi)\)), we always mean
its canonical representative in that interval.
We also use the standard two-argument arctangent \(\operatorname{atan2}(y,x)\)
with values in \((-\pi,\pi]\).

\subsection*{Euler-type relation \eqref{eulerH}}

The axiom \eqref{eulerH} is a two-level Euler-type identity (applied in the paper on a
subspace \(\mathrm{span}\{\ket{i},\ket{j}\}\) of a qudit).
Syntactically, the left-hand side is parameterised by two real angles
\(\alpha_0,\alpha_2\), while the right-hand side uses four angles
\(\beta_0,\beta_1,\beta_2,\beta_3\).  For every choice of
\((\alpha_0,\alpha_2)\), an admissible choice of
\((\beta_0,\beta_1,\beta_2,\beta_3)\) in the prescribed ranges makes \eqref{eulerH}
sound on the relevant two-dimensional subspace, and hence as a qudit equation
acting as the identity outside that subspace.  The formulas below choose a
canonical deterministic representative; in degenerate cases the admissible
family has the usual Euler one-parameter freedom.

We follow Delorme's explicit Euler parametrisation (stated for qubits) and use
the same extraction function, with the ``two-level global phase'' convention
explained above.

Given \(\alpha_0,\alpha_2\in\mathbb{R}\), define
\(z=-\sin\bigl((\alpha_0+\alpha_2)/2\bigr)
  +i\,\cos\bigl((\alpha_0-\alpha_2)/2\bigr)\) and
\(z'=\cos\bigl((\alpha_0+\alpha_2)/2\bigr)
  -i\,\sin\bigl((\alpha_0-\alpha_2)/2\bigr)\).

Define \(\beta_0,\beta_1,\beta_2,\beta_3\in[0,2\pi)\) by cases.
The degenerate cases \(z=0\) or \(z'=0\) are treated separately to avoid dividing
by~\(0\); the generic case uses the ratio \(z/z'\).

\begin{itemize}
  \item If \(z' = 0\), set
  \(\beta_0=2\,\arg(z)\),
  \(\beta_1=(\pi+\alpha_0+\alpha_2)/2-\arg(z)\),
  \(\beta_2=(\pi+\alpha_0+\alpha_2)/2-\arg(z)\), and
  \(\beta_3=0\).

  \item If \(z = 0\), set
  \(\beta_0=2\,\arg(z')\),
  \(\beta_1=(\alpha_0+\alpha_2)/2-\arg(z')\),
  \(\beta_2=\pi+(\alpha_0+\alpha_2)/2-\arg(z')\), and
  \(\beta_3=0\).

  \item Otherwise (when \(z \neq 0\) and \(z' \neq 0\)), set
  \(\beta_0=\arg(z)+\arg(z')\),
  \(\beta_1=-\arg\bigl(i+\lvert z/z'\rvert\bigr)
    +(\pi+\alpha_0+\alpha_2)/2-\arg(z)\),
  \(\beta_2=\arg\bigl(i+\lvert z/z'\rvert\bigr)
    +(\pi+\alpha_0+\alpha_2)/2-\arg(z)\), and
  \(\beta_3=\arg(z)-\arg(z')\).
\end{itemize}

Soundness is the property used later; in the degenerate cases one endpoint
phase remains free.
\begin{lemma}
\label{lem:EH-angle-soundness}
The canonical choices above make Equation~\eqref{eulerH} sound on the relevant
two-dimensional subspace.  If \(z=0\) or \(z'=0\), then
\(\beta_3\) may be prescribed arbitrarily and the other three angles may be
chosen so that the same two-level matrix is obtained.
\end{lemma}
\begin{proof}
Write
\[
H=\frac{1}{\sqrt2}\begin{pmatrix}1&1\\1&-1\end{pmatrix},
\qquad
P_j(\theta)=\begin{pmatrix}1&0\\0&e^{i\theta}\end{pmatrix},
\qquad
P_i(\theta)=\begin{pmatrix}e^{i\theta}&0\\0&1\end{pmatrix}.
\]
With the left-to-right convention of the diagrams, the left-hand side of
\eqref{eulerH} is
\[
H P_j(\alpha_2) H P_j(\alpha_0) H
=
\frac{e^{iS}}{\sqrt2}
\begin{pmatrix}
i\overline z & z'\\
\overline {z'} & iz
\end{pmatrix},
\qquad S=\frac{\alpha_0+\alpha_2}{2}.
\]
The right-hand side is
\[
P_j(\beta_3)H P_j(\beta_2)P_i(\beta_1)H P_j(\beta_0)
=
\begin{pmatrix}
A & B e^{i\beta_0}\\
B e^{i\beta_3} & A e^{i(\beta_0+\beta_3)}
\end{pmatrix},
\]
where \(A=(e^{i\beta_1}+e^{i\beta_2})/2\) and
\(B=(e^{i\beta_1}-e^{i\beta_2})/2\).

Assume first that \(z\ne0\) and \(z'\ne0\).  Put
\(\rho=|z/z'|\) and \(t=\arg(i+\rho)\).  Since
\(|z|^2+|z'|^2=2\), the generic formula gives
\(A=\frac{i e^{iS}\overline z}{\sqrt2}\) and
\(B=\frac{e^{iS}|z'|e^{-i\arg z}}{\sqrt2}\).
Together with \(\beta_0=\arg z+\arg z'\) and
\(\beta_3=\arg z-\arg z'\), the four entries of the right-hand side become
\[
\frac{e^{iS}}{\sqrt2}
\begin{pmatrix}
i\overline z & z'\\
\overline {z'} & iz
\end{pmatrix},
\]
which is the left-hand side.

If \(z'=0\), then \(|z|=\sqrt2\).  The canonical choice has \(B=0\),
\(A=e^{i((\pi+\alpha_0+\alpha_2)/2-\arg z)}
  =\frac{i e^{iS}\overline z}{\sqrt2}\), and
\(e^{i\beta_0}=e^{2i\arg z}\),
so the same diagonal matrix is obtained.  If \(z=0\), then \(|z'|=\sqrt2\).
The canonical choice has \(A=0\),
\(B=e^{i(S-\arg z')}\), and \(e^{i\beta_0}=e^{2i\arg z'}\),
and again the two off-diagonal entries agree with the displayed left-hand-side
matrix.

Finally let \(\tau\) be any prescribed value for \(\beta_3\), understood modulo
\(2\pi\).  When \(z'=0\), keep
\(\beta_1=\beta_2=(\pi+\alpha_0+\alpha_2)/2-\arg z\) and set
\(\beta_0=2\arg z-\tau\).  When \(z=0\), set
\(\beta_1=S-\arg z'-\tau\),
\(\beta_2=\beta_1+\pi\), and \(\beta_0=2\arg z'+\tau\).
After taking representatives in \([0,2\pi)\), these choices give the same
matrix as the canonical representative and have \(\beta_3=\tau\).
\end{proof}

Since every occurrence of \eqref{eulerH} in the main development is applied to a two-level
gate supported on \(\mathrm{span}\{\ket{i},\ket{j}\}\) and acts as the identity on
the other computational levels, Lemma~\ref{lem:EH-angle-soundness} applies
systematically for all \(d\ge 2\).  The canonical formulas above are kept as a
deterministic rewrite convention, while derivations may use any admissible
degenerate representative.

\subsection*{Three-rotation identities \eqref{3Rx} and derived \eqref{3CRx}}

The axiom \eqref{3Rx} and the derived mixed rule \eqref{3CRx} express the standard fact that
the same real \(3\times 3\) rotation (hence an element of \(\mathrm{SO}(3)\))
admits both a \emph{ZXZ} and an \emph{XZX} Euler decomposition.
They are the qudit/circuit-side reflections of the ``(E3)'' principle used in
the linear-optical setting: \eqref{3Rx} is the uncontrolled identity, and \eqref{3CRx} is
its mixed controlled consequence. The parameter relations are identical in both
cases.

Let \(R_x(\theta)\) and \(R_z(\theta)\) be the standard real rotation
matrices about the \(x\)- and \(z\)-axes.  The left-hand side angles
\(\gamma_1,\gamma_2,\gamma_3\) are the freely chosen parameters of~\eqref{3Rx}.
Write \(c_r:=\cos(\gamma_r)\) and \(s_r:=\sin(\gamma_r)\).  Define
\[
  R_{E3}:=R_z(\gamma_1)\,R_x(\gamma_2)\,R_z(\gamma_3)
  =
  \begin{pmatrix}
    c_1c_3-s_1c_2s_3 & -c_1s_3-s_1c_2c_3 & s_1s_2\\
    s_1c_3+c_1c_2s_3 & -s_1s_3+c_1c_2c_3 & -c_1s_2\\
    s_2s_3            & s_2c_3             & c_2
  \end{pmatrix},
\]
and write \(R_{E3}=(r_{\mu,\nu})_{\mu,\nu\in\{1,2,3\}}\).
The right-hand side angles \(\delta_1,\delta_2,\delta_3\) are the deterministic
XZX Euler angles extracted from this same matrix, so
\(R_{E3}=R_x(\delta_1)\,R_z(\delta_2)\,R_x(\delta_3)\).

The extraction formulas below are exactly those used in the linear-optical
literature, following standard Euler-angle extraction conventions
\cite{EulerAngles}; they are stated with explicit ``generic'' and
``gimbal-lock'' cases, since Euler angles are not unique in degenerate
configurations.

Assume first that we are not in a gimbal-lock configuration for XZX, i.e.\
\(-1 < r_{1,1} < 1\). Then set
\(\delta_2=\arccos(r_{1,1})\),
\(\delta_1=\operatorname{atan2}(r_{3,1},\,r_{2,1})\), and
\(\delta_3=\operatorname{atan2}(r_{1,3},\,-r_{1,2})\).
Equivalently, in terms of the freely chosen \(\gamma\)-angles:
\[
\begin{aligned}
  \delta_2
    &= \arccos(c_1c_3-s_1c_2s_3),\\
  \delta_1
    &= \operatorname{atan2}(s_2s_3,\,
        s_1c_3+c_1c_2s_3),\\
  \delta_3
    &= \operatorname{atan2}(s_1s_2,\,
        c_1s_3+s_1c_2c_3).
\end{aligned}
\]

The remaining cases correspond to gimbal lock for XZX, i.e.\
\(\delta_2\in\{0,\pi\}\). To keep the schema functional, we fix the canonical
representative with last \(X\)-rotation equal to \(0\):

\begin{itemize}
  \item If \(r_{1,1}=1\), set
  \(\delta_2=0\), \(\delta_3=0\), and
  \(\delta_1=\operatorname{atan2}(r_{3,2},\,r_{3,3})\).

  \item If \(r_{1,1}=-1\), set
  \(\delta_2=\pi\), \(\delta_3=0\), and
  \(\delta_1=\operatorname{atan2}(-r_{3,2},\,r_{3,3})\).
\end{itemize}

These choices are one representative among the usual one-parameter family of
gimbal-lock solutions, but they give a deterministic right-hand side for the
schematic rule.

The extraction is sound in the following sense.
\begin{lemma}
\label{lem:3Rx-angle-soundness}
With the choices of \(\delta_1,\delta_2,\delta_3\) above, the identity
\[
R_z(\gamma_1)R_x(\gamma_2)R_z(\gamma_3)
=R_x(\delta_1)R_z(\delta_2)R_x(\delta_3)
\]
holds.
Consequently the same parameter relation is sound for the derived mixed rule
\eqref{3CRx}, whose proof reduces to the uncontrolled three-rotation pattern.
\end{lemma}
\begin{proof}
For arbitrary \(a,b,c\), direct multiplication gives
\[
{\setlength{\arraycolsep}{0pt}
R_x(a)R_z(b)R_x(c)=
\begin{pmatrix}
\cos b & -\sin b\cos c & \sin b\sin c\\
\sin b\cos a & -\sin a\sin c+\cos a\cos b\cos c
  & -\sin a\cos c-\sin c\cos a\cos b\\
\sin a\sin b & \sin a\cos b\cos c+\sin c\cos a
  & -\sin a\sin c\cos b+\cos a\cos c
\end{pmatrix}.}
\]
In the generic case \(-1<r_{1,1}<1\), we have
\(\sin(\delta_2)>0\).  The displayed matrix then shows that
\(\cos(\delta_2)=r_{1,1}\),
\((\cos\delta_1,\sin\delta_1)
  =\frac{(r_{2,1},r_{3,1})}{\sin\delta_2}\), and
\((\cos\delta_3,\sin\delta_3)
  =\frac{(-r_{1,2},r_{1,3})}{\sin\delta_2}\),
which is exactly the stated \(\arccos/\operatorname{atan2}\) extraction.
Substituting these values into the matrix for
\(R_x(\delta_1)R_z(\delta_2)R_x(\delta_3)\) recovers all entries of
\(R_{E3}\).

If \(r_{1,1}=1\), then the first row and column force the matrix to be a pure
\(x\)-rotation on the last two coordinates, and
\(\operatorname{atan2}(r_{3,2},r_{3,3})\) is its angle.  Thus the canonical
choice \(\delta_2=\delta_3=0\) is sound.  If \(r_{1,1}=-1\), then the same
displayed matrix with \(b=\pi\) shows that the free invariant is
\(\delta_1-\delta_3\), and the canonical choice
\(\delta_2=\pi,\delta_3=0,\delta_1=\operatorname{atan2}(-r_{3,2},r_{3,3})\)
is sound.

The rule \eqref{3CRx} uses the same real rotation identity after the surrounding
controlled \(X\)-moves expose the uncontrolled three-rotation subdiagram in
Appendix~\ref{app:derivations-p2}; therefore no additional angle relation is
needed for the controlled version.
\end{proof}

%% file: appendix_derivations.tex
\section{Support diagrammatic derivations}\label{app:derivations}

All derivations in this appendix take place in the equational theory
\(\mathrm{QC}_{d}\).  We write \(\mathrm{QC}_{d} \vdash C_1 = C_2\) when the
circuit identity \(C_1 = C_2\) is derivable from the axioms of
Figures~\ref{fig:axioms_QC_part1}--\ref{fig:axioms_QC_part2}.
Unlabelled equalities in the displayed diagram chains are definitional
unfoldings of the derived notation or strict PROP coherence; labelled
equalities cite the axiom or derived rule doing the substantive rewrite.

\subsection{Catalogue of derived rules}

The catalogue is split into three batches. Figure~\ref{fig:derived_QC_part1}
gathers the elementary support identities used for scalar phases,
involutions, totalisations, controlled transpositions, and the basic
\(X\)-interaction rules. Figure~\ref{fig:derived_QC_part2} adds the local
Hadamard--phase conjugation rules together with the decomposition identities
for controlled \(H\), controlled phase, and higher-controlled phase gates.
Figure~\ref{fig:derived_QC_part3} records the remaining commutation laws
between differently value-controlled phases, controlled Hadamards, and
controlled transpositions. Subsection~\ref{subsec:derived-rules-proofs}
proves these rules in the same order.

\begin{figure*}[!ht]
\centering
\fbox{
\begin{minipage}{0.975\textwidth}\centering
\begin{minipage}[t]{0.4\textwidth}\centering
\schemarule{\textsc{P0}}{0phase}{
\input{./figures/0phase.tikz} = \input{./figures/empty_diagram.tikz}
}
\schemarule{\textsc{\ensuremath{\pi}S}}{piphase}{
\input{./figures/negpiphase.tikz} = \input{./figures/piphase.tikz}
}
\schemarule{\textsc{XAdj}}{Xiplusun-Xiplusun}{
\input{./figures/Xiiplusun-Xiiplusun.tikz} = \input{./figures/line.tikz}
}
\schemarule{\textsc{XInv}}{Xij-Xij}{
\input{./figures/Xij-Xij.tikz} = \input{./figures/line.tikz}
}
\schemarule{\textsc{HInv}}{Hij-Hij}{
\input{./figures/Hij-Hij.tikz} = \input{./figures/line.tikz}
}
\end{minipage}\hfill
\begin{minipage}[t]{0.48\textwidth}\centering
\schemarule{\textsc{HTot}}{Htotal}{
\input{./figures/Htotal.tikz} = \input{./figures/totalH.tikz}
}
\schemarule{\textsc{RTot}}{Rxtotal}{
\input{./figures/Rxtotal.tikz} = \input{./figures/totalRx.tikz}
}
\end{minipage}
\begin{minipage}[t]{0.48\textwidth}\centering
\schemarule{\textsc{HDec}}{hdec}{
\input{./figures/Hadij.tikz} = \input{./figures/hadamard-decomposition.tikz}
}
\end{minipage}
\begin{minipage}[t]{0.48\textwidth}\centering
\schemarule{\textsc{XCtl}}{XCtrl}{
\input{./figures/CtrlX.tikz} = \input{./figures/XCtrl.tikz}
}
\end{minipage}
\begin{minipage}[t]{0.4\textwidth}\centering
\schemarule{\textsc{XCS}}{XCtrlbis}{
\input{./figures/XCtrlbis.tikz} = \input{./figures/CtrlXbis.tikz}
}
\end{minipage}
\begin{minipage}[t]{0.4\textwidth}\centering
\schemarule{\textsc{RSym}}{RxSym}{
\input{./figures/Rxij.tikz} = \input{./figures/Rxji.tikz}
}
\end{minipage}
\begin{minipage}[t]{0.65\textwidth}\centering
\schemarule{\textsc{XX1}}{Xij-Xjkbis}{
  \input{./figures/Xij-Xjk.tikz} = \input{./figures/Xik-Xij.tikz}
}
\schemarule{\textsc{XX2}}{Xij-Xjk}{
  \input{./figures/Xij-Xjk.tikz} = \input{./figures/Xjk-Xik.tikz}
}
\end{minipage}
\end{minipage}}
\caption{Derived rules in \(\QCeq\) (batch~1). The rules \eqref{Xij-Xjkbis} and \eqref{Xij-Xjk} assume that \(i,j,k\) are pairwise distinct.}
\Description{Framed panel collecting a first batch of derived circuit equalities in the theory \(\QCeq\). The figure contains small one- and two-wire diagrammatic equations involving scalar phases, Hadamard-type gates, controlled operations, and two-level swaps. Each item is displayed as an equality between a left-hand and right-hand circuit diagram with a short printed tag. The last two rules are understood under the side condition that \(i,j,k\) are pairwise distinct.}
\label{fig:derived_QC_part1}
\end{figure*}

\begin{figure*}[!ht]
\centering
\fbox{
\begin{minipage}{0.975\textwidth}\centering
\begin{minipage}[t]{0.48\textwidth}\centering
\schemarule{\textsc{HCH}}{hch}{
\input{./figures/hchL.tikz} = \input{./figures/hchR.tikz}
}
\schemarule{\textsc{CHC}}{chc}{
\input{./figures/chcL.tikz} = \input{./figures/chcR.tikz}
}
\end{minipage}\hfill
\begin{minipage}[t]{0.48\textwidth}\centering
\schemarule{\textsc{PHP}}{php}{
\input{./figures/phpL.tikz} = \input{./figures/phpR.tikz}
}
\schemarule{\textsc{HPH}}{hph}{
\input{./figures/hphL.tikz} = \input{./figures/hphR.tikz}
}
\end{minipage}
\begin{minipage}[t]{0.48\textwidth}\centering
\schemarule{\textsc{CXP}}{cxp}{
\input{./figures/cxcR.tikz} = \input{./figures/cxcL.tikz}
}
\end{minipage}
\begin{minipage}[t]{0.7\textwidth}\centering
\schemarule{\textsc{CHD}}{axiom-controlled-hadamard-decompose}{
\input{./figures/controlled-Hadij.tikz}
  =
  \input{./figures/controlled-hadamard-decomposition.tikz}
}
\schemarule{EP}{cp-dec}{
\input{./figures/cp.tikz} = \input{./figures/cp-decomposition.tikz}
}
\schemarule{\textsc{P2D}}{axiom-ccp}{
\input{./figures/ccp.tikz} = \input{./figures/ccp-decomposition.tikz}
}
\schemarule{\textsc{P3D}}{axiom-cccp}{
\input{./figures/cccp.tikz} = \input{./figures/cccp-decomposition.tikz}
}
\end{minipage}
\begin{minipage}[t]{0.4\textwidth}\centering
\schemarule{\textsc{P2P}}{ccp-cp}{
  \input{./figures/comm-phasespiL.tikz} = \input{./figures/comm-phasespiR.tikz}
}
\schemarule{\textsc{N2P}}{ccnot-cp}{
  \input{./figures/comm-phasesbasecpiL.tikz} = \input{./figures/comm-phasesbasecpiR.tikz}
}
\end{minipage}
\begin{minipage}[t]{0.48\textwidth}\centering
\schemarule{\textsc{H2P}}{ch-cp}{
  \input{./figures/ch-cpL.tikz} = \input{./figures/ch-cpR.tikz}
}
\schemarule{\textsc{P2Pd}}{ccp-ccp-diff}{
  \input{./figures/comm-phasespipi-diffL.tikz} = \input{./figures/comm-phasespipi-diffR.tikz}
}
\end{minipage}
\end{minipage}}
\caption{Derived rules in \(\QCeq\) (batch~2). The rules \eqref{ccp-cp}, \eqref{ccnot-cp}, \eqref{ch-cp}, and \eqref{ccp-ccp-diff} assume \(k\neq \ell\).}
\Description{Framed panel presenting a second batch of derived circuit equalities in \(\QCeq\). It groups totalisation identities, interaction rules for two-level swaps with control, braid-like relations between distinct two-level swaps, symmetry properties for derived two-level gates, control-removal decompositions, and commutation schemata for operations controlled on distinct basis values. Each rule is shown as an equality between two circuit diagrams with a short printed tag. The last four rules are understood under the side condition \(k\neq \ell\).}
\label{fig:derived_QC_part2}
\end{figure*}

\begin{figure*}[!ht]
\centering
\fbox{
\begin{minipage}{0.975\textwidth}\centering
\begin{minipage}[t]{0.48\textwidth}\centering
\schemarule{\textsc{P2Ps}}{ccp-ccp}{
  \input{./figures/comm-phasespipipi-sameL.tikz} = \input{./figures/comm-phasespipipi-sameR.tikz}
}
\schemarule{\textsc{P2Hd}}{ccp-ch-diff}{
  \input{./figures/Hphasecommc2A.tikz} = \input{./figures/Hphasecommc2B.tikz}
}
\schemarule{\textsc{N2Pd}}{ccnot-cp-diff}{
  \input{./figures/comm-phases-copiL.tikz} = \input{./figures/comm-phases-copiR.tikz}
}
\schemarule{\textsc{N2Hd}}{ccnot-ch-diff}{
  \input{./figures/Hphase-comm3pi-sameL.tikz} = \input{./figures/Hphase-comm3pi-sameR.tikz}
}
\schemarule{\textsc{H2Hd}}{ch-ch-diff}{
  \input{./figures/CHCH-comm-diffL.tikz} = \input{./figures/CHCH-comm-diffR.tikz}
}
\schemarule{\textsc{NCHd}}{cnot-ch-diff}{
  \input{./figures/Hphase-commpi-diffL.tikz} = \input{./figures/Hphase-commpi-diffR.tikz}
}
\end{minipage}\hfill
\begin{minipage}[t]{0.48\textwidth}\centering
\schemarule{\textsc{P2N}}{ccp-cnot}{
  \input{./figures/comm-phasespipi-sameL.tikz} = \input{./figures/comm-phasespipi-sameR.tikz}
}
\schemarule{\textsc{N2H2}}{ccnot-ch-diff-2}{
  \input{./figures/Hphase-comm2pi-sameL.tikz} = \input{./figures/Hphase-comm2pi-sameR.tikz}
}
\schemarule{\textsc{N2N4}}{ccnot-ccnot-diff-both}{
  \input{./figures/comm-phasesc4pipi-sameL.tikz} = \input{./figures/comm-phasesc4pipi-sameR.tikz}
}
\schemarule{\textsc{N2N3}}{ccnot-ccnot-diff}{
  \input{./figures/comm-phasesc3pipi-sameL.tikz} = \input{./figures/comm-phasesc3pipi-sameR.tikz}
}
\schemarule{\textsc{P2N2}}{ccp-ccnot}{
  \input{./figures/comm-phases-cmpiR.tikz} = \input{./figures/comm-phases-cmpiL.tikz}
}
\schemarule{\textsc{N2N2}}{ccnot-ccnot}{
  \input{./figures/comm-phasescpipi-sameL.tikz} = \input{./figures/comm-phasescpipi-sameR.tikz}
}
\end{minipage}
\end{minipage}}
\caption{Derived rules in \(\QCeq\) (batch~3). All rules in this batch assume \(k\neq \ell\).}
\Description{Framed panel collecting additional derived commutation and conjugation equalities in \(\QCeq\) for gates controlled on distinct basis values. The displayed equations swap the order of pairs of differently-controlled components, covering controlled phases, controlled Hadamards, and CNOT-like controlled operations, together with special cases involving \(2\pi\) phases and higher-controlled variants. Each item is presented as a left-hand and right-hand circuit diagram connected by an equality sign, under the common side condition \(k\neq \ell\).}
\label{fig:derived_QC_part3}
\end{figure*}


\subsection{Proofs of the derived rules}\label{subsec:derived-rules-proofs}


\begin{proof}[Proof of \eqref{0phase}]
\[
  \input{./figures/0phase.tikz}
  \overset{\eqref{2pi}}{=}
  \input{./figures/0phase-2pi.tikz}
  \overset{\eqref{sum}}{=}
  \input{./figures/phase2pi.tikz}
  \overset{\eqref{2pi}}{=}
  \input{./figures/empty_diagram.tikz}
\]
\end{proof}


\begin{proof}[Proof of \eqref{piphase}]
\[
  \input{./figures/negpiphase.tikz}
  \overset{\eqref{2pi}}{=}
  \input{./figures/negpiphase-2pi.tikz}
  \overset{\eqref{sum}}{=}
  \input{./figures/piphase.tikz}
\]
\end{proof}

\begin{proof}[Proof of \eqref{Xiplusun-Xiplusun}]
\begin{gather*}
  \input{./figures/Xiiplusun-Xiiplusun.tikz}\\
  = \input{./figures/Xiiplusun-Xiiplusun-01.tikz}\\
  \overset{\eqref{Hiip-Hiip}}{=}
  \input{./figures/Xiiplusun-Xiiplusun-02.tikz}
  \overset{\eqref{sum}\,\eqref{2pi}}{=}
  \input{./figures/Hiip-Hiip.tikz}\\
  \overset{\eqref{Hiip-Hiip}}{=}
  \input{./figures/line.tikz}
\end{gather*}
\end{proof}

\begin{proof}[Proof of \eqref{Xij-Xij}]
If \(i = j\), this is trivial.

If \(i < j\)
\begin{gather*}
  \scalebox{0.75}{\input{./figures/Xij-Xij.tikz}}\\
  = \scalebox{0.75}{\input{./figures/Xij-Xij-low1.tikz}}\\
  \overset{\eqref{Xiplusun-Xiplusun}}{=}
  \scalebox{0.75}{\input{./figures/Xij-Xij-low2.tikz}}\\
  \overset{\eqref{Xiplusun-Xiplusun}}{=}
  \scalebox{0.75}{\input{./figures/Xij-Xij-low3.tikz}}
  \overset{\eqref{Xiplusun-Xiplusun}}{=}
  \scalebox{0.75}{\input{./figures/line.tikz}}
\end{gather*}

If \(i > j\), by definition \(X^{i,j} = X^{j,i}\) so this is covered by the previous case.
\end{proof}

\begin{proof}[Proof of \eqref{Hij-Hij}]
If \(i = j\), this is trivial.

If \(i < j\)
\begin{gather*}
  \input{./figures/Hij-Hij.tikz}\\
  = \input{./figures/Hij-Hij-01.tikz}\\
  \overset{\eqref{Xij-Xij}}{=}
  \input{./figures/Hij-Hij-02.tikz}\\
  \overset{\eqref{Hiip-Hiip}}{=}
  \input{./figures/Hij-Hij-03.tikz}\\
  \overset{\eqref{Xij-Xij}}{=}
  \input{./figures/line.tikz}
\end{gather*}

If \(i > j\),
\begin{gather*}
  \input{./figures/Hij-Hij.tikz}
  = \input{./figures/Hij-Hij-01b.tikz}\\
  \overset{\eqref{Xij-Xij}}{=}
  \input{./figures/Hij-Hij-02b.tikz}
  \overset{\text{Previous case}}{=}
  \input{./figures/Hij-Hij-03b.tikz}
  \overset{\eqref{Xij-Xij}}{=}
  \input{./figures/line.tikz}
\end{gather*}
\end{proof}


\begin{proof}[Proof of \eqref{hdec}]
\begin{gather*}
  \input{./figures/Hadij.tikz}
  \overset{\eqref{Hij-Hij}}{=}
  \input{./figures/Hdec-01.tikz}\\
  \overset{\eqref{0phase}}{=}
  \input{./figures/Hdec-02.tikz}\\
  \overset{\eqref{eulerH}}{=}
  \input{./figures/hadamard-decomposition.tikz}
\end{gather*}
\end{proof}


\begin{proof}[Proof of \eqref{Htotal}]
\begin{gather*}
  \input{./figures/Htotal.tikz}
    \overset{\eqref{0phase},\eqref{sum}}{=}
    \input{./figures/Htotal-01.tikz}\\
    \overset{\eqref{axiom-total-phase}}{=}
    \input{./figures/Htotal-02.tikz}\\
    \overset{\eqref{0phase},\eqref{sum}, \eqref{cp-cp}}{=}
    \input{./figures/Htotal-03.tikz}\\
    \overset{\eqref{Hphasecomm}}{=}
    \input{./figures/Htotal-04.tikz}\\
  \overset{\eqref{0phase},\eqref{sum}, \eqref{cp-cp}}{=}
  \input{./figures/Htotal-05.tikz}\\
  \overset{\eqref{axiom-total-phase}}{=}
    \input{./figures/Htotal-06.tikz}\\
    \overset{\eqref{0phase},\eqref{sum}}{=}
    \input{./figures/totalH.tikz}
\end{gather*}
\end{proof}

\begin{proof}[Proof of \eqref{Rxtotal}]
\begin{gather*}
  \input{./figures/Rxtotal.tikz}
    = \input{./figures/Rxtotal-01.tikz}\\
    \overset{\eqref{Htotal}}{=}
    \input{./figures/Rxtotal-02.tikz}\\
  \overset{\eqref{cp-cp}}{=}
    \input{./figures/Rxtotal-03.tikz}\\
    \overset{\eqref{Htotal}}{=}
    \input{./figures/Rxtotal-04.tikz}
    = \input{./figures/totalRx.tikz}
\end{gather*}
\end{proof}

\begin{proof}[Proof of \eqref{XCtrl}]
\begin{gather*}
  \input{./figures/CtrlX.tikz}
  = \input{./figures/XCtrl-01.tikz}\\
  \overset{\eqref{Hij-Hij}}{=}
  \input{./figures/XCtrl-02.tikz}\\
  \overset{\eqref{eulerH}}{=}
  \input{./figures/XCtrl-03.tikz}\\
  \overset{\eqref{cp-cp}, \eqref{Htotal}, \eqref{sum}}{=}
  \scalebox{0.90}{\input{./figures/XCtrl-04.tikz}}\\
  \overset{\eqref{sum}, \eqref{piphase}}{=}
  \input{./figures/XCtrl-05.tikz}\\
  \overset{\eqref{hdec}}{=}
  \input{./figures/XCtrl.tikz}
\end{gather*}
\end{proof}


\begin{proof}[Proof of \eqref{Xij-Xjkbis}]
\begin{gather*}
  \input{./figures/Xij-Xjk.tikz}\\
    = \input{./figures/Xij-Xjk-01.tikz}\\
    \overset{\eqref{XH}}{=}
    \input{./figures/Xij-Xjk-02.tikz}\\
    \overset{\eqref{Hphasecomm}}{=}
    \input{./figures/Xij-Xjk-03.tikz}\\
    \overset{\eqref{Hij-Hij}}{=}
    \input{./figures/Xij-Xjk-04.tikz}\\
    = \input{./figures/Xik-Xij.tikz}
\end{gather*}
\end{proof}

\begin{proof}[Proof of \eqref{Xij-Xjk}]
\begin{gather*}
  \input{./figures/Xij-Xjk.tikz}
    \overset{\eqref{Xij-Xij}}{=}
    \input{./figures/Xjk-Xik-01.tikz}\\
    \overset{\eqref{Xij-Xjkbis}}{=}
    \input{./figures/Xjk-Xik-02.tikz}
    \overset{\eqref{Xij-Xij}}{=}
    \input{./figures/Xjk-Xik.tikz}
\end{gather*}
\end{proof}


\begin{proof}[Proof of \eqref{XCtrlbis}]
\begin{gather*}
  \input{./figures/XCtrlbis.tikz}
    \overset{\eqref{Xij-Xij}}{=}
    \input{./figures/XCtrlbis-01.tikz}\\
    \overset{\eqref{XCtrl}}{=}
    \input{./figures/XCtrlbis-02.tikz}
    \overset{\eqref{Xij-Xij}}{=}
    \input{./figures/CtrlXbis.tikz}
\end{gather*}
\end{proof}

\begin{proof}[Proof of \eqref{RxSym}]
\begin{gather*}
  \scalebox{0.67}{\input{./figures/Rxij.tikz}}
    = \scalebox{0.67}{\input{./figures/Rxij-01.tikz}}\\
  \overset{\eqref{sum}\,\eqref{2pi}}{=}
    \scalebox{0.67}{\input{./figures/Rxij-02.tikz}}\\
    \overset{\eqref{cp-cp}}{=}
    \scalebox{0.67}{\input{./figures/Rxij-02b.tikz}}\\
  \overset{\eqref{Hij-Hij}}{=}
    \scalebox{0.67}{\input{./figures/Rxij-03.tikz}}\\
    = \scalebox{0.67}{\input{./figures/Rxij-04.tikz}}\\
    \overset{\eqref{XH}}{=}
    \scalebox{0.67}{\input{./figures/Rxij-05.tikz}}\\
    \overset{\eqref{XCtrl},\eqref{XCtrlbis}}{=}
    \scalebox{0.67}{\input{./figures/Rxij-06.tikz}}\\
    \overset{\eqref{XH}}{=}
    \scalebox{0.67}{\input{./figures/Rxij-07.tikz}}\\
    = \scalebox{0.67}{\input{./figures/Rxij-08.tikz}}\\
    \overset{\eqref{Xij-Xij}}{=}
    \scalebox{0.67}{\input{./figures/Rxji.tikz}}
\end{gather*}
\end{proof}

\begin{proof}[Proof of \eqref{hch}]
\begin{gather*}
  \input{./figures/hchL.tikz}
  \overset{\eqref{axiom-total-hadamard}}{=}
  \input{./figures/hch-01.tikz}\\
  \overset{\eqref{ch-ch}}{=}
  \input{./figures/hch-02.tikz}\\
  \overset{\eqref{ch-ccp}}{=}
  \input{./figures/hch-03.tikz}\\
  \overset{\eqref{ch-ch}}{=}
  \input{./figures/hch-04.tikz}
  \overset{\eqref{Hij-Hij}}{=}
  \input{./figures/hchR.tikz}
\end{gather*}
\end{proof}

\begin{proof}[Proof of \eqref{chc}]
\begin{gather*}
  \input{./figures/chcL.tikz}
  = \input{./figures/chc-00.tikz}\\
  \overset{\eqref{axiom-total-phase}}{=}
  \input{./figures/chc-00b.tikz}\\
  \overset{\eqref{cp-cp}}{=}
  \input{./figures/chc-01.tikz}\\
  \overset{\eqref{axiom-compat}}{=}
  \input{./figures/chc-02.tikz}\\
  \overset{\eqref{ch-ccp}}{=}
  \input{./figures/chc-03.tikz}\\
  \overset{\eqref{axiom-compat}}{=}
  \input{./figures/chc-04.tikz}\\
  \overset{\eqref{cp-cp}}{=}
  \input{./figures/chc-05.tikz}\\
  \overset{\eqref{sum}, \eqref{0phase}}{=}
  \input{./figures/chc-06.tikz}\\
  \overset{\eqref{axiom-compat}}{=}
  \input{./figures/chcR.tikz}
\end{gather*}
\end{proof}

\begin{proof}[Proof of \eqref{php}]
\begin{gather*}
  \input{./figures/phpL.tikz}
  = 
  \input{./figures/php-00.tikz}\\
  \overset{\eqref{axiom-total-phase}}{=}
  \input{./figures/php-00b.tikz}\\
   \overset{\eqref{cp-cp}}{=}
   \input{./figures/php-01.tikz}\\
  \overset{\eqref{axiom-compat-pi}, \eqref{axiom-compat}}{=}
  \input{./figures/php-02.tikz}\\
  \overset{\eqref{ch-ccnot}}{=}
  \input{./figures/php-03.tikz}\\
  \overset{\eqref{axiom-compat-pi}, \eqref{axiom-compat}}{=}
  \input{./figures/php-04.tikz}\\
  \overset{\eqref{cp-cp}}{=}
  \input{./figures/php-05.tikz}\\
  \overset{\eqref{sum}, \eqref{0phase}}{=}
  \input{./figures/php-06.tikz}\\
  \overset{\eqref{axiom-compat-pi}, \eqref{axiom-compat}}{=}
  \input{./figures/phpR.tikz}
\end{gather*}
\end{proof}

\begin{proof}[Proof of \eqref{hph}]
\begin{gather*}
  \input{./figures/hphL.tikz}
  \overset{\eqref{axiom-total-hadamard}}{=}
  \input{./figures/hph-01.tikz}\\
  \overset{\eqref{ch-ch}}{=}
  \input{./figures/hph-02.tikz}\\
  \overset{\eqref{ch-ccnot}}{=}
  \input{./figures/hph-03.tikz}\\
  \overset{\eqref{ch-ch}, \eqref{Hij-Hij}}{=}
  \input{./figures/hphR.tikz}
\end{gather*}
\end{proof}

\begin{proof}[Proof of \eqref{cxp}]
\begin{gather*}
  \input{./figures/cxcR.tikz}
  = \input{./figures/cxc-01.tikz}\\
  \overset{\eqref{sum}\,\eqref{0phase}}{=}
  \input{./figures/cxc-02.tikz}\\
  \overset{\eqref{cp-cp}}{=}
  \input{./figures/cxc-03.tikz}\\
  \overset{\eqref{chc}}{=}
  \input{./figures/cxc-04.tikz}\\
  \overset{\eqref{cp-cnot}\,\eqref{sum}\,\eqref{0phase}}{=}
  \input{./figures/cxc-05.tikz}
  = \input{./figures/cxcL.tikz}
\end{gather*}
\end{proof}


\begin{proof}[Proof of \eqref{axiom-controlled-hadamard-decompose}]
\begin{gather*}
  \input{./figures/controlled-Hadij.tikz}
  \overset{\eqref{hdec}}{=}
  \input{./figures/ch-01.tikz}\\
  = \input{./figures/ch-02.tikz}\\
  \overset{\eqref{sum}}{=}
  \input{./figures/ch-04.tikz}\\
  \overset{\eqref{Hij-Hij}}{=}
  \input{./figures/ch-05.tikz}\\
  = \input{./figures/ch-06.tikz}\\
  \overset{\eqref{cp-cp}}{=}
  \input{./figures/ch-07.tikz}\\
  \overset{\eqref{XCtrl}}{=}
  \input{./figures/ch-07b.tikz}\\
  \overset{\eqref{sum}}{=}
  \input{./figures/ch-08.tikz}\\
  \overset{\eqref{chc}}{=}
  \input{./figures/ch-09.tikz}\\
  \overset{\eqref{cxp}}{=}
  \input{./figures/ch-10.tikz}\\
  \overset{\eqref{sum}}{=}
  \input{./figures/ch-11.tikz}\\
  \overset{\eqref{cxp}}{=}
  \input{./figures/ch-12b.tikz}\\
  = \input{./figures/ch-12.tikz}\\
  \overset{\eqref{hch}}{=}
  \input{./figures/ch-13.tikz}\\
  \overset{\eqref{hch}}{=}
  \input{./figures/ch-14.tikz}\\
  \overset{\eqref{sum}\,\eqref{0phase}}{=}
  \scalebox{0.88}{\input{./figures/ch-15.tikz}}\\
  \overset{\eqref{chc}}{=}
  \scalebox{0.88}{\input{./figures/ch-16.tikz}}\\
  \overset{\eqref{sum}\,\eqref{0phase}}{=}
  \input{./figures/ch-17.tikz}\\
  \overset{\eqref{hch}}{=}
  \input{./figures/controlled-hadamard-decomposition.tikz}
\end{gather*}
\end{proof}

\begin{proof}[Proof of \eqref{cp-dec}]
\begin{gather*}
  \input{./figures/cp-decomposition.tikz}\\
  = \input{./figures/cp-fold.tikz}\\
  \overset{\eqref{XCtrl}}{=} \input{./figures/cp-fold-01.tikz}\\
  \overset{\eqref{Xij-Xij}}{=} \input{./figures/cp-fold-02.tikz}\\
  \overset{\eqref{cp-cp}, \eqref{sum}}{=} \input{./figures/cp-fold-03.tikz}\\
  \overset{\eqref{axiom-total-phase}}{=} \input{./figures/cp-fold-04.tikz}\\
  \overset{\eqref{sum}, \eqref{0phase}}{=} \input{./figures/cp.tikz}
\end{gather*}
\end{proof}

\begin{proof}[Proof of \eqref{axiom-ccp}]
\begin{gather*}
  \input{./figures/ccp.tikz}
  \overset{\eqref{cp-dec}}{=}
  \input{./figures/ccp-01.tikz}\\
  \overset{\eqref{sum}\,\eqref{0phase}}{=}
  \scalebox{0.93}{\input{./figures/ccp-02.tikz}}\\
  \overset{\eqref{chc}}{=}
  \scalebox{0.93}{\input{./figures/ccp-03.tikz}}\\
  \overset{\eqref{hch}}{=}
  \input{./figures/ccp-04.tikz}\\
  \overset{\eqref{sum}\,\eqref{0phase}}{=}
  \input{./figures/ccp-decomposition.tikz}
\end{gather*}
\end{proof}

\begin{proof}[Proof of \eqref{axiom-cccp}]
\begin{gather*}
  \input{./figures/cccp.tikz}
  \overset{\eqref{axiom-ccp}}{=}
  \input{./figures/cccp-02.tikz}\\
  \overset{\eqref{hph}}{=}
  \input{./figures/cccp-03.tikz}\\
  \overset{\eqref{hch}}{=}
  \input{./figures/cccp-03d.tikz}\\
  \overset{\eqref{sum}\,\eqref{2pi}}{=}
  \scalebox{0.86}{\input{./figures/cccp-04.tikz}}\\
  \overset{\eqref{php}}{=}
  \scalebox{0.86}{\input{./figures/cccp-05.tikz}}\\
  \overset{\eqref{sum}\,\eqref{2pi}}{=}
  \scalebox{0.90}{\input{./figures/cccp-06.tikz}}\\
  \overset{\eqref{sum}\,\eqref{2pi}}{=}
  \input{./figures/cccp-07.tikz}\\
  \overset{\eqref{hch}}{=}
  \input{./figures/cccp-decomposition.tikz}
\end{gather*}
\end{proof}


\begin{proof}[Proof of \eqref{ccp-cp}]
\begin{gather*}
  \input{./figures/comm-phasespiL.tikz}
  \overset{\eqref{axiom-ccp}}{=}
  \scalebox{0.90}{\input{./figures/comm-phasespi-01.tikz}}\\
  \overset{\eqref{cp-cnot}}{=}
  \scalebox{0.90}{\input{./figures/comm-phasespi-02.tikz}}\\
  \overset{\eqref{cp-cp}}{=}
  \scalebox{0.90}{\input{./figures/comm-phasespi-03.tikz}}\\
  \overset{\eqref{axiom-ccp}}{=}
  \input{./figures/comm-phasespiR.tikz}
\end{gather*}
\end{proof}

\begin{proof}[Proof of \eqref{ccnot-cp}]
\begin{gather*}
  \input{./figures/comm-phasesbasecpiL.tikz}
  \overset{\eqref{axiom-cccp}}{=}
  \scalebox{0.90}{\input{./figures/comm-phasescpi-01.tikz}}\\
  \overset{\eqref{ccp-cp}}{=}
  \scalebox{0.90}{\input{./figures/comm-phasescpi-02.tikz}}\\
  \overset{\eqref{axiom-cccp}}{=}
  \input{./figures/comm-phasesbasecpiR.tikz}
\end{gather*}
\end{proof}

\begin{proof}[Proof of \eqref{ch-cp}]
\begin{gather*}
  \input{./figures/ch-cpL.tikz}
  \overset{\eqref{axiom-controlled-hadamard-decompose}}{=}
  \scalebox{0.82}{\input{./figures/ch-cp-01.tikz}}\\
  \overset{\eqref{cp-cnot}}{=}
  \scalebox{0.82}{\input{./figures/ch-cp-02.tikz}}\\
  \overset{\eqref{axiom-controlled-hadamard-decompose}}{=}
  \input{./figures/ch-cpR.tikz}
\end{gather*}
\end{proof}

\begin{proof}[Proof of \eqref{ccp-ccp-diff}]
\begin{gather*}
  \scalebox{0.68}{\input{./figures/comm-phasespipi-diffL.tikz}}\\
  \overset{\eqref{axiom-ccp}}{=}
  \scalebox{0.56}{\input{./figures/comm-phasespipi-diff-01.tikz}}\\
  \overset{\eqref{cp-cnot}}{=}
  \scalebox{0.56}{\input{./figures/comm-phasespipi-diff-02.tikz}}\\
  \overset{\eqref{cp-cp}}{=}
  \scalebox{0.56}{\input{./figures/comm-phasespipi-diff-03.tikz}}\\
  \overset{\eqref{cnot-cnot-diff}}{=}
  \scalebox{0.56}{\input{./figures/comm-phasespipi-diff-04.tikz}}\\
  \overset{\eqref{cp-cnot}}{=}
  \scalebox{0.56}{\input{./figures/comm-phasespipi-diff-05.tikz}}\\
  \overset{\eqref{axiom-ccp}}{=}
  \scalebox{0.68}{\input{./figures/comm-phasespipi-diffR.tikz}}
\end{gather*}
\end{proof}

\begin{proof}[Proof of \eqref{ccp-cnot}]
\begin{gather*}
  \input{./figures/comm-phasespipi-sameL.tikz}
  \overset{\eqref{axiom-ccp}}{=}
  \scalebox{0.90}{\input{./figures/comm-phasespipi-same-01.tikz}}\\
  \overset{\eqref{hch}}{=}
  \scalebox{0.90}{\input{./figures/comm-phasespipi-same-02.tikz}}\\
  \overset{\eqref{ch-ccp}\,\eqref{cp-cnot}\,\eqref{cnot-cnot}}{=}
  \scalebox{0.86}{\input{./figures/comm-phasespipi-same-03.tikz}}\\
  \overset{\eqref{cp-cnot}}{=}
  \scalebox{0.90}{\input{./figures/comm-phasespipi-same-04.tikz}}\\
  \overset{\eqref{hch}}{=}
  \scalebox{0.90}{\input{./figures/comm-phasespipi-same-05.tikz}}\\
  \overset{\eqref{axiom-ccp}}{=}
  \input{./figures/comm-phasespipi-sameR.tikz}
\end{gather*}
\end{proof}

\begin{proof}[Proof of \eqref{ccp-ccp}]
\begin{gather*}
  \input{./figures/comm-phasespipipi-sameL.tikz}
  \overset{\eqref{axiom-ccp}}{=}
  \scalebox{0.9}{\input{./figures/comm-phasespipipi-same-01.tikz}}\\
  \overset{\eqref{hch}}{=}
  \scalebox{0.9}{\input{./figures/comm-phasespipipi-same-02.tikz}}\\
  \overset{\eqref{ccp-cp}}{=}
  \scalebox{0.9}{\input{./figures/comm-phasespipipi-same-03.tikz}}\\
  \overset{\eqref{ccp-cnot}\,\eqref{ch-ccp}\,\eqref{cp-cnot}}{=}
  \scalebox{0.9}{\input{./figures/comm-phasespipipi-same-04.tikz}}\\
  \overset{\eqref{hch}}{=}
  \scalebox{0.9}{\input{./figures/comm-phasespipipi-same-05.tikz}}\\
  \overset{\eqref{axiom-ccp}}{=}
  \scalebox{0.9}{\input{./figures/comm-phasespipipi-sameR.tikz}}
\end{gather*}
\end{proof}

\begin{proof}[Proof of \eqref{ccp-ch-diff}]
\begin{gather*}
  \input{./figures/Hphasecommc2A.tikz}
  \overset{\eqref{axiom-controlled-hadamard-decompose}}{=}
  \scalebox{0.86}{\input{./figures/Hphasecommc2-01.tikz}}\\
  = \input{./figures/Hphasecommc2-02.tikz}\\
  \overset{\eqref{ccp-ccp-diff}}{=}
  \scalebox{0.88}{\input{./figures/Hphasecommc2-03.tikz}}\\
  = \input{./figures/Hphasecommc2-04.tikz}\\
  \overset{\eqref{axiom-controlled-hadamard-decompose}}{=}
  \input{./figures/Hphasecommc2B.tikz}
\end{gather*}
\end{proof}

\begin{proof}[Proof of \eqref{ccnot-cp-diff}]
\begin{gather*}
  \input{./figures/comm-phases-copiL.tikz}
  \overset{\eqref{axiom-cccp}}{=}
  \scalebox{0.90}{\input{./figures/comm-phasesocpi-01.tikz}}\\
  \overset{\eqref{ccp-ccp-diff}}{=}
  \scalebox{0.90}{\input{./figures/comm-phasesocpi-02.tikz}}\\
  \overset{\eqref{axiom-cccp}}{=}
  \input{./figures/comm-phases-copiR.tikz}
\end{gather*}
\end{proof}

\begin{proof}[Proof of \eqref{ccnot-ch-diff}]
\begin{gather*}
  \input{./figures/Hphase-comm3pi-sameL.tikz}
  \overset{\eqref{axiom-controlled-hadamard-decompose}}{=}
  \scalebox{0.86}{\input{./figures/Hphase-comm3pi-same-01.tikz}}\\
  = \input{./figures/Hphase-comm3pi-same-02.tikz}\\
  \overset{\eqref{ccnot-cp-diff}}{=}
  \scalebox{0.88}{\input{./figures/Hphase-comm3pi-same-03.tikz}}\\
  = \input{./figures/Hphase-comm3pi-same-04.tikz}\\
  \overset{\eqref{axiom-controlled-hadamard-decompose}}{=}
  \input{./figures/Hphase-comm3pi-sameR.tikz}
\end{gather*}
\end{proof}

\begin{proof}[Proof of \eqref{ch-ch-diff}]
\begin{gather*}
  \scalebox{0.68}{\input{./figures/CHCH-comm-diffL.tikz}}\\
  \overset{\eqref{axiom-controlled-hadamard-decompose}}{=}
  \scalebox{0.58}{\input{./figures/CHCH-comm-diff-01.tikz}}\\
  = \scalebox{0.60}{\input{./figures/CHCH-comm-diff-02.tikz}}\\
  \overset{\eqref{ccp-ccp-diff}}{=}
   \scalebox{0.60}{\input{./figures/CHCH-comm-diff-03.tikz}}\\
  = \scalebox{0.58}{\input{./figures/CHCH-comm-diff-04.tikz}}\\
  \overset{\eqref{axiom-controlled-hadamard-decompose}}{=}
  \scalebox{0.68}{\input{./figures/CHCH-comm-diffR.tikz}}
\end{gather*}
\end{proof}

\begin{proof}[Proof of \eqref{cnot-ch-diff}]
\begin{gather*}
  \input{./figures/Hphase-commpi-diffL.tikz}
  \overset{\eqref{axiom-controlled-hadamard-decompose}}{=}
  \scalebox{0.86}{\input{./figures/Hphase-commpi-diff-01.tikz}}\\
  = \scalebox{0.90}{\input{./figures/Hphase-commpi-diff-02.tikz}}\\
  \overset{\eqref{ccp-ccp-diff}}{=}
  \input{./figures/Hphase-commpi-diff-03.tikz}\\
  = \input{./figures/Hphase-commpi-diff-04.tikz}\\
  \overset{\eqref{axiom-controlled-hadamard-decompose}}{=}
  \input{./figures/Hphase-commpi-diffR.tikz}
\end{gather*}
\end{proof}

\begin{proof}[Proof of \eqref{ccnot-ch-diff-2}]
\begin{gather*}
  \input{./figures/Hphase-comm2pi-sameL.tikz}
  \overset{\eqref{axiom-compat}}{=}
  \input{./figures/Hphase-comm2pi-same-01.tikz}\\
  \overset{\eqref{ch-ccnot}}{=}
  \input{./figures/Hphase-comm2pi-same-02.tikz}\\
  \overset{\eqref{axiom-compat}}{=}
  \input{./figures/Hphase-comm2pi-sameR.tikz}
\end{gather*}
\end{proof}

\begin{proof}[Proof of \eqref{ccnot-ccnot-diff-both}]
\begin{gather*}
  \input{./figures/comm-phasesc4pipi-sameL.tikz}
  \overset{\eqref{axiom-cccp}}{=}
  \scalebox{0.90}{\input{./figures/comm-phasesc4pipi-same-01.tikz}}\\
  \overset{\eqref{ccnot-cp-diff}}{=}
  \scalebox{0.95}{\input{./figures/comm-phasesc4pipi-same-02.tikz}}\\
  \overset{\eqref{axiom-cccp}}{=}
  \input{./figures/comm-phasesc4pipi-sameR.tikz}
\end{gather*}
\end{proof}

\begin{proof}[Proof of \eqref{ccnot-ccnot-diff}]
\begin{gather*}
  \scalebox{0.62}{\input{./figures/comm-phasesc3pipi-sameL.tikz}}
  \overset{\eqref{axiom-cccp}}{=}
  \scalebox{0.52}{\input{./figures/comm-phasesc3pipi-same-01.tikz}}\\
  \overset{\eqref{ccp-ccp-diff}}{=}
  \scalebox{0.52}{\input{./figures/comm-phasesc3pipi-same-02.tikz}}\\
  \overset{\eqref{ccp-ccp}}{=}
  \scalebox{0.52}{\input{./figures/comm-phasesc3pipi-same-03.tikz}}\\
  \overset{\eqref{ccp-ccp-diff}}{=}
  \scalebox{0.52}{\input{./figures/comm-phasesc3pipi-same-04.tikz}}\\
  \overset{\eqref{ccp-ccp-diff}}{=}
  \scalebox{0.52}{\input{./figures/comm-phasesc3pipi-same-05.tikz}}\\
  \overset{\eqref{axiom-cccp}}{=}
  \scalebox{0.62}{\input{./figures/comm-phasesc3pipi-sameR.tikz}}
\end{gather*}
\end{proof}

\begin{proof}[Proof of \eqref{ccp-ccnot}]
\begin{gather*}
  \input{./figures/comm-phases-cmpiR.tikz}
  \overset{\eqref{axiom-cccp}}{=}
  \scalebox{0.92}{\input{./figures/comm-phases-cmpi-01.tikz}}\\
  \overset{\eqref{ccp-ccp}}{=}
  \scalebox{0.92}{\input{./figures/comm-phases-cmpi-02.tikz}}\\
  \overset{\eqref{ccp-ccp-diff}}{=}
  \scalebox{0.92}{\input{./figures/comm-phases-cmpi-03.tikz}}\\
  \overset{\eqref{axiom-cccp}}{=}
  \input{./figures/comm-phases-cmpiL.tikz}
\end{gather*}
\end{proof}

\begin{proof}[Proof of \eqref{ccnot-ccnot}]
\begin{gather*}
  \scalebox{0.80}{\input{./figures/comm-phasescpipi-sameL.tikz}}
  \overset{\eqref{axiom-ccp}}{=}
  \scalebox{0.80}{\input{./figures/comm-phasesc1pipi-same-01.tikz}}\\
  \overset{\eqref{sum}\,\eqref{2pi}}{=}
  \scalebox{0.80}{\input{./figures/comm-phasesc1pipi-same-02.tikz}}\\
  \overset{\eqref{ccp-ccnot}}{=}
  \scalebox{0.80}{\input{./figures/comm-phasesc1pipi-same-03.tikz}}\\
  \overset{\eqref{php}}{=}
  \scalebox{0.80}{\input{./figures/comm-phasesc1pipi-same-04.tikz}}
\end{gather*}
For any \(i\) we have
\begin{gather}
  \scalebox{0.50}{\input{./figures/comm-phasesc1pipi-same-05-01.tikz}}
  \overset{\eqref{ccnot-ch-diff-2}}{=}
  \scalebox{0.50}{\input{./figures/comm-phasesc1pipi-same-05-02.tikz}}\\
  \overset{\eqref{ccp-ccnot}}{=}
  \scalebox{0.50}{\input{./figures/comm-phasesc1pipi-same-05-03.tikz}}\\
  \overset{\eqref{ccnot-ch-diff-2}}{=}
  \scalebox{0.50}{\input{./figures/comm-phasesc1pipi-same-05-04.tikz}}\\
  \overset{\eqref{ccp-ccnot}}{=}
  \scalebox{0.50}{\input{./figures/comm-phasesc1pipi-same-05-05.tikz}}\\
  \overset{\eqref{ccp-ccnot}\,\eqref{ccnot-ch-diff-2}}{=}
  \scalebox{0.50}{\input{./figures/comm-phasesc1pipi-same-05-06.tikz}}
\end{gather}
and hence, applying this for all \(i\),
\begin{gather*}
  \scalebox{0.60}{\input{./figures/comm-phasesc1pipi-same-06.tikz}}\\
  \overset{\eqref{ccp-ccnot}}{=}
  \scalebox{0.60}{\input{./figures/comm-phasesc1pipi-same-07.tikz}}\\
  \overset{\eqref{php}}{=}
  \scalebox{0.60}{\input{./figures/comm-phasesc1pipi-same-08.tikz}}\\
  \overset{\eqref{ccp-ccnot}\,\eqref{sum}\,\eqref{2pi}}{=}
  \scalebox{0.60}{\input{./figures/comm-phasesc1pipi-same-09.tikz}}\\
  \overset{\eqref{axiom-ccp}}{=}
  \scalebox{0.60}{\input{./figures/comm-phasescpipi-sameR.tikz}}
\end{gather*}
\end{proof}

%% file: appendix_compatibility.tex

\section{Compatibility}\label{app:compatibility}

We prove Theorem~\ref{thm:compatibility}: mixed-control compatibility, in the
sense of Definition~\ref{def:control-algebra}, is derivable in \(\QCeq\).  The
argument is purely syntactic and uses the same normalisation step as
Appendix~\ref{app:commuting}.  Throughout we work modulo:
\begin{itemize}
\item the strict symmetric monoidal coherence in \(\cat{CQC}_d\)
      (associativity, unit, symmetry, interchange), so wire permutations can be
      freely inserted and removed; and
\item the coherence laws for each control functor \(\ctrl_k\)
      (functoriality, strength, and naturality), so controls can be pushed
      through structural contexts.
\end{itemize}

The reduction is a factor normalisation.  Once the two nested-control circuits
are in that form, compatibility becomes a local swap-through calculation:
every factor except the doubly controlled \(\pi\)-phase case passes through the
outer swap by structural naturality.

\subsection{Factor normalisation}\label{app:separate-normalisation}

This subsection proves the factor normalisation used in
Lemma~\ref{lemma-Greducible}.  We use the set \(\mathcal G\) of
Definition~\ref{def:G-factors}.  A contextual \(\mathcal G\)-factor is a copy
of one element of \(\mathcal G\), tensored with identities on the left and on
the right.  Structural symmetries are written as products of adjacent swaps,
which are themselves \(\mathcal G\)-factors.  Let \(\varepsilon_n\) denote the
empty sequential product at arity \(n\).

Before applying the clauses, expand \(X^{(i,j)}\), \(H^{(i,j)}\), and
\(R_x^{(i,j)}(\theta)\) into adjacent gates by
Appendix~\ref{app:derived-notation}.  The first clauses are
\[
\begin{array}{rcl}
\Separate(\id_n)&=&\varepsilon_n,\\
\Separate(T_1\circ T_2)&=&\Separate(T_1)\circ\Separate(T_2),\\
\Separate(T_1\otimes T_2)&=&
(\Separate(T_1)\otimes\id_{\arity{T_2}})\circ
(\id_{\arity{T_1}}\otimes\Separate(T_2)).
\end{array}
\]
When a control reaches a composite, tensor, adjacent target swap, or generator,
use
\[
\begin{array}{rcl}
\Separate(\ctrl_k(T_1\circ T_2))&=&
  \Separate(\ctrl_k(T_1))\circ\Separate(\ctrl_k(T_2)),\\
\Separate(\ctrl_k(T_1\otimes T_2))&=&
(\Separate(\ctrl_k(T_1))\otimes\id_{\arity{T_2}})\circ
(\id_1\otimes\sigma_{\arity{T_1},\arity{T_2}})\circ{}\\
&&(\Separate(\ctrl_k(T_2))\otimes\id_{\arity{T_1}})\circ
(\id_1\otimes\sigma_{\arity{T_2},\arity{T_1}}),\\
\Separate(\ctrl_k(\sigma_{1,1}))&=&
\Separate\Bigl(
\prod_{0\le a<b<d}
\bigl(\ctrl_k(\ctrl_a(X^{(a,b)}))\circ
(\id_1\otimes\sigma_{1,1})\circ{}\\
&&\qquad
\ctrl_k(\ctrl_a(X^{(a,b)}))\circ(\id_1\otimes\sigma_{1,1})\circ
\ctrl_k(\ctrl_a(X^{(a,b)}))\bigr)\Bigr),\\
\Separate(\ctrl_k(\id_n))&=&\varepsilon_{n+1},\\
\Separate(\ctrl_k(\input{./figures/phasealpha.tikz}))&=&\ctrl_k(\input{./figures/phasealpha.tikz}),\\
\Separate(\ctrl_k(H^{(r,r+1)}))&=&\Separate(\mathrm{CHD}_{k,r}),\\
\Separate(\ctrl_{k\ell}(\input{./figures/piphase.tikz}))&=&
  \ctrl_{k\ell}(\input{./figures/piphase.tikz}),\\
\Separate(\ctrl_{k\ell}(\input{./figures/phasebeta.tikz}))&=&
  \Separate(\mathrm{P2D}_{k,\ell,\beta}),\\
\Separate(\ctrl_{k\ell m}(\input{./figures/phasealpha.tikz}))&=&
  \Separate(\mathrm{P3D}_{k,\ell,m,\alpha}),\\
\Separate(\ctrl_k(T))&=&\Separate(\ctrl_k(\Separate(T)))
\quad\text{otherwise.}
\end{array}
\]
For a controlled block permutation, first write the block permutation as a
product of adjacent target swaps.  The adjacent-swap clause is
\eqref{swap-decomp} with the outer \(k\)-control placed on the three
non-structural factors in each \((a,b)\)-term; the two occurrences of
\(\id_1\otimes\sigma_{1,1}\) are the target crossings displayed in the
left-hand side of \eqref{swap-decomp}.  If \(F\) is already a contextual
\(\mathcal G\)-factor and none of the controlled clauses applies, set
\(\Separate(F)=F\).
Subscripts such as \(k\ell m\) abbreviate nested controls in the input of a
recursive clause.  The terms
\(\mathrm{CHD}_{k,r}\), \(\mathrm{P2D}_{k,\ell,\beta}\), and
\(\mathrm{P3D}_{k,\ell,m,\alpha}\) denote the right-hand sides of the derived
rules \eqref{axiom-controlled-hadamard-decompose}, \eqref{axiom-ccp}, and
\eqref{axiom-cccp}.

\begin{lemma}\label{lem:separate-terminates}
\Separate{} terminates on all inputs \(T\) and returns a possibly empty
sequential product of contextual \(\mathcal G\)-factors, with
\(\QCeq\vdash T=\Separate(T)\).
\end{lemma}

\begin{proof}
After the derived one-qudit notation has been expanded, order recursive calls
lexicographically by \(\mu(T)=(s(T),b(T))\).  Here \(s(T)\) is the number of
constructor nodes \(\circ,\otimes,\ctrl_k\) in the raw syntax of \(T\), and
\(b(T)\) is the number of subterms \(\ctrl_u(S)\) whose controlled body is not
a contextual \(\mathcal G\)-factor.  The usual lexicographic order on
\(\mathbb N^2\) is well founded.

The clauses for composition, tensor, controlled composition, and controlled
tensor call \(\Separate{}\) only on proper subterms, possibly with the same
outer control added.  In each case the constructor under inspection has been
removed, so \(s\) strictly decreases.  The identity clauses return the empty
product.  A plain contextual \(\mathcal G\)-factor returns itself.

The remaining cases occur when a control reaches a generator.  Controlled
phases and doubly controlled \(\pi\)-phases are already contextual
\(\mathcal G\)-factors.  For controlled adjacent swaps, the clause uses the
fixed swap presentation \eqref{swap-decomp}; the recursive calls then act on
controlled \(X^{(a,b)}\), which are expanded into adjacent gates by
Appendix~\ref{app:derived-notation}.  For controlled adjacent Hadamards,
twice-controlled phases, and three-times-controlled phases, the clauses use
the derived rules \eqref{axiom-controlled-hadamard-decompose},
\eqref{axiom-ccp}, and \eqref{axiom-cccp}.  Inspecting the right-hand sides of
these rules gives finite products in which the corresponding controlled body
has fewer outer controls before it reaches a \(\mathcal G\)-factor; this
strictly decreases \(b\).

In the final otherwise clause, the body \(T\) is first separated by induction.
Functoriality of \(\ctrl_k\) turns \(\ctrl_k(\Separate(T))\) into a product of
controls applied to already separated factors.  Each new recursive call is
then either one of the generator cases above or has smaller \(b\).  Hence the
procedure terminates.

Every step is justified by structural coherence, by a control-functor law, by
the fixed swap presentation \eqref{swap-decomp}, by a derived one-qudit
definition from Appendix~\ref{app:derived-notation}, or by one of the three
derived rules \eqref{axiom-controlled-hadamard-decompose},
\eqref{axiom-ccp}, and \eqref{axiom-cccp}.  Therefore the output is provably
equal to the input in \(\QCeq\).
\end{proof}

\begin{corollary}\label{corollary-Greducible}
Every morphism of \(\CQC\) is provably equal in \(\QCeq\) to a possibly empty
finite sequential product of contextual \(\mathcal{G}\)-factors.
\end{corollary}

\begin{lemma}[Controlled factor checks]\label{lem:finite-factor-laws}
Let \(\mathcal G^\circ\) be \(\mathcal G\) without the adjacent-swap core, and
let \(L,M:n\to n\) be contextual \(\mathcal G^\circ\)-factors.  In \(\QCeq\),
\(\ctrl_s(L)\circ\ctrl_t(M)=\ctrl_t(M)\circ\ctrl_s(L)\) for \(s\neq t\), and
\(\ctrl_s(\ctrl_t(L))\circ(\sigma_{1,1}\otimes\id_n)
= (\sigma_{1,1}\otimes\id_n)\circ\ctrl_t(\ctrl_s(L))\).
If \(P:n\to n\) is a contextual factor whose core is either a phase
or an adjacent Hadamard, then
\(\ctrl_0(P)\circ\cdots\circ\ctrl_{d-1}(P)=\id_1\otimes P\).
\end{lemma}
\begin{proof}
For branch commutation, target-wire symmetries reduce the placement of \(L\)
and \(M\) to the finite adjacent-window offsets displayed in
Appendix~\ref{app:commuting}.  The table there cites the primitive \(B\)-rules
of Figure~\ref{fig:axioms_QC_part2} and the derived commutations of
Appendix~\ref{app:derivations}.  The \(S\)-rules give the
disjoint-active-level cases.

For nested controls, factors touching at most one of the two distinguished
control wires pass the outer swap by structural naturality.  The remaining
two-wire cases are the primitive compatibility cells \eqref{axiom-compat} and
\eqref{axiom-compat-pi}, together with the separated expansions
\eqref{axiom-controlled-hadamard-decompose}, \eqref{axiom-ccp}, and
\eqref{axiom-cccp}.  This gives the one-factor case of
Lemma~\ref{lem:separate-outer-swap-equivariant}.  The totalisation statement
is the phase rule \eqref{axiom-total-phase} and the adjacent-Hadamard rule
\eqref{axiom-total-hadamard}, tensored with identities.  PROP coherence gives
the stated contextual instances.
\end{proof}

\begin{lemma}\label{lem:separate-outer-swap-equivariant}
Let \(a\ne b\), let \(f:n\to n\), and put
\(F_{ab}:=\ctrl_a(\ctrl_b(f))\),
\(F_{ba}:=\ctrl_b(\ctrl_a(f))\), and
\(S:=\sigma_{1,1}\otimes\id_n\).  Then
\(\QCeq\vdash \Separate(F_{ab})\circ S =
S\circ \Separate(F_{ba})\).
\end{lemma}

\begin{proof}
We inspect the recursion defining \Separate{} and argue by induction on the
termination measure used in Lemma~\ref{lem:separate-terminates}.  The induction
invariant is slightly more general: if a recursive call is made under the two
distinguished outer controls \(a,b\), then exchanging those two controls before
running the same recursive clause exchanges the two controls in every produced
factor and preserves the sequential order of the produced factors.

The clauses for composition and tensor are functorial.  For
\(\ctrl_m(g\circ h)\), \Separate{} produces the separated \(g\)-part followed
by the separated \(h\)-part on both sides, so the induction hypotheses for the
two subcalls compose in the same order.  For \(\ctrl_m(g\otimes h)\), the
procedure serialises tensor by inserting only target-wire symmetries and
identity padding.  These structural maps do not involve the two distinguished
outer control values; sliding \(S\) past them is strict symmetric-monoidal
naturality, and the two recursive subcalls are handled by induction.

The identity and phase-generator clauses follow from the definitions.  The
clauses that expand controlled adjacent Hadamards, controlled phases, and
triply controlled phases
use the derived decompositions
\(\eqref{axiom-controlled-hadamard-decompose}\),
\(\eqref{axiom-ccp}\), and \(\eqref{axiom-cccp}\).  These decompositions are
closed under adding surrounding controls and under target-wire permutation, so
the expansion obtained after exchanging the two outer controls is the
exchanged expansion of the original call; the induction hypothesis applies to
the recursive calls created by these expansions.

It remains to consider the one branch where \Separate{} leaves a
nested control undecomposed: a doubly controlled scalar \(\pi\)-phase.  If the
factor involves at most one of the distinguished outer controls, \(S\) passes it
by structural naturality.  If the factor is supported on both distinguished
controls, it has the form
\(\ctrl_a(\ctrl_b(\input{./figures/piphase.tikz}))\) in structural context on one side and
\(\ctrl_b(\ctrl_a(\input{./figures/piphase.tikz}))\) in the corresponding context on the
other.  The required swap-through step is the compatibility axiom
\(\eqref{axiom-compat}\), tensored with identities and conjugated by the
surrounding target-wire symmetry.

Every factor produced by \Separate{} falls under one of the cases above.  Pushing
\(S\) through the finite separated product from right to left therefore
transforms the factors of \(\Separate(F_{ab})\) into the corresponding factors of
\(\Separate(F_{ba})\), in the same sequential order, giving the stated
derivation.
\end{proof}

\subsection{Compatibility of mixed controls}\label{app:compatibility-proof}

Mixed-control compatibility follows from the factor normalisation and the
swap-through lemma proved above.

\begin{proof}[Proof of Theorem~\ref{thm:compatibility}]
Write \(S:=\sigma_{1,1}\otimes \id_{n}:2+n\to 2+n\).
Consider the two \((2+n)\)-wire circuits
\(F_{ab}:=\ctrl_{a}(\ctrl_{b}(f))\) and
\(F_{ba}:=\ctrl_{b}(\ctrl_{a}(f))\).
If \(a=b\), the desired equality \(F_{aa}\circ S=S\circ F_{aa}\) is the
same-control nested-swap coherence included in the structural congruence of
\(\CQC\).  Hence assume \(a\ne b\) for the rest of the proof.

Apply the normalisation procedure \Separate{} from
Appendix~\ref{app:separate-normalisation} to \(F_{ab}\) and \(F_{ba}\).
By construction, each rewrite step performed by \Separate{} is justified
either by the definitional expansions of Appendix~\ref{app:derived-notation},
by strict PROP coherence, by control-functoriality/strength/naturality, or by
one of the bounded-arity axiom schemata of
Figures~\ref{fig:axioms_QC_part1}--\ref{fig:axioms_QC_part2} or the derived
factor-expansion rules cited above.
Therefore \(\QCeq\vdash F_{ab}=\Separate(F_{ab})\) and
\(\QCeq\vdash F_{ba}=\Separate(F_{ba})\).

By Lemma~\ref{lem:separate-outer-swap-equivariant},
\(\QCeq\vdash \Separate(F_{ab})\circ S = S\circ \Separate(F_{ba})\).
Replacing \(\Separate(F_{ab})\) by \(F_{ab}\) and \(\Separate(F_{ba})\) by
\(F_{ba}\) using the equalities above yields
\(F_{ab}\circ S = S\circ F_{ba}\), as required.
\end{proof}

%% file: appendix_commuting.tex

\section{Commutativity}\label{app:commuting}

Fix \(d\ge 2\) and the polycontrolled PROP \(\cat{CQC}_d\) of
\S\ref{sec:qudit-circuits}.  Theorem~\ref{thm:commuting} says that circuits
controlled on \emph{distinct} computational-basis values commute.

As in Appendix~\ref{app:compatibility}, we work modulo strict symmetric
monoidal coherences and the coherence laws of each control functor.  We use the
factor normalisation of Lemma~\ref{lem:separate-terminates} to reduce the
commutativity statement to a finite list of generator-level commutations.

\begin{proof}[Proof of Theorem~\ref{thm:commuting}]
By Lemma~\ref{lem:separate-terminates}, the proof reduces to commuting two
placed controlled contextual \(\mathcal G\)-factors.  Plain adjacent-swap
factors are structural and are handled by naturality, so the table below
concerns the cores in
\(\mathcal G^\circ:=\mathcal G\setminus\{\sigma_{1,1}\}\).  Repeating the pairwise swaps
commutes the finite products returned by \Separate{}, by induction on the two
product lengths.  By SMC coherence we may therefore assume that
\(f = id_p \otimes f' \otimes id_{n-n_f-p}\) and
\(g = id_q \otimes g' \otimes id_{n-n_g-q}\), with
\(f' : n_f \to n_f\) and \(g' : n_g \to n_g\) belonging to
\(\mathcal{G}^{\circ}\), and where \(0\le p \le q\).
Identity factors have already been erased as empty products in
Lemma~\ref{lem:separate-terminates}, so neither local core is an identity.

If the target windows of \(f'\) and \(g'\) are disjoint, plain tensor
interchange does not apply, since the two factors still read the same outer
control wire with distinct branch values.  We close any irrelevant gap between
the two target windows by target-wire symmetries.  Control naturality for target
permutations moves those symmetries through the outer controls, so a
commutation with separated target windows is conjugate to the corresponding
commutation with adjacent target windows.  Scalar cores have \(n_f=0\) (or
\(n_g=0\)) and are placed at the edge of the same local window.  Hence it is
enough to check the adjacent-window offsets recorded in the table below:
\(\Delta=0\) for coincident left edge or scalar-core cases, \(\Delta=1\) for
the first shifted one-wire cases, and \(\Delta=2\) for the first shifted
two-wire cases.  Larger target gaps are obtained from these by tensoring
identities and undoing the target permutation.

The remaining possibilities are summarised in the following table.  Each
non-trivial entry cites either a generator commutation from
Section~\ref{subsec:derived-control-algebra} or a derived commutation proved in
Appendix~\ref{app:derivations}.
The label \textit{sym.} means that the
transposed generator pair is handled by the same cited equality, used in the
opposite direction after exchanging the names of the two factors and tensoring
with identities on untouched wires.  This does not conflict with the convention
\(p\le q\): that convention only fixes the left edge of the local window, while
the cited diagrammatic equations are equalities in the symmetric monoidal
theory and may be reflected by coherence.  The symbols \(i,j,k,\ell\) inside
the table are local level/control indices of the displayed cores, while
\(\Delta\) is the relative offset of the two placed cores.  A dash means that
after the gap-closing reduction and, if necessary, exchanging the names of the
two factors, the case is one of the displayed non-dash entries.

\begin{center}
\renewcommand{\arraystretch}{1.2}
\begin{adjustbox}{max width=\textwidth}
\begin{tabular}{c|c|c|c|c}
 & \(f' = \input{./figures/phase.tikz}\) &
 \(f' = \ctrl_{i}(\input{./figures/phase.tikz})\) &
 \(f' = \ctrl_{i}(\ctrl_{j}(\input{./figures/piphase.tikz}))\) &
 \(f' = \input{./figures/Hadiip.tikz}\)\\ \hline
\(g' = \input{./figures/phaseb.tikz}\) &
\(\Delta=0\): \eqref{cp-cp} &
 \(\Delta=0\): \eqref{ccp-cp} &
 \(\Delta=0\): \eqref{ccnot-cp} &
 \(\Delta=0\): \eqref{ch-cp} \\ \hline
\(g' = \ctrl_{k}(\input{./figures/phaseb.tikz})\) &
\(\Delta=0\): \eqref{ccp-cp} \textit{ sym.} &
 \(\Delta=0\): \eqref{ccp-ccp} &
 \(\Delta=0\): \eqref{ccp-ccnot} &
 \(\Delta=0\): \eqref{ch-ccp} \\
 & -- & \(\Delta=1\): \eqref{ccp-ccp-diff} &
 \(\Delta=1\): \eqref{ccnot-cp-diff} &
 \(\Delta=1\): \eqref{ccp-ch-diff} \\ \hline
\(g' = \ctrl_{k}(\ctrl_{\ell}(\input{./figures/piphase.tikz}))\) &
\(\Delta=0\): \eqref{ccnot-cp} \textit{ sym.} &
 \(\Delta=0\): \eqref{ccp-ccnot} \textit{ sym.} &
 \(\Delta=0\): \eqref{ccnot-ccnot} &
 \(\Delta=0\): \eqref{ch-ccnot} \\
 & -- & \(\Delta=1\): \eqref{ccnot-cp-diff} \textit{ sym.} &
 \(\Delta=1\): \eqref{ccnot-ccnot-diff} &
 \(\Delta=1\): \eqref{ccnot-ch-diff-2} \\
 & -- & -- &
 \(\Delta=2\): \eqref{ccnot-ccnot-diff-both} &
 \(\Delta=2\): \eqref{ccnot-ch-diff} \\ \hline
\(g' = \input{./figures/Hadiip.tikz}\) &
\(\Delta=0\): \eqref{ch-cp} \textit{ sym.} &
 \(\Delta=0\): \eqref{ch-ccp} \textit{ sym.} &
 \(\Delta=0\): \eqref{ch-ccnot} \textit{ sym.} &
 \(\Delta=0\): \eqref{ch-ch} \\
 & -- & \(\Delta=1\): \eqref{ccp-ch-diff} \textit{ sym.} &
 \(\Delta=1\): \eqref{ccnot-ch-diff-2} \textit{ sym.} &
\(\Delta=1\): \eqref{ch-ch-diff}
\end{tabular}
\end{adjustbox}
\end{center}

The Hadamard row and column use adjacent instances
\(H^{(r,r+1)}\), as required by \(\mathcal G\).  Each cited lemma states that
the corresponding pair of controlled generators commutes for the
indicated offset \(\Delta\).  Tensoring with identities on the untouched wires and
inserting coherence isomorphisms as needed yields
\(\ctrl_{k}(f)\circ\ctrl_{\ell}(g)=\ctrl_{\ell}(g)\circ\ctrl_{k}(f)\)
for each of the finitely many generator pairs.  Repeating these swaps over the
normal forms of the two circuits gives the general case.
\end{proof}

%% file: appendix_exhaustivity.tex
\section{Exhaustivity}\label{app:exhaustivity}

The exhaustivity component of Theorem~\ref{thm:exhaustive} uses the two
control-algebra principles already established in
Appendices~\ref{app:compatibility} and~\ref{app:commuting}.  Semantically, the
equation decomposes the control wire into its \(d\) basis projectors.  The
syntactic proof proves the generator cases and then closes them under
composition, tensor product, and added controls.

We use the convention \(\prod_{k=0}^{m} f_k=f_m\circ\cdots\circ f_0\).
For a circuit \(f\), write \(\mathsf{E}(f)\) for the equation
\(\prod_{k=0}^{d-1}\ctrl_{k}(f)=\gI\otimes f\).  The totalisation support
rules give the generator cases; commutativity and compatibility propagate the
equation through the circuit constructors.

\begin{lemma}\label{lem:exhaustivity-generators}
The property \(\mathsf{E}(f)\) holds when \(f\) is the monoidal unit, a single
wire, a scalar phase, or an adjacent two-level Hadamard.
\end{lemma}

\begin{proof}
For the unit and one-wire identity, the control functor laws give
\(\prod_{k=0}^{d-1}\ctrl_{k}(\gempty)=\gI=\gI\otimes\gempty\) and
\(\prod_{k=0}^{d-1}\ctrl_{k}(\gI)=\id_{2}=\gI\otimes\gI\).  The general
identity \(\id_n\) then follows by repeated tensoring, using
Lemma~\ref{lem:exhaustivity-closure}.

The remaining primitive cases are the two totalisation support rules.
For a scalar phase, \eqref{axiom-total-phase} gives
\(\prod_{k=0}^{d-1}\ctrl_{k}(\input{./figures/phase.tikz})=\gI\otimes\input{./figures/phase.tikz}\).
For an adjacent two-level Hadamard, the adjacent instance of
\eqref{axiom-total-hadamard} gives
\(\prod_{k=0}^{d-1}\ctrl_{k}(\input{./figures/Hadiip.tikz})
=\gI\otimes\input{./figures/Hadiip.tikz}\).
\end{proof}

\begin{lemma}\label{lem:exhaustivity-closure}
The property \(\mathsf{E}\) is closed under sequential composition, tensor
product, and adding one more control.
\end{lemma}

\begin{proof}
We check the three constructors separately.

\emph{Sequential composition.}
Assume \(\mathsf{E}(g)\) and \(\mathsf{E}(h)\).  Functoriality gives
\(\ctrl_{k}(g\circ h)=\ctrl_{k}(g)\circ\ctrl_{k}(h)\).  In the product over
\(k\), Appendix~\ref{app:commuting} allows the differently controlled factors
to be collected into the two blocks belonging to \(g\) and \(h\):
\begin{align*}
  \prod_{k=0}^{d-1}\ctrl_{k}(g\circ h)
  &= \prod_{k=0}^{d-1}\bigl(\ctrl_{k}(g)\circ\ctrl_{k}(h)\bigr)\\
  &= \left(\prod_{k=0}^{d-1}\ctrl_{k}(g)\right)
     \circ
     \left(\prod_{k=0}^{d-1}\ctrl_{k}(h)\right)\\
  &= (\gI\otimes g)\circ(\gI\otimes h)
   = \gI\otimes(g\circ h).
\end{align*}
Therefore \(\mathsf{E}(g\circ h)\) holds.

\emph{Tensor product.}
Assume \(\mathsf{E}(g)\) and \(\mathsf{E}(h)\), and put
\(a:=\arity{g}\), \(b:=\arity{h}\).  Strength for unused wires and symmetry
naturality express a controlled tensor as the serialised composite
\begin{equation*}
  \ctrl_{k}(g\otimes h)=
  (\ctrl_{k}(g)\otimes \id_b)
  \circ(\gI\otimes\swap{a}{b})
  \circ(\ctrl_{k}(h)\otimes \id_a)
  \circ(\gI\otimes\swap{b}{a}).
\end{equation*}
Collecting the controlled \(g\)-parts and controlled \(h\)-parts, again using
Appendix~\ref{app:commuting}, gives
\begin{align*}
  \prod_{k=0}^{d-1}\ctrl_{k}(g\otimes h)
  &=
    \left(\prod_{k=0}^{d-1}(\ctrl_{k}(g)\otimes \id_b)\right)
    \circ(\gI\otimes\swap{a}{b})\\
  &\qquad{}\circ
    \left(\prod_{k=0}^{d-1}(\ctrl_{k}(h)\otimes \id_a)\right)
    \circ(\gI\otimes\swap{b}{a})\\
  &=
    (\gI\otimes g\otimes \id_b)
    \circ(\gI\otimes\swap{a}{b})
    \circ(\gI\otimes h\otimes \id_a)\\
  &\qquad{}\circ(\gI\otimes\swap{b}{a})
   = \gI\otimes(g\otimes h).
\end{align*}
Therefore \(\mathsf{E}(g\otimes h)\) holds.

\emph{Nested control.}
Assume \(\mathsf{E}(g)\), fix \(m\in\{0,\dots,d-1\}\), and write
\(S:=\sigma_{1,1}\otimes\id_{\arity{g}}\).  Compatibility from
Appendix~\ref{app:compatibility} gives
\(\ctrl_{k}(\ctrl_{m}(g))=
S\circ\ctrl_{m}(\ctrl_{k}(g))\circ S\).  Since \(S^2=\id\), the outer swaps
remain at the two ends of the product:
\begin{align*}
  \prod_{k=0}^{d-1}\ctrl_{k}(\ctrl_{m}(g))
  &=
    S\circ
    \left(\prod_{k=0}^{d-1}\ctrl_{m}(\ctrl_{k}(g))\right)
    \circ S\\
  &=
    S\circ
    \ctrl_{m}\left(\prod_{k=0}^{d-1}\ctrl_{k}(g)\right)
    \circ S\\
  &=
    S\circ\ctrl_{m}(\gI\otimes g)\circ S
   = \gI\otimes\ctrl_{m}(g).
\end{align*}
The second equality uses functoriality of \(\ctrl_m\) together with the
product convention \(\prod_{k=0}^{d-1} f_k=f_{d-1}\circ\cdots\circ f_0\), so
the order of the \(k\)-indexed factors is preserved inside \(\ctrl_m\).
The last equality is just the strength and coherence law for \(\ctrl_m\), with
the two control wires exchanged back by \(S\).  Therefore
\(\mathsf{E}(\ctrl_m(g))\) holds.
\end{proof}

\begin{proof}[Proof of Theorem~\ref{thm:exhaustive}]
By the swap decomposition \eqref{swap-decomp} and strict symmetric-monoidal
coherence, it suffices to treat presentations built from the non-structural
generators without taking swaps as primitive constructors.  Lemma
\ref{lem:exhaustivity-generators} gives the base cases, and
Lemma~\ref{lem:exhaustivity-closure} gives the induction steps for
\(\circ\), \(\otimes\), and \(\ctrl_m\).  Structural induction therefore proves
\(\mathsf{E}(f)\) for every such presentation.

The same argument applies to the fixed decomposition of each swap
\(\swap{m}{n}\), so structural symmetries are covered as well.  Therefore
\(\prod_{k=0}^{d-1}\ctrl_k(f)=\gI\otimes f\) for every circuit
\(f\in\CQC\).
\end{proof}

\begin{corollary}\label{cor:iterated-exhaustivity}
For every \(m\ge0\) and every circuit \(f\), the product over all control
words \(u\in[d]^m\), taken in lexicographic order, satisfies
\(\prod_{u\in[d]^m}\ctrl_u(f)=\id_m\otimes f\).
\end{corollary}

\begin{proof}
The case \(m=0\) is the empty control word.  For the induction step, write
each word as \(ku\) with \(k\in[d]\) and \(u\in[d]^m\).  Functoriality of
\(\ctrl_k\) and the induction hypothesis give
\(\prod_{u\in[d]^m}\ctrl_k(\ctrl_u(f))
= \ctrl_k\left(\prod_{u\in[d]^m}\ctrl_u(f)\right)
= \ctrl_k(\id_m\otimes f)\).
Taking the product over \(k\) and applying Theorem~\ref{thm:exhaustive} once
more yields \(\id_1\otimes(\id_m\otimes f)=\id_{m+1}\otimes f\).
\end{proof}

\begin{remark}
Semantically, this says that the projectors
\((\ket{k}\bra{k})_{k=0}^{d-1}\) on the control wire form a partition of the
identity.  The syntactic form above is the one used later when decoding
Gray-code gadgets produces products over all control branches.
\end{remark}

%% file: appendix_derivations_after_commuting.tex
\section{Further derived identities}\label{app:derivations-p2}

The identities in this appendix are derived after the control algebra of
Section~\ref{subsec:derived-control-algebra} is available.  We write
\(\mathrm{QC}_{d} \vdash C_1 = C_2\) when the circuit identity \(C_1 = C_2\) follows
from the axioms of
Figures~\ref{fig:axioms_QC_part1}--\ref{fig:axioms_QC_part2}.

\subsection{Catalogue of derived rules}

\begin{figure*}[!ht]
\centering
\fbox{
\begin{minipage}{0.975\textwidth}\centering
\begin{minipage}[t]{0.6\textwidth}\centering
\schemarule{\textsc{HP2}}{h-ccpt}{
  \input{./figures/hthroughPL.tikz} = \input{./figures/hthroughPR.tikz}
}
\end{minipage}
\begin{minipage}[t]{0.6\textwidth}\centering
\schemarule{\textsc{CHHC}}{Ch-Hc}{
\input{./figures/Ch-HcL.tikz} = \input{./figures/Ch-HcR.tikz}
}
\end{minipage}
\begin{minipage}[t]{0.5\textwidth}\centering
\schemarule{\textsc{CXHC}}{Cx-Hc}{
\input{./figures/Cx-HcL.tikz} = \input{./figures/Cx-HcR.tikz}
}
\end{minipage}
\begin{minipage}[t]{0.45\textwidth}\centering
\schemarule{\textsc{CXXC}}{Cx-Xc}{
\input{./figures/Cx-XcL.tikz} = \input{./figures/Cx-XcR.tikz}
}
\end{minipage}
\end{minipage}}
\caption{Derived rules in \(\QCeq\) (batch~4). The rule \eqref{h-ccpt} assumes that \(m,\ell,b\) are pairwise distinct. The rules \eqref{Ch-Hc}, \eqref{Cx-Hc}, and \eqref{Cx-Xc} hold when either \(i,j,k\) are pairwise distinct or \(\ell,m,n\) are pairwise distinct.}
\Description{Framed panel collecting derived commutation rules between controlled Hadamards, controlled swaps, and (multi-)controlled phases. Each equation is displayed as a circuit equality, with the relevant index side conditions stated once in the caption.}
\label{fig:derived_QC_part4}
\end{figure*}

\begin{figure*}[!ht]
\centering
\fbox{
\begin{minipage}{0.975\textwidth}\centering
\begin{minipage}[t]{0.75\textwidth}\centering
\schemarule{RxXX}{RxXX}{
\input{./figures/RxXX.tikz} = \input{./figures/XXRx.tikz}
}
\schemarule{\textsc{3CRx}}{3CRx}{
\input{./figures/3bsaltL.tikz} = \input{./figures/3bsaltR.tikz}
}
\schemarule{\textsc{EulB}}{eulerb}{
\input{./figures/eulerbL.tikz} = \input{./figures/eulerR.tikz}
}
\schemarule{\textsc{Eulr}}{euler}{
\input{./figures/eulerL.tikz} = \input{./figures/eulerR.tikz}
}
\end{minipage}
\end{minipage}}
\caption{Derived rules in \(\QCeq\) (batch~5).}
\Description{Framed panel collecting the controlled rotation-through-controlled-X rule, the mixed controlled three-rotation rule, and two Euler-style derived circuit equalities.}
\label{fig:derived_QC_part5}
\end{figure*}

\subsection{Proofs of the derived rules}

\begin{proof}[Proof of \eqref{h-ccpt}]
\begin{gather*}
  \input{./figures/hthroughPL.tikz}\\
    \overset{\eqref{axiom-ccp}}{=}
    \input{./figures/hthroughP-01.tikz}\\
    \overset{\eqref{Hphasecomm}}{=}
    \input{./figures/hthroughP-02.tikz}\\
    \overset{\eqref{Hcnot}}{=}
    \input{./figures/hthroughP-03.tikz}\\
    \overset{\eqref{axiom-ccp}}{=}
    \input{./figures/hthroughPR.tikz}
\end{gather*}
\end{proof}

\begin{proof}[Proof of \eqref{Ch-Hc}]
We treat the case where \(i,j,k\) are pairwise distinct; the case with \(\ell,m,n\) is analogous.
\begin{gather*}
  \input{./figures/Ch-HcL.tikz}\\
    = \input{./figures/Ch-Hc-01.tikz} \\
    \overset{\eqref{Hij-Hij}}{=}
    \scalebox{0.78}{\input{./figures/Ch-Hc-02.tikz}}\\
    \overset{\eqref{hch}}{=}
    \scalebox{0.78}{\input{./figures/Ch-Hc-03.tikz}} \\
    \overset{\eqref{Hij-Hij}}{=}
    \input{./figures/Ch-Hc-04.tikz}\\
    \overset{\eqref{h-ccpt}}{=}
    \input{./figures/Ch-Hc-05.tikz}\\
    \overset{\eqref{hch},\eqref{Hij-Hij}}{=}
    \input{./figures/Ch-Hc-05b.tikz}\\
    \overset{\text{Commutativity}}{=}
    \input{./figures/Ch-Hc-06.tikz}\\
    \overset{\eqref{hch},\eqref{Hij-Hij}}{=}
    \input{./figures/Ch-Hc-06b.tikz}\\
    \overset{\eqref{h-ccpt}}{=}
    \input{./figures/Ch-Hc-07.tikz}\\
    = \input{./figures/Ch-HcR.tikz}
\end{gather*}
\end{proof}

\begin{proof}[Proof of \eqref{Cx-Hc}]
Again consider \(i,j,k\) pairwise distinct; the case with \(\ell,m,n\) is analogous.
\begin{gather*}
  \input{./figures/Cx-HcL.tikz}\\
    = \input{./figures/Cx-Hc-01.tikz}\\
    \overset{\eqref{Hij-Hij}}{=}
    \input{./figures/Cx-Hc-02.tikz}\\
    \overset{\eqref{sum},\eqref{0phase}}{=}
    \input{./figures/Cx-Hc-03.tikz}\\
    \overset{\eqref{Htotal}}{=}
    \input{./figures/Cx-Hc-04.tikz}\\
    = \input{./figures/Cx-Hc-05.tikz}\\
    \overset{\eqref{Ch-Hc}}{=}
    \input{./figures/Cx-Hc-06.tikz}\\
    = \input{./figures/Cx-Hc-07.tikz}\\
    \overset{\eqref{hch},\eqref{Hij-Hij}}{=}
    \input{./figures/Cx-Hc-08.tikz}\\
    \overset{\eqref{h-ccpt},\text{Commutativity}}{=}
    \input{./figures/Cx-Hc-09.tikz}\\
    = \input{./figures/Cx-HcR.tikz}
\end{gather*}
\end{proof}

\begin{proof}[Proof of \eqref{Cx-Xc}]
For \(i,j,k\) pairwise distinct (the case with \(\ell,m,n\) is similar),
\begin{gather*}
  \input{./figures/Cx-XcL.tikz}\\
    = \input{./figures/Cx-Xc-01.tikz}\\
    \overset{\eqref{Hij-Hij}}{=}
    \input{./figures/Cx-Xc-02.tikz}\\
    \overset{\eqref{sum},\eqref{0phase}}{=}
    \input{./figures/Cx-Xc-03.tikz}\\
    \overset{\eqref{Htotal}}{=}
    \input{./figures/Cx-Xc-04.tikz}\\
    = \input{./figures/Cx-Xc-05.tikz}\\
    \overset{\eqref{Cx-Hc}}{=}
    \input{./figures/Cx-Xc-06.tikz}\\
    = \input{./figures/Cx-Xc-07.tikz}\\
    \overset{\eqref{hch},\eqref{Hij-Hij}}{=}
    \input{./figures/Cx-Xc-08.tikz}\\
    \overset{\eqref{h-ccpt},\text{Commutativity}}{=}
    \input{./figures/Cx-Xc-09.tikz}\\
    = \input{./figures/Cx-XcR.tikz}
\end{gather*}
\end{proof}

\begin{lemma}\label{lem:rxxx-classical-factor}
On two wires we use the following three-gate \(T\)-block:
\begin{center}
\begin{tabular}{c}
 \(T_{ce}\), for \(2\le c<e<d\)\\[0.35em]
\input{./figures/rxxx-T.tikz}
\end{tabular}
\end{center}
In this lemma, every single \(c\)-frame has \(2\le c<d\), and every \(T\)-frame has
\(2\le c<e<d\).
Then the following two identities are derivable:
\refstepcounter{equation}
\begin{equation*}
\begin{gathered}
  \input{./figures/rxxx-factor-00.tikz}
  =
  \scalebox{0.55}{\input{./figures/rxxx-factor-04.tikz}}
\end{gathered}
\tag{G1a}\label{eq:rxxx-factorisation}
\end{equation*}
\refstepcounter{equation}
\begin{equation*}
\begin{aligned}
&\scalebox{0.8}{\input{./figures/rxxx-residual-slide-L.tikz}}\\[-0.2em]
={}&\scalebox{0.8}{\input{./figures/rxxx-residual-slide-R.tikz}} .
\end{aligned}
\tag{G1b}\label{eq:rxxx-residual-slide}
\end{equation*}
\end{lemma}

\begin{proof}
First we record the local \(c\)-block calculation in circuits.  The condition
\(2\le c<d\) is used throughout, so \(0,1,c\) are pairwise distinct.
\refstepcounter{equation}
\begin{align*}
\scalebox{0.60}{\input{./figures/rxxx-cb-00.tikz}}\notag
\overset{\eqref{CX-XC-CX}}{=}\;&
  \scalebox{0.60}{\input{./figures/rxxx-cb-01.tikz}}\notag\\
\overset{\eqref{Cx-Xc}}{=}\;&
  \scalebox{0.60}{\input{./figures/rxxx-cb-02.tikz}}\notag\\[-0.1em]
\overset{\ref{thm:commuting}}{=}\;&
  \scalebox{0.60}{\input{./figures/rxxx-cb-03.tikz}}\notag\\[-0.1em]
\overset{\eqref{Cx-Xc}}{=}\;&
  \scalebox{0.60}{\input{./figures/rxxx-cb-04.tikz}}\notag\\[-0.1em]
\overset{\eqref{Cx-Xc}}{=}\;&
  \scalebox{0.60}{\input{./figures/rxxx-cb-05.tikz}}\notag\\[-0.1em]
\overset{\ref{thm:commuting}}{=}\;&
  \scalebox{0.60}{\input{./figures/rxxx-cb-06.tikz}}\notag\\[-0.1em]
\overset{\eqref{Cx-Xc}}{=}\;&
  \scalebox{0.60}{\input{./figures/rxxx-cb-07.tikz}}\notag\\[-0.1em]
\overset{\ref{thm:commuting}}{=}\;&
  \scalebox{0.60}{\input{./figures/rxxx-cb-08.tikz}}\notag\\[-0.1em]
\overset{\eqref{CX-XC-CX}}{=}\;&
  \scalebox{0.60}{\input{./figures/rxxx-cb-09.tikz}}\notag\\[-0.1em]
\overset{\ctrl_c(\eqref{Xij-Xjk})}{=}\;&
  \scalebox{0.60}{\input{./figures/rxxx-cb-10.tikz}}\notag\\[-0.1em]
\overset{\eqref{CX-XC-CX}}{=}\;&
  \scalebox{0.60}{\input{./figures/rxxx-cb-11.tikz}}\notag\\[-0.1em]
\overset{\ctrl_c(\eqref{Xij-Xjk})}{=}\;&
  \scalebox{0.60}{\input{./figures/rxxx-cb-12.tikz}}\notag\\[-0.1em]
\overset{\ref{thm:commuting}}{=}\;&
  \scalebox{0.60}{\input{./figures/rxxx-cb-13.tikz}}\notag\\[-0.1em]
\overset{\eqref{Cx-Xc}}{=}\;&
  \scalebox{0.60}{\input{./figures/rxxx-cb-14.tikz}}\notag\\[-0.1em]
\overset{\ref{thm:commuting}}{=}\;&
  \scalebox{0.60}{\input{./figures/rxxx-cb-15.tikz}}\notag\\[-0.1em]
\overset{\eqref{CX-XC-CX}}{=}\;&
  \scalebox{0.60}{\input{./figures/rxxx-cb-16.tikz}}\notag\\[-0.1em]
\overset{\eqref{CX-XC-CX}}{=}\;&
  \scalebox{0.60}{\input{./figures/rxxx-cb-17.tikz}}\notag\\[-0.1em]
\overset{\ref{thm:commuting}}{=}\;&
  \scalebox{0.60}{\input{./figures/rxxx-cb-18.tikz}}\notag\\[-0.1em]
\overset{\eqref{Cx-Xc}}{=}\;&
  \scalebox{0.60}{\input{./figures/rxxx-cb-19.tikz}}\notag\\[-0.1em]
\overset{\ref{thm:commuting}}{=}\;&
  \scalebox{0.60}{\input{./figures/rxxx-cb-20.tikz}}\notag\\[-0.1em]
\overset{\eqref{Cx-Xc}}{=}\;&
  \scalebox{0.60}{\input{./figures/rxxx-cb-21.tikz}}\notag\\[-0.1em]
\overset{\eqref{CX-XC-CX}}{=}\;&
  \scalebox{0.60}{\input{./figures/rxxx-cb-22.tikz}}\notag\\[-0.1em]
\overset{\ctrl_c(\eqref{Xij-Xjk})}{=}\;&
  \scalebox{0.60}{\input{./figures/rxxx-cb-23.tikz}}\notag\\[-0.1em]
\overset{\ctrl_c(\eqref{Xij-Xij})}{=}\;&
  \scalebox{0.60}{\input{./figures/rxxx-cb-24.tikz}}.
\tag{CB}\label{eq:rxxx-c-block}
\end{align*}

For \eqref{eq:rxxx-factorisation}, use naturality of the structural symmetry
before expanding anything.  This avoids regrouping by interpretation:
\(\sigma_{1,1}\) carries the final \(X^{01}\) on the first wire to the same
\(X^{01}\) on the second wire.

\begin{align*}
  &\input{./figures/rxxx-factor-00.tikz}
  \overset{\text{Def.~\ref{def:structural-congruence-CQC_d}}}{=}
  \input{./figures/rxxx-factor-nat-00.tikz}
  \overset{\text{Thm.~\ref{thm:exhaustive}}}{=}
  \scalebox{0.86}{\input{./figures/rxxx-factor-nat-01.tikz}}\\
  &\overset{\eqref{swap-decomp}}{=}
  \scalebox{0.88}{\input{./figures/rxxx-factor-nat-02.tikz}}\\
  &\overset{\scriptscriptstyle\eqref{Cx-Xc},\,\ref{thm:commuting}}{=}
  \scalebox{0.66}{\input{./figures/rxxx-factor-nat-03.tikz}}\\
  &\overset{\scriptscriptstyle\eqref{Xij-Xij}}{=}
  \scalebox{0.63}{\input{./figures/rxxx-factor-nat-04.tikz}}\\
  &\overset{\eqref{eq:rxxx-c-block}}{=}
  \scalebox{0.68}{\input{./figures/rxxx-factor-04.tikz}} .
\end{align*}
This proves \eqref{eq:rxxx-factorisation}.

\begin{equation*}
\begin{aligned}
&\scalebox{0.88}{\input{./figures/rxxx-slide-c-0.tikz}}\\
\overset{\ref{thm:commuting}}{=}&
  \scalebox{0.88}{\input{./figures/rxxx-slide-c-1.tikz}}\\
\overset{\eqref{Cx-Hc}}{=}&
  \scalebox{0.88}{\input{./figures/rxxx-slide-c-2.tikz}}\\
\overset{\ref{thm:commuting}}{=}&
  \scalebox{0.88}{\input{./figures/rxxx-slide-c-3.tikz}}\\
\overset{\eqref{Cx-Hc}}{=}&
  \scalebox{0.88}{\input{./figures/rxxx-slide-c-4.tikz}}\\
\overset{\ref{thm:commuting}}{=}&
  \scalebox{0.88}{\input{./figures/rxxx-slide-c-5.tikz}} .
\end{aligned}
\end{equation*}
When \(2\le c<e<d\), every level touched by \(T_{ce}\) is outside
\(\{0,1\}\), so the three local moves are ordinary commutations.  The rotation
is moved past \(C^c(X^{ce})\), then \(C_c(X^{ce})\), then \(C^c(X^{ce})\);
the control labels \(c,e\) are never \(0\) or \(1\):
\begin{equation*}
\begin{aligned}
&\input{./figures/rxxx-slide-t-0.tikz}\\
\overset{\ref{thm:commuting}}{=}&
  \input{./figures/rxxx-slide-t-1.tikz}\\
\overset{\ref{thm:commuting}}{=}&
  \input{./figures/rxxx-slide-t-2.tikz}\\
\overset{\ref{thm:commuting}}{=}&
  \input{./figures/rxxx-slide-t-3.tikz}.
\end{aligned}
\end{equation*}
Applying these local slides inside the product frames, in their increasing
lexicographic order, gives \eqref{eq:rxxx-residual-slide}.
The two steps labelled by \eqref{Cx-Hc} are its vertical reflections: a
first-wire controlled \(R_X^{01}(\theta)\) commutes past a second-wire
controlled \(X\), with the support side condition supplied by the distinct
levels \(0,1,c\).
\end{proof}

\begin{proof}[Proof of \eqref{RxXX}]
Since the rule is schematic in the two levels, and using \eqref{RxSym} to fix
the orientation, it suffices to prove the instance with levels \(0,1\).  Let
\(\mathcal R\) be the residual circuit
\[
  C^0(X^{01});\prod_c
  \bigl(C^c(X^{0c});C_c(X^{1c});C^c(X^{1c});C_c(X^{0c});C^c(X^{0c})\bigr);
  \prod_{c<e}T_{ce},
\]
namely the part of the right-hand side of \eqref{eq:rxxx-factorisation} after
\(C^1(X^{01});C_0(X^{01})\).  This residual is reversible and equal to its own
inverse read backwards.  The derivation below starts with
\(C_0(R_X^{01}(\theta));C^1(X^{01});C_0(X^{01})\), right-composes with
\(\mathcal R\), moves through the factorisation and naturality steps, and then
cancels the same residual at the end.  The fourth equality uses
\eqref{eq:rxxx-residual-slide} to move the rotation left through exactly
\(\mathcal R\).  The first and last arrows below are not standalone rewrites:
they indicate post-composition with \(\mathcal R\) and, after the displayed
equalities, post-composition with \(\mathcal R^{-1}\):
\begin{gather*}
  \input{./figures/rxxx-rx-left.tikz}\\
  \overset{\text{right-compose with }\mathcal R}{\Longrightarrow}
  \scalebox{0.72}{\input{./figures/rxxx-rx-left-appended.tikz}}\\
  \overset{\eqref{eq:rxxx-factorisation}}{=}
  \input{./figures/rxxx-rx-AF.tikz}\\
  \overset{\textsc{Nat}_{\sigma},\,\textsc{CtrlNat}}{=}
  \input{./figures/rxxx-rx-FB.tikz}\\
  \overset{\eqref{eq:rxxx-factorisation}}{=}
  \scalebox{0.72}{\input{./figures/rxxx-rx-right-appended.tikz}}\\
  \overset{\eqref{eq:rxxx-residual-slide}}{=}
  \scalebox{0.72}{\input{./figures/rxxx-rx-right-slid.tikz}}\\
  \overset{\text{cancel }\mathcal R}{\Longrightarrow}
  \input{./figures/rxxx-rx-right.tikz}.
\end{gather*}
The cancellation step is composing by \(\mathcal R^{-1}\) on the right and
using \eqref{Xij-Xij} factor by factor, in reverse order because circuits are
read left to right.  This is the \(0,1\) instance of \eqref{RxXX};
relabelling the two levels gives
the stated rule.
\end{proof}

\begin{proof}[Proof of \eqref{3CRx}]
\smallskip
\noindent\emph{Auxiliary equality \eqref{eq:3CRx-PA}.}
\refstepcounter{equation}
\begin{equation*}
  \scalebox{0.9}{\input{./figures/3CRx-PA.tikz}}
  =
  \scalebox{0.9}{\input{./figures/3CRx-AP.tikz}}
  \tag{PA}\label{eq:3CRx-PA}
\end{equation*}
In this derivation, the two \(\textsc{Nat}_{\sigma}\) steps are conjugations
through a wire swap: the controlled rotation changes orientation when it passes
through the swap, becoming a lower-controlled, upper-target gate, and changes
back at the second swap.  The two intermediate diagrams display this flipped
orientation explicitly; the middle commutativity step then has disjoint
support.
\par\smallskip
\noindent\emph{Derivation.}
\begin{gather*}
  \input{./figures/3CRx-PA.tikz}\\
    \overset{\eqref{RxXX},\,\eqref{RxSym}}{=}
  \input{./figures/3CRx-PA-01.tikz}\\
    \overset{\textsc{Nat}_{\sigma}}{=}
  \input{./figures/3CRx-PA-01a.tikz}\\
    \overset{\text{Commutativity}}{=}
  \input{./figures/3CRx-PA-01b.tikz}\\
    \overset{\textsc{Nat}_{\sigma}}{=}
  \input{./figures/3CRx-PA-02.tikz}\\
    \overset{\eqref{RxXX},\,\eqref{RxSym}}{=}
  \input{./figures/3CRx-AP.tikz}.
\end{gather*}

\smallskip
\noindent\emph{Auxiliary equality \eqref{eq:3CRx-BQ}.}
\refstepcounter{equation}
\begin{equation*}
  \input{./figures/3CRx-BQ.tikz}
  =
  \input{./figures/3CRx-QB0.tikz}
  \tag{BQ}\label{eq:3CRx-BQ}
\end{equation*}
\noindent\emph{Derivation.}
\begin{gather*}
  \input{./figures/3CRx-BQ.tikz}\\
    \overset{\textsc{Nat}_{\sigma}}{=}
  \input{./figures/3CRx-BQ-01.tikz}\\
    \overset{\eqref{RxXX},\,\eqref{RxSym}}{=}
  \input{./figures/3CRx-QB0-01.tikz}\\
    \overset{\textsc{Nat}_{\sigma}}{=}
  \input{./figures/3CRx-QB0.tikz}.
\end{gather*}

\smallskip
\noindent\emph{Auxiliary equality \eqref{eq:3CRx-PB}.}
\refstepcounter{equation}
\begin{equation*}
  \scalebox{0.84}{\input{./figures/3CRx-PB0P.tikz}}
  =
  \scalebox{0.84}{\input{./figures/3CRx-Btheta.tikz}}
  \tag{PB}\label{eq:3CRx-PB}
\end{equation*}
\noindent\emph{Derivation.}
\begin{gather*}
  \input{./figures/3CRx-PB0P.tikz}\\
    \overset{\substack{\eqref{XCtrl},\,\eqref{Rxtotal}\\
      \eqref{Hij-Hij},\,\text{Commutativity}}}{=}
  \input{./figures/3CRx-QB0Q.tikz}\\
    \overset{\eqref{eq:3CRx-BQ}}{=}
  \input{./figures/3CRx-BQQ.tikz}\\
    \overset{\sigma^2=\id}{=}
  \input{./figures/3CRx-BQQL.tikz}\\
    \overset{\eqref{Xij-Xij}}{=}
  \input{./figures/3CRx-BSSG.tikz}\\
    \overset{\sigma^2=\id}{=}
  \input{./figures/3CRx-BGG.tikz}\\
    \overset{\eqref{Xij-Xij}}{=}
  \input{./figures/3CRx-Btheta.tikz}.
\end{gather*}

Apply \eqref{eq:3CRx-PA} and \eqref{eq:3CRx-PB} to expose the uncontrolled
three-rotation pattern, use \eqref{3Rx} together with the symmetry of \(R_X\),
and then reverse the same auxiliary equalities.
\begin{gather*}
  \input{./figures/3bsaltL.tikz}\\
    \overset{\eqref{eq:3CRx-PA},\,\eqref{eq:3CRx-PB}}{=}
  \scalebox{0.64}{\input{./figures/3CRx-01.tikz}}\\
    \overset{\eqref{3Rx},\,\eqref{RxSym}}{=}
  \scalebox{0.64}{\input{./figures/3CRx-02.tikz}}\\
    \overset{\eqref{eq:3CRx-PA},\,\eqref{eq:3CRx-PB}}{=}
  \input{./figures/3bsaltR.tikz}.
\end{gather*}
\end{proof}

The following \(\pi\)-normalisation lemma is used in the proof of
\eqref{eulerb}.
\label{app:pi-normalisation}

We introduce an auxiliary real parameter \(x\)
(by splitting a middle rotation/phase into two pieces \(x\) and \(\alpha_2-x\)).
This produces a \emph{family} of circuit identities indexed by \(x\), in which
the angles \(\beta_i(x)\) and \(\gamma_i(x)\) produced by the Euler-extraction
depend on \(x\) (while the input angles \(\alpha_1,\alpha_2,\alpha_3\) are fixed).

To continue the purely diagrammatic part of the derivation, we need one
additional property: we can choose a value \(x_0\) and admissible representatives
for the two \eqref{eulerH} decompositions such that the \emph{bridge phase}
\(\beta_3(x_0)+\gamma_0(x_0)\) is a multiple of \(\pi\).
Since angles are always understood modulo \(2\pi\),
\(\beta_3(x_0)+\gamma_0(x_0)\) is either \(0\) or \(\pi\) modulo \(2\pi\). These
are the two cases treated in the proof of \eqref{eulerb}.

The following analytic lemma is the normalisation argument needed for the Euler splitter.
\begin{lemma}
\label{lem:pi-normalised-split}
Fix \(\alpha_1,\alpha_2,\alpha_3\in\mathbb{R}\) and consider the continuous function
\(N:[-\frac{\pi}{2},\frac{\pi}{2}]\to\mathbb{R}\) defined by
\(N(x):=\sin(\alpha_1)\cos(x)\cos(\alpha_3)
+\cos(\alpha_1)\sin(\alpha_3)\cos(\alpha_2-x)\).
Then there exists \(x_0\in[-\frac{\pi}{2},\frac{\pi}{2}]\) such that \(N(x_0)=0\).

In the Euler-extraction step used in the proof of \eqref{eulerb}, the two \eqref{eulerH}
instances can be chosen from their admissible parameter families so that the
corresponding bridge phase satisfies
\(\beta_3(x_0)+\gamma_0(x_0)\in\pi\mathbb{Z}\); hence (modulo \(2\pi\)) we necessarily have
\(\beta_3(x_0)+\gamma_0(x_0)\equiv 0\) or \(\equiv \pi\).
\end{lemma}

\begin{proof}
First note that
\(N(-\pi/2)=
\cos(\alpha_1)\sin(\alpha_3)\cos(\alpha_2+\pi/2)
=-\cos(\alpha_1)\sin(\alpha_3)\sin(\alpha_2)\), while
\(N(\pi/2)=
\cos(\alpha_1)\sin(\alpha_3)\cos(\alpha_2-\pi/2)
=\cos(\alpha_1)\sin(\alpha_3)\sin(\alpha_2)\).
Hence \(N(-\frac{\pi}{2})=-N(\frac{\pi}{2})\), so
\(N(-\frac{\pi}{2})\cdot N(\frac{\pi}{2})\le 0\).
By continuity of \(N\), the intermediate value theorem yields
\(x_0\in[-\frac{\pi}{2},\frac{\pi}{2}]\) with \(N(x_0)=0\).

For the second part, use the notation of Appendix~\ref{appendix:relations_angles}
for the two Euler extractions.  Let \(z_\beta,z'_\beta\) be the quantities
\(z,z'\) for the first extraction, with input angles \((-\alpha_1,x)\), and
let \(z_\gamma,z'_\gamma\) be the corresponding quantities for the second
extraction, with input angles \((\alpha_2-x,-\alpha_3)\).  Put
\(Z(x):=
\bigl(\sin(\alpha_1)\cos(x)+i\cos(\alpha_1)\bigr)
\bigl(\sin(\alpha_3)\cos(\alpha_2-x)+i\cos(\alpha_3)\bigr)\).
In the generic case where none of
\(z_\beta,z'_\beta,z_\gamma,z'_\gamma\) vanishes, the formulas for \eqref{eulerH} give
\begin{equation*}
  e^{i(\beta_3(x)+\gamma_0(x))}
  =
  \frac{
    z_\beta\,\overline{z'_\beta}\,z_\gamma\,z'_\gamma
  }{
    \left|z_\beta\,\overline{z'_\beta}\,z_\gamma\,z'_\gamma\right|
  }.
\end{equation*}
The two products in the numerator reduce to
\begin{equation*}
\begin{aligned}
  z_\beta\,\overline{z'_\beta}
    &= \sin(\alpha_1)\cos(x)+i\cos(\alpha_1),\\
  z_\gamma\,z'_\gamma
    &= \sin(\alpha_3)\cos(\alpha_2-x)+i\cos(\alpha_3).
\end{aligned}
\end{equation*}
Thus the numerator is \(Z(x)\), and
\begin{equation*}
  \tan(\beta_3(x)+\gamma_0(x))
  =
  \frac{
    \sin(\alpha_1)\cos(x)\cos(\alpha_3)
    +\cos(\alpha_1)\sin(\alpha_3)\cos(\alpha_2-x)
  }{
    \sin(\alpha_1)\sin(\alpha_3)\cos(x)\cos(\alpha_2-x)
    -\cos(\alpha_1)\cos(\alpha_3)
  }.
\end{equation*}
The numerator is \(N(x)\).  If \(Z(x_0)\neq0\), then \(N(x_0)=0\)
implies that \(Z(x_0)\) is a non-zero real number.  Hence
\(Z(x_0)/|Z(x_0)|\in\{1,-1\}\), so the bridge phase is a multiple of \(\pi\).

It remains to consider the case \(Z(x_0)=0\).  Then at least one of
\(z_\beta,z'_\beta,z_\gamma,z'_\gamma\) is zero, so at least one of the two
\eqref{eulerH} extractions is degenerate.  By the admissible degenerate convention of
Appendix~\ref{appendix:relations_angles}, the endpoint phase \(\beta_3\) in a
degenerate first extraction, or the endpoint phase \(\gamma_0\) in a degenerate
second extraction, may be prescribed arbitrarily while staying in the same
sound \eqref{eulerH} parameter family.  Choose that free endpoint phase so that
\(\beta_3(x_0)+\gamma_0(x_0)\in\pi\mathbb Z\).  Reducing modulo \(2\pi\) gives
the claimed dichotomy \(\equiv 0\) or \(\equiv\pi\).
\end{proof}

\begin{remark}
The lemma gives \emph{existence} (not uniqueness) of a suitable
\(x_0\), and any such choice is valid for the subsequent diagrammatic
derivation.  Once \(x_0\) is fixed, the remaining steps use the equational
rules of \(\QCeq\).
\end{remark}

\begin{proof}[Proof of \eqref{eulerb}]
\begin{gather*}
  \scalebox{0.65}{\input{./figures/eulerbL.tikz}}\\
    \overset{\eqref{sum}}{=}
    \scalebox{0.65}{\input{./figures/eulerbF-01.tikz}}\\
    = \scalebox{0.65}{\input{./figures/eulerbF-02.tikz}}\\
    \overset{\text{Commutativity},\eqref{sum},\eqref{Htotal}}{=}
    \scalebox{0.65}{\input{./figures/eulerbF-03.tikz}}\\
    \overset{\eqref{Hij-Hij}}{=}
    \scalebox{0.65}{\input{./figures/eulerbF-04.tikz}}\\
    \overset{\eqref{eulerH}}{=}
    \scalebox{0.65}{\input{./figures/eulerbF-05.tikz}}\\
    \overset{\eqref{sum}}{=}
    \scalebox{0.65}{\input{./figures/eulerbF-06.tikz}}\\
    \overset{\text{Commutativity},\eqref{sum},\eqref{Htotal}}{=}
    \scalebox{0.65}{\input{./figures/eulerbF-07.tikz}}
\end{gather*}

Choose \(x_0\) and the admissible \eqref{eulerH} representatives as in
Lemma~\ref{lem:pi-normalised-split}, so that
\(\beta_3(x_0)+\gamma_0(x_0)\in\pi\mathbb{Z}\).
Working modulo \(2\pi\), there are two cases:
(i) \(\beta_3(x_0)+\gamma_0(x_0)\equiv 0\), and
(ii) \(\beta_3(x_0)+\gamma_0(x_0)\equiv \pi\).
We treat these cases separately below.

\noindent\emph{Case \(\beta_3+\gamma_0=0\).}
\begin{gather*}
  \input{./figures/eulerbF-08-01.tikz}\\
  \overset{\eqref{0phase}}{=}
  \input{./figures/eulerbF-08-02.tikz}\\
  \overset{\eqref{Hij-Hij}}{=}
  \input{./figures/eulerbF-08-03.tikz}\\
  \overset{\text{Commutativity},\eqref{sum}}{=}
  \input{./figures/eulerbF-08-04.tikz}\\
  = \input{./figures/eulerbF-08-05.tikz}
\end{gather*}

\noindent\emph{Case \(\beta_3+\gamma_0=\pi\).}
\begin{gather*}
  \input{./figures/eulerbF-09-01.tikz}\\
  = \input{./figures/eulerbF-09-02.tikz}\\
  \overset{\eqref{XCtrl}}{=}
  \input{./figures/eulerbF-09-03.tikz}\\
  \overset{\text{Commutativity},\eqref{sum}}{=}
  \input{./figures/eulerbF-09-05.tikz}\\
  \overset{\eqref{Hij-Hij}}{=}
  \input{./figures/eulerbF-09-06.tikz}\\
  \overset{\text{Commutativity},\eqref{sum}}{=}
  \input{./figures/eulerbF-09-07.tikz}\\
  = \input{./figures/eulerbF-09-08.tikz}
\end{gather*}
\end{proof}

\begin{proof}[Proof of \eqref{euler}]
\begin{gather*}
  \input{./figures/eulerL.tikz}\\
    \overset{\eqref{Rxtotal}}{=}
    \input{./figures/eulerF-01.tikz}\\
    \overset{\eqref{eulerb}}{=}
    \input{./figures/eulerF-02.tikz}\\
    \overset{\text{Commutativity},\eqref{sum}}{=}
    \input{./figures/eulerR.tikz}
\end{gather*}
\end{proof}

%% file: appendix_gray.tex
\section{Details of the Gray-code lifts}\label{app:encoding-details}

The encoding of \cref{subsec:encoding} uses two raw-\(\LOPP\) constructions:
\begin{enumerate}
  \item a \(d\)-mode gadget \(B_d^{(i,i+1)}\) whose single-photon semantics is the adjacent two-level Hadamard \(H_d^{(i,i+1)}\);
\item the lift operators \(\Lift_d^p(-)\) and \(\Lift_{d,j}^p(-)\), which replicate a local subcircuit across Gray blocks and mirror it on the reversed branches.
\end{enumerate}
This appendix fixes explicit choices and the block combinatorics used later by
decoding.

\subsection{\texorpdfstring{A single-qudit network for \(H_d^{(i,i+1)}\)}{A single-qudit network for Hd(i,i+1)}}

Fix \(d \ge 2\) and \(i \in \{0,\dots,d-2\}\).  We need a \(d\)-mode
\(\LOPP\) circuit whose single-photon semantics mixes levels \(i\) and \(i+1\)
and fixes the others.

Let \(H := \frac{1}{\sqrt{2}}\begin{psmallmatrix} 1 & 1\\ 1 & -1\end{psmallmatrix}\in U(2)\)
be the usual \(2\times 2\) Hadamard, and set
\(H_d^{(i,i+1)} := I_i \oplus H \oplus I_{d-i-2}\in U(d)\).

\begin{definition}\label{def:hadamard-bs-gadget}
Write \(P(\phi):1\to 1\) for the \(\LOPP\) phase-shifter generator of angle
\(\phi\), and \(B(\theta):2\to 2\) for the \(\LOPP\) beam-splitter generator of
angle \(\theta\) (as in \cref{def:lopp-semantics-sp}).  Set
\(H_{\mathrm{BS}}:=(\id_1\otimes P(-\pi/2))\circ B(\pi/4)\circ
(\id_1\otimes P(-\pi/2)):2\to 2\).  The \(d\)-mode gadget placing this
interferometer on modes \(i,i+1\) is
\(B_d^{(i,i+1)}:=\id_i\otimes H_{\mathrm{BS}}\otimes \id_{d-i-2}:d\to d\).
\end{definition}

Here \(B(\pi/4)\) is the standard \(50/50\) beam splitter in our phase
convention, and the surrounding phase shifters correct the relative phases so
that the induced \(2\times2\) matrix is \(H\).

\begin{lemma}\label{lem:Hd-singlequdit-network}
For each \(d\ge 2\) and \(0\le i\le d-2\), the circuit \(B_d^{(i,i+1)}\) is a raw \(\LOPP\) term built only from phase shifters, beam splitters, and identities, and its single-photon semantics satisfies
\(\interp{B_d^{(i,i+1)}}_{\mathrm{sp}} = H_d^{(i,i+1)}\in U(d)\).
\end{lemma}

\begin{proof}
Under \(\interp{-}_{\mathrm{sp}}\), the phase shifter \(P(\phi)\) acts as \([e^{\mathrm{i}\phi}]\) on the corresponding mode, so
\(\interp{\id_1\otimes P(\phi)}_{\mathrm{sp}}=\mathrm{diag}(1,e^{\mathrm{i}\phi})\) on \(2\) modes.
The beam splitter has matrix \(\interp{B(\theta)}_{\mathrm{sp}}=\begin{psmallmatrix}\cos(\theta) & \mathrm{i}\sin(\theta)\\ \mathrm{i}\sin(\theta) & \cos(\theta)\end{psmallmatrix}\).
For \(\theta=\pi/4\) and \(\phi=-\pi/2\) we have \(\cos(\pi/4)=\sin(\pi/4)=1/\sqrt{2}\) and \(e^{\mathrm{i}\phi}=e^{-\mathrm{i}\pi/2}=-\mathrm{i}\), so
\(\interp{H_{\mathrm{BS}}}_{\mathrm{sp}}=\mathrm{diag}(1,-\mathrm{i})\cdot \frac{1}{\sqrt{2}}\begin{psmallmatrix}1 & \mathrm{i}\\ \mathrm{i} & 1\end{psmallmatrix}\cdot \mathrm{diag}(1,-\mathrm{i})=\frac{1}{\sqrt{2}}\begin{psmallmatrix}1 & 1\\ 1 & -1\end{psmallmatrix}=H\).
Since parallel composition in \(\LOPP\) is interpreted as block-diagonal direct sum on disjoint mode blocks, we have
\(\interp{\id_i\otimes H_{\mathrm{BS}}\otimes \id_{d-i-2}}_{\mathrm{sp}}=I_i\oplus H \oplus I_{d-i-2}=H_d^{(i,i+1)}\).
\end{proof}

\subsection{Alternating mirror symmetry on Gray-coded modes}

The lift operators of \cref{subsec:encoding} replicate a local \(\LOPP\)
circuit across Gray-code blocks.  The needed combinatorial fact is that, for a
fixed prefix \(u\), the words \(u0,\dots,u(d-1)\) form one contiguous Gray
block, read either forward or backward according to the reflected recursion.
The lifts account for this by mirroring on the reversed branches.

\begin{definition}\label{def:mode-reversal}
For each \(r\ge 0\), let \(\rho_r : 1^{\otimes d^r} \longrightarrow 1^{\otimes d^r}\)
be the \(\LOPP\) permutation circuit that reverses the order of the \(d^r\) modes:
on the single-photon basis \(\{\ket{0},\dots,\ket{d^r-1}\}\) it acts by
\(\rho_r\ket{k}=\ket{d^r-1-k}\).

For any raw \(\LOPP\) circuit \(C:1^{\otimes d^r}\to 1^{\otimes d^r}\), define
its mirror by \(\Rev(C) := \rho_r \circ C \circ \rho_r^{-1}\).  Since
\(\rho_r\) is involutive, \(\rho_r^{-1}=\rho_r\).
\end{definition}

For \(q\in[d]\), set \(\Rev^{q}(C):=C\) if \(q\) is even, and \(\Rev^{q}(C):=\Rev(C)\) if \(q\) is odd.
This parity test matches the reflected recursion in the Gray code: odd leading digits reverse the traversal of the suffix.

\begin{definition}\label{def:gray-lifts}
Fix \(d\ge 2\), let \(m\ge 0\), and let
\(C:1^{\otimes d^m}\to 1^{\otimes d^m}\) be a raw \(\LOPP\) circuit.
\begin{itemize}
\item The ordinary lift is defined by
\(\Lift_d^{0}(C):=C\) and
\(\Lift_d^{p+1}(C):=\bigotimes_{q=0}^{d-1}\Rev^{q}(\Lift_d^{p}(C))\).
It has arity \(d^{m+p}\).
\item For \(j\in[d]\), set \(C^{(q)}:=C\) if \(q=j\) and
\(C^{(q)}:=\id_{d^m}\) otherwise.  The selected lift is defined by
\(\Lift_{d,j}^{0}(C):=\bigotimes_{q=0}^{d-1}\Rev^{q}(C^{(q)})\) and
\(\Lift_{d,j}^{p+1}(C):=\bigotimes_{q=0}^{d-1}
\Rev^{q}(\Lift_{d,j}^{p}(C))\).  It has arity \(d^{m+1+p}\).
\end{itemize}
\end{definition}

At level \(p+1\), the modes split into \(d\) consecutive blocks of size
\(d^{m+p}\); the previous lift is placed on each block and mirrored when the
new leading digit is odd.  The selected lift does the same, but
starts from one chosen copy of \(C\) and identities elsewhere.

The next lemma isolates the combinatorial feature that makes the lifts correct: fixing a prefix produces a contiguous block of indices, and within that block the last digit runs either forward or backward.

\begin{lemma}\label{lem:gray-blocks}
Fix \(d\ge 2\) and \(n\ge 0\).
For each word \(u\in[d]^n\), there exists an index \(s(u)\) with \(0\le s(u)\le d^{n+1}-d\) and a sign \(\varepsilon(u)\in\{+1,-1\}\) such that
\(\{\,t\in\{0,\dots,d^{n+1}-1\}\mid G_{n+1}^d(t)\text{ has prefix }u\,\}=\{s(u),s(u)+1,\dots,s(u)+d-1\}\),
and for every \(r\in\{0,\dots,d-1\}\), \(G_{n+1}^d\bigl(s(u)+r\bigr)=u\,r\) if \(\varepsilon(u)=+1\), while \(G_{n+1}^d\bigl(s(u)+r\bigr)=u\,(d-1-r)\) if \(\varepsilon(u)=-1\),
where \(u\,r\) denotes concatenation of the word \(u\) with the digit \(r\).

One may take \(\varepsilon\) to be the unique function \(\varepsilon:[d]^*\to\{\pm1\}\) defined by \(\varepsilon(\epsilon)=+1\) and
\(\varepsilon(qu)=\varepsilon(u)\) if \(q\) is even and \(\varepsilon(qu)=-\varepsilon(u)\) if \(q\) is odd.
In particular, \(\varepsilon(u)\) flips sign exactly once for each odd digit encountered when reading \(u\) from left to right.
\end{lemma}

\begin{proof}
We prove the block statement and the recursion for \(\varepsilon\) simultaneously by induction on \(n\).

Base case \(n=0\):
we have \(u=\epsilon\) and \(G_1^d(t)=(t)\) for \(0\le t<d\), so we may take \(s(\epsilon)=0\) and \(\varepsilon(\epsilon)=+1\).

Inductive step:
let \(n\ge 1\) and write \(u=qv\) with \(q\in[d]\) and \(v\in[d]^{n-1}\).
Any \(t\in\{0,\dots,d^{n+1}-1\}\) can be written uniquely as \(t=q\cdot d^{\,n}+r\) with \(0\le r<d^n\).

If \(q\) is even, then by \cref{def:gray-code} we have
\(G_{n+1}^d(qd^n+r)=q\cdot G_n^d(r)\).
Thus \(G_{n+1}^d(qd^n+r)\) has prefix \(qv\) if and only if \(G_n^d(r)\) has prefix \(v\).
By the induction hypothesis, the set of such \(r\) is an interval of length \(d\), hence so is the set of such \(t\), with the same last-digit direction.
This corresponds to setting \(\varepsilon(qv)=\varepsilon(v)\).

If \(q\) is odd, then \(G_{n+1}^d(qd^n+r)=q\cdot G_n^d(d^n-1-r)\).
Put \(r':=d^n-1-r\).
By the induction hypothesis, the set of \(r'\) for which \(G_n^d(r')\) has prefix \(v\) is an interval \([a,a+d-1]\).
Its image under \(r=d^n-1-r'\) is the interval \([d^n-d-a,\,d^n-1-a]\), again of length \(d\).
Since the map \(r'\mapsto d^n-1-r'\) reverses order, the last digit is traversed in the opposite direction on this block; this corresponds to setting \(\varepsilon(qv)=-\varepsilon(v)\).
\end{proof}

\Cref{lem:gray-blocks} supplies the fact used by the lift recursion: each fixed
prefix determines one contiguous mode block, and the last digit runs through
that block either forward or backward.  Mirroring on the odd branches restores
the same logical orientation in every block.

%% file: appendix_encodingdecoding.tex
\section{Encoding and decoding}\label{app:encoding-decoding}

Theorem~\ref{thm:encoding-decoding} states that for every \(C:n\to n\) in
\(\CQC\) and every \(k,\ell\ge 0\),
\[
  \QCeq\vdash \mathrm{D}^{0}_{k+n+\ell}(E^{k}_{\ell}(C))
  =
  \id_k\otimes C\otimes \id_\ell.
\]
Here \(E^k_\ell\) is the contextual encoding of
Definition~\ref{def:encoding}, and \(\mathrm D_n^t\) is the contextual
decoding of Definition~\ref{def:decoding}.  Because \(\otimes\) is tensor on
qudits but direct sum on optics, both translations are context-sensitive and
serialise tensor products.  The proof identifies the behaviour of
\(\mathrm D_n^t\) on one Gray block, applies the same calculation to the lift
gadgets, and closes by structural induction on \(\CQC\).

\subsection{Decoding a fixed Gray-code block}

Fix \(a\ge 0\).  In reflected \(d\)-ary Gray order, the \(d^{a+1}\) modes split
into \(d\) consecutive blocks of length \(d^a\), indexed by the leading Gray
digit.  On an even block, decoding is just the \(a\)-qudit decoding with an
outer control; on an odd block, the same holds after mirroring the local
circuit.

\begin{lemma}\label{lem:dec-ctrl-even}
Let \(C\) be a \(\LOPP\) circuit acting on \(\ell\) consecutive modes, and fix
\(a \ge 0\) and \(s\) with \(s+\ell \le d^{a}\).  For every even
\(j \in \{0,\dots,d-1\}\) we have \(\mathrm{D}^{d^{a}j + s}_{a+1}(C)=\ctrl_{j}\bigl(\mathrm{D}^{s}_{a}(C)\bigr)\).
\end{lemma}

\begin{proof}
We induct on the raw syntax of \(C\).  For even \(j\),
\(G^d_{a+1}(d^aj+t)=j\cdot G^d_a(t)\) throughout the block.  Hence the empty
diagram and identity decode to \(\id_{a+1}=\ctrl_j(\id_a)\), and a phase
shifter at local offset \(t\) decodes to
\(\ctrl_j\bigl(\mathrm D_a^t(\input{./figures/lov-phase.tikz})\bigr)\).  The same
factorisation holds for a beam splitter or swap on \(t,t+1\): the two Gray
labels keep the same leading digit \(j\) and differ only in the suffix part,
so decoding produces the \(a\)-qudit decoding with an extra outer
\(j\)-control.

Sequential composition is immediate because \(\mathrm D\) is compositional on
\(\circ\).  For \(C=C_1\otimes C_2\), with \(C_1\) on \(\ell_1\) modes, the
serialisation clause gives
\[
  \mathrm D_{a+1}^{d^aj+s}(C)
  =
  \mathrm D_{a+1}^{d^aj+s+\ell_1}(C_2)\circ
  \mathrm D_{a+1}^{d^aj+s}(C_1),
\]
and the induction hypotheses on \(C_1\) and \(C_2\) turn this into
\(\ctrl_j(\mathrm D_a^s(C))\) by functoriality of \(\ctrl_j\).
\end{proof}

\begin{lemma}\label{lem:dec-ctrl-odd}
Let \(C\) be a circuit of \(\LOPP\) acting on \(\ell\) consecutive modes, and fix
\(a \ge 0\) and \(s\) with \(s + \ell \le d^a\).  For every odd
\(j \in \{0,\dots,d-1\}\) we have \(\mathrm{D}^{d^{a}j+s}_{a+1}(C)=\ctrl_{j}\bigl(\mathrm{D}^{d^{a}-s-\ell}_{a}(\mathrm{Rev}_{\ell}(C))\bigr)\), where \(\mathrm{Rev}_{\ell}(C)\) denotes the circuit obtained from \(C\) by reversing the order of its \(\ell\) local modes.  In the full-block case \(\ell=d^a\), this is the mirror \(\Rev(C)\) of Definition~\ref{def:mode-reversal}.
\end{lemma}

\begin{proof}
For odd \(j\), the block is traversed in reverse:
\(G^d_{a+1}(d^aj+t)=j\cdot G^d_a(d^a-1-t)\).  Thus a local interval of length
\(\ell\) starting at \(s\) is read as the mirrored interval starting at
\(d^a-s-\ell\).  Replacing \(C\) by \(\mathrm{Rev}_\ell(C)\) compensates for
that reversal at the syntactic level, so the generator cases reduce to
Lemma~\ref{lem:dec-ctrl-even} at the mirrored offset.  For swaps, the
reflected orientation is harmless because \(X^{(r,s)}=X^{(s,r)}\).  The
composition and tensor cases are then identical to the even-block proof.
\end{proof}

The two orientation cases can be packaged as one decoding rule.

\begin{corollary}\label{lem:dec-ctrl-block}
Let \(C\) be a \(\LOPP\) circuit acting on \(\ell\) consecutive modes, and fix
\(a\ge0\), \(j\in[d]\), and \(s\) with \(s+\ell\le d^a\).  Define
\(s_j(s,\ell)\) to be \(s\) if \(j\) is even and \(d^a-s-\ell\) if \(j\) is
odd; define \(R_j^\ell(C)\) to be \(C\) if \(j\) is even and
\(\mathrm{Rev}_{\ell}(C)\) if \(j\) is odd.  Then
\(\mathrm{D}^{d^a j+s}_{a+1}(C)=\ctrl_j\bigl(\mathrm{D}^{s_j(s,\ell)}_a(R_j^\ell(C))\bigr)\).
\end{corollary}

\begin{proof}
If \(j\) is even this is Lemma~\ref{lem:dec-ctrl-even}; if \(j\) is odd this
is Lemma~\ref{lem:dec-ctrl-odd}.
\end{proof}

\begin{lemma}\label{lem:dec-Hd-gadget}
For every \(d\ge2\) and \(0\le i\le d-2\),
\(\mathrm{D}^0_1\bigl(B_d^{(i,i+1)}\bigr)=\input{./figures/Hadiip.tikz}\).
\end{lemma}

\begin{proof}
For one qudit, the reflected Gray order is the ordinary order
\(0,1,\ldots,d-1\).  The identity factors in
\(B_d^{(i,i+1)}=\id_i\otimes H_{\mathrm{BS}}\otimes\id_{d-i-2}\)
therefore decode to identities, while the two phase shifters
\(P(-\pi/2)\) decode to the corresponding level phases on level \(i+1\) and
the beam splitter \(B(\pi/4)\) decodes to the adjacent
\(R_x^{(i,i+1)}(\pi/4)\).  By the definition of \(H_{\mathrm{BS}}\), the
decoded circuit is the adjacent instance of the Hadamard decomposition
proved in \eqref{hdec}; hence it is equal to \(\input{./figures/Hadiip.tikz}\) in
\(\QCeq\).
\end{proof}

\subsection{Decoding the lifts}

The lift gadgets add Gray digits one layer at a time.  Decoding an ordinary
lift should therefore produce identity padding on the new qudits, while
decoding a selected lift should produce the same padding together with one
extra value-control.

\begin{lemma}\label{lem:lift-arity}
Let \(C\) be a \(\LOPP\) circuit acting on \(d^n\) modes.  Then for every
\(k \ge 0\) the circuit \(\Lift_d^k(C)\) acts on \(d^{n+k}\) modes.
\end{lemma}

\begin{proof}
Induct on \(k\).  The case \(k=0\) is immediate.  For the step,
\(\Lift_d^{k+1}(C)=\bigotimes_{q=0}^{d-1}\Rev^q(\Lift_d^k(C))\); each
\(\Rev^q\) preserves arity, so tensoring \(d\) copies of an
\(d^{n+k}\)-mode circuit yields \(d^{n+k+1}\) modes.
\end{proof}

\begin{lemma}\label{lem:dec-lift}
Let \(C\) be a \(\LOPP\) circuit acting on \(d^n\) modes, and let \(m \ge 0\).
Then \(\mathrm{D}^{0}_{n+m}\bigl(\Lift_d^{m}(C)\bigr)=\id_m \otimes \mathrm{D}^{0}_{n}(C)\) as circuits in \(\CQC\).
\end{lemma}

\begin{proof}
Induct on \(m\).  The case \(m=0\) is immediate.  For the step, set
\(F:=\Lift_d^m(C)\).  Then
\(\Lift_d^{m+1}(C)=\bigotimes_{j=0}^{d-1}\Rev^j(F)\), so decoding serialises
to
\[
  \mathrm D_{n+m+1}^0(\Lift_d^{m+1}(C))
  =
  \prod_{j=0}^{d-1}
  \mathrm D_{n+m+1}^{d^{n+m}j}(\Rev^j(F)).
\]
Corollary~\ref{lem:dec-ctrl-block}, with \(s=0\) and full-block length
\(\ell=d^{n+m}\), gives
\(\mathrm D_{n+m+1}^{d^{n+m}j}(\Rev^j(F))
  =\ctrl_j(\mathrm D_{n+m}^0(F))\) for every \(j\).  By commutativity of
distinct controls and exhaustivity, the product collapses to
\(\id_1\otimes \mathrm D_{n+m}^0(F)\).  The induction hypothesis then yields
\[
  \mathrm D_{n+m+1}^0(\Lift_d^{m+1}(C))
  =
  \id_1\otimes(\id_m\otimes \mathrm D_n^0(C))
  =
  \id_{m+1}\otimes \mathrm D_n^0(C).
\]
\end{proof}

\begin{lemma}\label{lem:dec-selected-lift}
Let \(C\) be a \(\LOPP\) circuit acting on \(d^n\) modes, and let \(m \ge 0\).
Then \(\mathrm{D}^{0}_{n+1+m}\bigl(\Lift_{d,j}^{m}(C)\bigr)=\id_m \otimes \ctrl_{j}\bigl(\mathrm{D}^{0}_{n}(C)\bigr)\) as circuits in \(\CQC\).
\end{lemma}

\begin{proof}
Induct on \(m\).  For \(m=0\), only the \(j\)-th block contains \(C\), and
Corollary~\ref{lem:dec-ctrl-block} turns that block into
\(\ctrl_j(\mathrm D_n^0(C))\), with identities on all other blocks.

For the step, set \(F:=\Lift_{d,j}^m(C)\).  Then
\(\Lift_{d,j}^{m+1}(C)=\bigotimes_{q=0}^{d-1}\Rev^q(F)\), so
\[
  \mathrm D_{n+1+m+1}^0(\Lift_{d,j}^{m+1}(C))
  =
  \prod_{q=0}^{d-1}
  \mathrm D_{n+1+m+1}^{d^{n+1+m}q}(\Rev^q(F)).
\]
Applying Corollary~\ref{lem:dec-ctrl-block} to each full block gives
\(\mathrm D_{n+1+m+1}^{d^{n+1+m}q}(\Rev^q(F))
  =\ctrl_q(\mathrm D_{n+1+m}^0(F))\).  Commutativity and exhaustivity then
collapse the product to \(\id_1\otimes \mathrm D_{n+1+m}^0(F)\), and the
induction hypothesis yields
\(\id_{m+1}\otimes \ctrl_j(\mathrm D_n^0(C))\).
\end{proof}

\subsection{Encoding respects the structural quotient}

This subsection proves Lemma~\ref{lem:encoding-respects-equiv}.  Since the
encoding is defined on raw circuits whereas \(\CQC\) is a quotient by strict
PROP coherence and the control-functor laws, we first record the raw semantic
invariant used in the case analysis.

\begin{lemma}\label{lem:raw-encoding-semantics}
Let \(C:n\to n\) be a raw qudit circuit, interpreted by the same matrix clauses
as in Section~\ref{sec:qudit-circuits} before quotienting by structural
congruence.  For all \(a,b\ge0\),
\(\interp{\mathrm E^a_b(C)}_{\mathrm{LOPP}}
  = I_{d^a}\otimes \interp{C}\otimes I_{d^b}\).
\end{lemma}

\begin{proof}
We argue by induction on the raw syntax of \(C\).  The clauses for identities,
composition, and scalar phases follow immediately from the recursive definition
of \(\mathrm E^a_b\) and the corresponding semantic clauses.

For tensor products, write \(C_1:n_1\to n_1\) and \(C_2:n_2\to n_2\).  By
definition, \(\mathrm E^a_b(C_1\otimes C_2)
  = \mathrm E^{a+n_1}_b(C_2)\circ \mathrm E^a_{b+n_2}(C_1)\).
Using the induction hypotheses, the two factors have Gray-ordered semantics
\(I_{d^{a+n_1}}\otimes\interp{C_2}\otimes I_{d^b}\) and
\(I_{d^a}\otimes\interp{C_1}\otimes I_{d^{n_2+b}}\), respectively.  Their product is
\(I_{d^a}\otimes\interp{C_1}\otimes\interp{C_2}\otimes I_{d^b}\), which is the
required semantics of \(C_1\otimes C_2\) in the middle register.

For the adjacent Hadamard generator, Lemma~\ref{lem:Hd-singlequdit-network}
gives the \(d\)-mode Hadamard semantics of \(B_d^{(i,i+1)}\).  The lift
\(\Lift_d^{a+b}\) duplicates this same action on every branch of the \(a+b\)
context digits, mirroring on those Gray blocks whose reflected order is
reversed.  The two block permutations
\(\sigma^d_{a,b,1}\) and \(\sigma^d_{a,1,b}\) move that lifted action from the
right end of the context block to the middle qudit position.  Hence the total
semantics is \(I_{d^a}\otimes H^{(i,i+1)}\otimes I_{d^b}\).  The swap-generator
case is the same calculation with the block-permutation semantics of
\(\sigma^d_{a,b,c}\): the three permutations in Definition~\ref{def:encoding}
realise the qudit symmetry \(I_{d^a}\otimes\sigma_{1,1}\otimes I_{d^b}\).

For the control constructor, let \(C=\ctrl_j(C')\) with \(C':m\to m\).  By induction,
\(\mathrm E^0_0(C')\) has Gray-ordered semantics \(\interp{C'}\).  The selected
lift \(\Lift_{d,j}^{a+b}\) places this action on the context branches
whose distinguished digit has value \(j\), and places identities on the other
branches.  The surrounding block permutations put the distinguished digit in
front of the encoded target block.  The resulting block-diagonal matrix is
\[
  I_{d^a}\otimes
  \left(
    \ket j\!\bra j\otimes\interp{C'}+
    \sum_{\ell\ne j}\ket \ell\!\bra \ell\otimes I_{d^m}
  \right)
  \otimes I_{d^b},
\]
which is \(I_{d^a}\otimes\interp{\ctrl_j(C')}\otimes I_{d^b}\).
\end{proof}

\begin{lemma}\label{lem:encoding-respects-equiv-app}
For all \(a,b\ge0\), if raw qudit circuits \(C,C':n\to n\) satisfy
\(C\equiv C'\) in the structural congruence of
Definition~\ref{def:structural-congruence-CQC_d}, then
\(\mathrm{LOPP}\vdash \mathrm E^a_b(C)=\mathrm E^a_b(C')\).
\end{lemma}

\begin{proof}
The congruence \(\equiv\) is generated by two kinds of equations: strict
symmetric-monoidal coherence for circuits and the control-functor equations for
each value-control constructor.  By Lemma~\ref{lem:raw-encoding-semantics}, it
is enough to check that each generator has equal raw qudit semantics, because
equal Gray-ordered semantics implies equal single-photon semantics after
conjugating by the fixed Gray permutation; Theorem~\ref{thm:lopp-complete} then
derives the equality in the LOPP equational theory.  This use of LOPP
completeness is external to the qudit completeness theorem being proved.

Strict PROP coherence is sound for the standard qudit interpretation: units
and associativity of \(\circ\), units and associativity of \(\otimes\),
interchange, and the symmetry equations are the matrix equalities for
identity matrices, Kronecker products, and permutation matrices.  Hence the
encodings of both sides have the same optical semantics by
Lemma~\ref{lem:raw-encoding-semantics}.

It remains to list the control-functor generators.  The identity and
composition laws follow from the projector formula:
\(\interp{\ctrl_k(\id_n)}=I_{d^{1+n}}\) and
\(\interp{\ctrl_k(G\circ F)}
  = \interp{\ctrl_k(G)}\,\interp{\ctrl_k(F)}\).
The strength law
\(\ctrl_k(F\otimes\id_m)=\ctrl_k(F)\otimes\id_m\) says that the same
block-diagonal controlled action is tensored with an untouched \(m\)-wire
identity block.  Naturality with respect to a target permutation \(\pi\) is the
equality
\(\interp{\ctrl_k(\pi^{-1}F\pi)}
  = (I_d\otimes P_\pi^{-1})\,
    \interp{\ctrl_k(F)}\,
    (I_d\otimes P_\pi)\),
where \(P_\pi\) is the permutation matrix of \(\pi\).  The same-control
nested-swap law holds because both sides are block diagonal on the two control
digits and the only non-identity target block is the branch where both digits
are \(k\); swapping the two equal control values leaves that branch unchanged.
Thus every generator of \(\equiv\) has equal raw qudit semantics, and the
encoded circuits are derivably equal in LOPP.
\end{proof}

\subsection{Encoding and decoding are inverse}

The remaining step is a structural induction.  The previous lemmas identify
the lift gadgets with identity padding or genuine control, and the block
permutations with the corresponding qudit-wire symmetries.

\begin{proof}
We prove Theorem~\ref{thm:encoding-decoding} by structural induction on \(f\).

\emph{Typing convention.}
Fix \(k,\ell\ge 0\) and a circuit \(f:n\to n\) in \(\CQC\).
Then \(E^{k}_{\ell}(f)\) is a \(d^{k+n+\ell}\)-mode \(\LOPP\) circuit, and the decoding appearing
in the theorem is \(\mathrm{D}^{0}_{k+n+\ell}\).
For readability, throughout this proof we write
\(D(-):=\mathrm{D}^{0}_{k+n+\ell}(-)\) when the arity is clear from context.

\begin{description}
\item[\emph{Empty and identity.}]
If \(f = \input{./figures/empty_diagram.tikz}\), then by definition \(E^{k}_{\ell}(f) = \input{./figures/line.tikz}^{\otimes d^{k+\ell}}\), and hence \(D\bigl(E^{k}_{\ell}(f)\bigr)=D\left(\input{./figures/line.tikz}^{\otimes d^{k+\ell}}\right)=\id_{k+\ell}=\id_k \otimes \id_0 \otimes \id_\ell\).
The case \(f=\input{./figures/line0.tikz}\) is similar.

\item[\emph{Swap.}]
If \(f=\gswap\), the encoding uses the optical block-permutations
\(\sigma^d_{k,n,\ell}\): \(E^{k}_{\ell}(\gswap)=\sigma^d_{k,\ell,2} \circ \sigma^d_{k+\ell,1,1} \circ \sigma^d_{k,2,\ell}\).
These permutations implement the qudit-wire swap by reindexing Gray blocks of
modes.
By Lemma~\ref{lem:dec-swap} from Appendix~\ref{app:swap-encoding} we have \(D\bigl(\sigma^d_{k,\ell,2}\bigr) = \id_k \otimes \sigma_{\ell,2}\), \(D\bigl(\sigma^d_{k+\ell,1,1}\bigr) = \id_{k+\ell}\otimes \sigma_{1,1}\), and \(D\bigl(\sigma^d_{k,2,\ell}\bigr) = \id_k \otimes \sigma_{2,\ell}\).
Therefore, using strict PROP coherence, \(D\bigl(E^{k}_{\ell}(\gswap)\bigr)=\id_k \otimes \sigma_{1,1} \otimes \id_\ell=\id_k \otimes \gswap \otimes \id_\ell\).

\item[\emph{Global phase.}]
If \(f=\input{./figures/phase.tikz}\), then \(E^{k}_{\ell}(\input{./figures/phase.tikz})=\bigl(\input{./figures/lov-phase.tikz}\bigr)^{\otimes d^{k+\ell}}\),
so decoding produces one basis-controlled phase per Gray basis vector of the context register:
\(D\bigl(E^{k}_{\ell}(\input{./figures/phase.tikz})\bigr)=\prod_{u \in [d]^{k+\ell}} \ctrl_{u}\bigl(\input{./figures/phase.tikz}\bigr)\).
By Theorem~\ref{thm:commuting} we may reorder these distinct controlled
branches into the order used in Corollary~\ref{cor:iterated-exhaustivity}.
By iterated exhaustivity of control (Corollary~\ref{cor:iterated-exhaustivity}), the product over all words \(u\) yields \(D\bigl(E^{k}_{\ell}(\input{./figures/phase.tikz})\bigr)=\id_{k+\ell} \otimes \input{./figures/phase.tikz}=\id_k \otimes \input{./figures/phase.tikz} \otimes \id_\ell\).

\item[\emph{Two-level Hadamard.}]
Let \(f = \input{./figures/Hadiip.tikz}\). By Definition~\ref{def:encoding} we have \(E^{k}_{\ell}(\input{./figures/Hadiip.tikz})=\sigma^d_{k,\ell,1}\circ\Lift_d^{k+\ell}\bigl(B_d^{(i,i+1)}\bigr)\circ\sigma^d_{k,1,\ell}\).
Applying \(D\) and functoriality for sequential composition gives \(D\bigl(E^{k}_{\ell}(\input{./figures/Hadiip.tikz})\bigr)=D\bigl(\sigma^d_{k,\ell,1}\bigr)\circ D\Bigl(\Lift_d^{k+\ell}\bigl(B_d^{(i,i+1)}\bigr)\Bigr)\circ D\bigl(\sigma^d_{k,1,\ell}\bigr)\).
By Lemma~\ref{lem:dec-swap}, \(D\bigl(\sigma^d_{k,\ell,1}\bigr) = \id_k \otimes \sigma_{\ell,1}\) and \(D\bigl(\sigma^d_{k,1,\ell}\bigr) = \id_k \otimes \sigma_{1,\ell}\).
By Lemma~\ref{lem:dec-lift}, instantiated with \(C = B_d^{(i,i+1)}\) and \(m = k+\ell\), \(D\Bigl(\Lift_d^{k+\ell}\bigl(B_d^{(i,i+1)}\bigr)\Bigr)=\id_{k+\ell} \otimes \mathrm{D}^{0}_{1}\bigl(B_d^{(i,i+1)}\bigr)\).
Lemma~\ref{lem:dec-Hd-gadget} gives
\(\mathrm{D}^{0}_{1}\bigl(B_d^{(i,i+1)}\bigr) = \input{./figures/Hadiip.tikz}\).
Combining these equalities and simplifying using PROP coherence yields \(D\bigl(E^{k}_{\ell}(\input{./figures/Hadiip.tikz})\bigr) = \id_k \otimes \input{./figures/Hadiip.tikz} \otimes \id_\ell\).

\item[\emph{Sequential composition.}]
If \(f = g \circ h\), then by Definition~\ref{def:encoding}, \(E^{k}_{\ell}(g \circ h) = E^{k}_{\ell}(g) \circ E^{k}_{\ell}(h)\), so \(D\bigl(E^{k}_{\ell}(g \circ h)\bigr)=D\bigl(E^{k}_{\ell}(g)\bigr)\circ D\bigl(E^{k}_{\ell}(h)\bigr)\).
The induction hypothesis gives
\(D(E^{k}_{\ell}(g))=\id_k\otimes g\otimes\id_\ell\) and
\(D(E^{k}_{\ell}(h))=\id_k\otimes h\otimes\id_\ell\), hence \(D\bigl(E^{k}_{\ell}(g \circ h)\bigr) = \id_k \otimes (g\circ h) \otimes \id_\ell\).

\item[\emph{Parallel composition.}]
If \(f = g \otimes h\) with \(g:n_g\to n_g\) and \(h:n_h\to n_h\), then
Definition~\ref{def:encoding} serialises the tensor product:
\(E^{k}_{\ell}(g \otimes h) = E^{k+n_g}_{\ell}(h) \circ E^{k}_{\ell+n_h}(g)\).
Decoding preserves \(\circ\), so \(D\bigl(E^{k}_{\ell}(g \otimes h)\bigr)=D\bigl(E^{k+n_g}_{\ell}(h)\bigr)\circ D\bigl(E^{k}_{\ell+n_h}(g)\bigr)\).
Applying the induction hypotheses to \(h\) at context \((k+n_g,\ell)\) and to \(g\) at context \((k,\ell+n_h)\) yields \(D\bigl(E^{k}_{\ell}(g \otimes h)\bigr)=\id_k\otimes (g\otimes h)\otimes \id_\ell\).

\item[\emph{Control.}]
Let \(f=\ctrl_j(g)\) with \(g:m\to m\).
By Definition~\ref{def:encoding}, \(E^{k}_{\ell}\bigl(\ctrl_j(g)\bigr)=\sigma^d_{k,\ell,m+1}\circ\Lift_{d,j}^{k+\ell}\bigl(E^{0}_{0}(g)\bigr)\circ\sigma^d_{k,m+1,\ell}\).
Decoding and using Lemma~\ref{lem:dec-swap} reduces the outer permutations to the corresponding
qudit-wire symmetries.
For the middle term we use Lemma~\ref{lem:dec-selected-lift} to obtain
\(D\Bigl(\Lift_{d,j}^{k+\ell}\bigl(E^{0}_{0}(g)\bigr)\Bigr)=\id_{k+\ell} \otimes \ctrl_{j}\bigl(\mathrm{D}^{0}_{m}(E^{0}_{0}(g))\bigr)\).
By the induction hypothesis applied to \(g\) with \(k=\ell=0\),
\(\mathrm{D}^{0}_{m}(E^{0}_{0}(g))=g\), hence the middle factor is \(\id_{k+\ell}\otimes\ctrl_j(g)\),
and the surrounding symmetries cancel by PROP coherence.
Therefore \(D\bigl(E^{k}_{\ell}(\ctrl_j(g))\bigr)=\id_k \otimes \ctrl_j(g) \otimes \id_\ell\).
\end{description}

This exhausts the constructors of \(\CQC\), so
\(\QCeq \vdash \mathrm{D}^{0}_{k+n+\ell}(E^{k}_{\ell}(f))=\id_k\otimes f\otimes \id_\ell\)
for all \(f\), as required.
\end{proof}

%% file: appendix_mimicking.tex
\section{Mimicking rules}\label{app:mimicking-rules}

Theorem~\ref{thm:mimicking-rules} says that every LOPP derivation decodes to a
derivation in \(\QCeq\): if \(\mathrm{LOPP}\vdash C_1=C_2\) for \(d^n\)-mode
circuits \(C_1,C_2\), then \(\QCeq\vdash D(C_1)=D(C_2)\).

It is enough to check:
\begin{enumerate}
  \item decoding is compatible with the strict-PROP congruence on raw LOPP terms
        (i.e.\ the strict symmetric monoidal/PROP equations);
  \item decoding maps each non-structural axiom of the LOPP calculus
        (Figure~\ref{fig:axiom_lopp}) to a derivable equation in \(\QCeq\).
\end{enumerate}
We write \(\equiv\) for the strict-PROP congruence on raw LOPP circuits as in
Definition~\ref{def:raw-lopp}; Figure~\ref{fig:axiom-raw} reproduces its
generating equations.

\begin{figure*}[htb]
  \begin{adjustbox}{max width=\textwidth}
  \fbox{\begin{minipage}{1.18\textwidth}
  \newcommand{\nspazer}{0.5em}
  \begin{multicols}{2}%
  \noindent\(\begin{array}[t]{rcccl}
    id_k \circ C& \equiv & C& \equiv &C\circ id_k\\[0.5em]
    \input{./figures/Cidmultifils.tikz}&\equiv&\input{./figures/Cmultifils.tikz}&\equiv&\input{./figures/idCmultifils.tikz}
  \end{array}\)\par
  \vspace{\nspazer}
  \noindent\(\begin{array}[t]{rcl}
    (C_3\circ C_2)\circ C_1 &\equiv & C_3\circ (C_2\circ C_1)\\[0.5em]
    \input{./figures/C1puisC2C3.tikz}&\equiv&\input{./figures/C1C2puisC3.tikz}
  \end{array}\)\par
  \vspace{\nspazer}
  \noindent\(\begin{array}[t]{rcccl}
    \gempty \otimes C &\equiv &C &\equiv &C\otimes  \gempty\\[0.5em]
    \input{./figures/diagrammevidesurC.tikz}&\equiv&\input{./figures/Cmultifils.tikz}&\equiv&\input{./figures/Csurdiagrammevide.tikz}
  \end{array}\)\par
  \vspace{\nspazer}
  \noindent\(\begin{array}[t]{rcl}
    \sigma_{k}\circ (C\otimes \gI) &\equiv  &(\gI\otimes C) \circ     \sigma_{k}\\[0.5em]
    \input{./figures/Cswapn1.tikz}&\equiv&\input{./figures/swapn1C.tikz}
  \end{array}\)\par
  \vspace{\nspazer}
  \noindent\(\begin{array}[t]{rcl}
    (C_1\otimes C_2)\otimes C_3&\equiv &C_1\otimes(C_2\otimes C_3)\\[0.5em]
    \input{./figures/C1surC2puissurC3.tikz}&\equiv&\input{./figures/C1puissurC2surC3.tikz}
  \end{array}\)\par
  \vspace{\nspazer}
  \noindent\(\begin{array}[t]{rcl}
    (C_2\circ C_1)\otimes (C_4\circ C_3) &\equiv &(C_2\otimes C_4)\circ (C_1\otimes C_3)\\[0.5em]
    \input{./figures/C1C2surC3C4.tikz}&\equiv&\input{./figures/C1surC3puisC2surC4.tikz}
  \end{array}\)\par
  \vspace{\nspazer}
  \noindent\(\begin{array}[t]{rcl}
    \gswap \circ \gswap &\equiv &\gI\otimes \gI\\[0.5em]
    \input{./figures/swapswap.tikz}&\equiv&\input{./figures/two-lines.tikz}
  \end{array}\)\par
  \end{multicols}
  where \(id_0 = \gempty\) and \(id_{k+1} = id_k\otimes \gI\), and
  \(\sigma_{0} := \gI\),
  \(\sigma_{k+1} := (\gswap\otimes id_k )\circ
  (\gI\otimes \sigma_k)\).
  \end{minipage}}
  \end{adjustbox}
  \caption{Structural axioms generating the congruence \(\equiv\) on raw
  circuits (either raw quantum circuits or raw optical circuits).}
  \Description{Structural (strict symmetric monoidal/PROP) equations on raw circuits:
  unit and associativity for composition and tensor; interchange; naturality of swap; and involutivity of swap.
  Each equation is shown both algebraically and as a string diagram.}
  \label{fig:axiom-raw}
\end{figure*}

\subsection*{Commutation of decoded circuits on disjoint blocks}

The first check is a commuting property for decoded optical subcircuits on
disjoint mode intervals.

\begin{lemma}\label{lem:decoded-disjoint-generators-commute}
Decoded LOPP generators supported on disjoint mode intervals commute in
\(\QCeq\).
\end{lemma}

\begin{proof}
By Lemma~\ref{lem:decoding-locality}, such a decoded generator is supported on
one Gray basis word or on one Gray-neighbouring pair.  Disjoint mode intervals
therefore give disjoint Gray supports.  After applying the \Separate{} normal
form of Lemma~\ref{lemma-Greducible}, every resulting contextual
\(\mathcal G\)-factor has one of two forms relative to a factor from the other decoded generator: either
the surrounding control words require incompatible basis values on some qudit,
or the active target windows are disjoint.  In the first case the pairwise
swap is one of the distinct-control commutations used in
Theorem~\ref{thm:commuting}; in the second case it is strict
symmetric-monoidal interchange and naturality.  Composing these finite
pairwise swaps commutes the two decoded generators.
\end{proof}

\begin{lemma}\label{decodagefilsdisjoints}
Let \(C_i:\ell_i\to\ell_i\) be raw optical circuits for \(i=1,2\).  Let
\(n\ge0\), and let \(t_1,t_2\) be nonnegative integers satisfying
\(t_i+\ell_i\le d^n\) for \(i=1,2\).  If the intervals
\([t_1,t_1+\ell_1)\) and \([t_2,t_2+\ell_2)\) are disjoint, then
\(\mathrm{QC}_{d}\vdash
\mathrm{D}^{t_2}_{n}(C_2)\circ\mathrm{D}^{t_1}_{n}(C_1)
=
\mathrm{D}^{t_1}_{n}(C_1)\circ\mathrm{D}^{t_2}_{n}(C_2)\).
\end{lemma}

\begin{proof}
The disjointness assumption gives either \(t_1+\ell_1\le t_2\) or
\(t_2+\ell_2\le t_1\).  Since the conclusion is symmetric in the two pairs up
to reversing the displayed equality, it suffices to treat the first case.
Assume therefore that \(t_1+\ell_1\le t_2\).

We argue by structural induction on the pair \((C_1,C_2)\).

If either circuit is a sequential composite, the statement follows immediately
from the induction hypothesis and the defining clause
\(\mathrm{D}^{t}_{n}(C_2\circ C_1)=\mathrm{D}^{t}_{n}(C_2)\circ \mathrm{D}^{t}_{n}(C_1)\).

If either circuit is a tensor composite, we use the serialisation clause
\(\mathrm{D}^{t}_{n}(C_1\otimes C_2)=\mathrm{D}^{t+\ell_1}_{n}(C_2)\circ \mathrm{D}^{t}_{n}(C_1)\)
(Definition~\ref{def:decoding}) together with associativity of \(\circ\) to reduce
to commuting decoded subcircuits whose acted-on mode intervals remain disjoint;
the required commutations are then provided by the induction hypothesis.

The induction therefore reduces to the case where both \(C_1\) and \(C_2\) are
LOPP generators (phase shifter, beam splitter, or swap).
This proves Lemma~\ref{lem:decoded-disjoint-generators-commute}.
\end{proof}

\subsection*{Contexts and compatibility with structural congruence}

\begin{definition}
A (one-hole) raw context \(\C[\cdot]_i\) is a term over the LOPP signature built
from \(\circ\) and \(\otimes\) with a distinguished placeholder of arity \(i\).
If its outer arity is \(N\), we call it an \(N\)-mode context.  For any
\(i\)-mode raw circuit \(C\), write \(\C[C]\) for the circuit obtained by
filling the hole with \(C\).
\end{definition}

\begin{lemma}\label{decodingtoporules}
Let \(\C[\cdot]_i\) be an \(N\)-mode raw context with a hole of arity \(i\), let
\(C_1,C_2:i\to i\) be raw optical circuits with \(C_1 \equiv C_2\), and let
\(t,n\) satisfy \(t+N\le d^n\). Then
\(\mathrm{QC}_{d}\vdash
\mathrm{D}^{t}_{n}(\C[C_1])=\mathrm{D}^{t}_{n}(\C[C_2])\).
In particular, if \(C_1,C_2\) are full \(d^n\)-mode circuits and
\(C_1\equiv C_2\), then \(\mathrm{QC}_d\vdash D(C_1)=D(C_2)\).
\end{lemma}

\begin{proof}
By induction on the context structure it suffices to check preservation of the
generating axioms of Figure~\ref{fig:axiom-raw}. Identities and tensor units are
sent to \(\id_n\) by Definition~\ref{def:decoding}. Associativity of \(\circ\) is
preserved because \(\mathrm{D}^{t}_{n}\) is defined compositionally on \(\circ\).
Associativity of \(\otimes\) is preserved because \(\mathrm{D}^{t}_{n}\) serialises \(\otimes\)
into a fixed sequential composite, and the two bracketings decode to the same
composite up to associativity of \(\circ\).
For the mixed-product/interchange axiom in Figure~\ref{fig:axiom-raw},
expanding both sides
via the serialisation clause yields composites of decoded subcircuits acting on
disjoint mode blocks; Lemma~\ref{decodagefilsdisjoints} provides the required
commutations.
It remains to check the two structural equations involving raw mode swaps.  An
adjacent optical swap on modes \(r,r+1\) decodes, by
Definition~\ref{def:decoding} and Lemma~\ref{lem:gray-neighbours}, to
\(\Lambda^u_w(X^{(v,v+\varepsilon)})\), i.e.\ to the controlled adjacent
two-level transposition of the two Gray basis words labelling those modes.
Its involutivity is therefore the controlled instance of \eqref{Xij-Xij},
tensored with identities and conjugated by the placement permutation
\(\sigma_{|u|,|w|,1}\).

For naturality, first factor the raw structural permutation \(\sigma_k\) into
adjacent optical swaps, as in Figure~\ref{fig:axiom-raw}.  Decoding this
factorisation gives the corresponding product of adjacent controlled
two-level transpositions on Gray basis words.  The relations
\eqref{Xij-Xij}, \eqref{Xij-Xjkbis}, and \eqref{Xij-Xjk} give the involutive
and braid moves needed to make this decoded permutation independent of the
chosen adjacent-swap factorisation.  Conjugating a decoded generator by one
adjacent decoded swap relabels its Gray support: for phases this is the
phase-through-\(X\) calculation \eqref{cxp}; for decoded swaps it is the same
adjacent-transposition calculus; and for beam splitters it is the identical
calculation with \(X\) replaced by the adjacent \(R_x\)-gate, using
\eqref{RxSym} for orientation reversal.  The statement for an arbitrary raw
circuit \(C\) then follows by induction on \(C\), using the already proved
composition and tensor clauses.  Hence the decoded forms of raw swap naturality
and swap involutivity are derivable in \(\QCeq\).
\end{proof}

\subsection*{Mimicking the LOPP equational theory}

Recall the non-structural axioms of LOPP from~\cite{LOPP-min}
(Figure~\ref{fig:axiom_lopp}).

\begin{figure}[htb]
  \centering
  \fbox{\begin{minipage}{0.975\textwidth}
    \centering
    \begin{minipage}[c]{0.28\textwidth}
      \centering
\schemarule{A}{lopp-eq:sum}{
\input{./figures/lov-two-phase.tikz} = \input{./figures/lov-phase-add.tikz}
}
    \end{minipage}\hfill
    \begin{minipage}[c]{0.18\textwidth}
      \centering
\schemarule{\ensuremath{2\pi}}{lopp-eq:2pi}{
\input{./figures/lov-phase-2pi.tikz} = \input{./figures/line.tikz}
}
    \end{minipage}
    \begin{minipage}[c]{0.32\textwidth}
      \centering
\schemarule{SW}{lopp-eq:swap}{
\input{./figures/swap.tikz} = \input{./figures/swap-equiv.tikz}
}
    \end{minipage}
    \par\smallskip
    \begin{minipage}[c]{0.72\textwidth}
      \centering
\schemarule{E}{lopp-eq:euler}{
\input{./figures/3beamsplitter-left.tikz} = \input{./figures/3beamsplitter-right.tikz}
}
    \end{minipage}
    \par\smallskip
    \begin{minipage}[c]{0.88\textwidth}
      \centering
\schemarule{3BS}{lopp-eq:3bs}{
\input{./figures/3xrot-left.tikz} = \input{./figures/3xrot-right.tikz}
}
    \end{minipage}
  \end{minipage}}
  \caption{Non-structural axioms of the LOPP calculus:
  phase addition, \(2\pi\)-periodicity, swap decomposition, and
  Euler/three-beamsplitter identities.}
  \Description{Five non-structural LOPP rewrite axioms: phase addition on one mode;
  \(2\pi\)-periodicity; a swap decomposition into phases and beam splitters; and two
  standard identities involving three beam splitters (Euler form and a 3-beam-splitter relation).}
  \label{fig:axiom_lopp}
\end{figure}

\begin{lemma}\label{lem:mimicking-criterion}
Let a source equational theory be the congruence generated by a structural
congruence \(\equiv\) and a set of non-structural axiom schemata.  Suppose a
decoding into a target theory respects \(\equiv\) in every raw context and
that every contextual instance of each non-structural source axiom decodes to
a derivable target equation.  Then every source derivation decodes to a target
derivation.
\end{lemma}

\begin{proof}
The source derivability relation is the least congruence containing the
structural congruence and the listed non-structural axioms.  The two
hypotheses give the result on those generators, and closure under contexts,
composition, tensor, symmetry, and transitivity follows by the congruence
rules of the target theory.
\end{proof}

\begin{proof}[Proof of Theorem~\ref{thm:mimicking-rules}]
By Lemma~\ref{lem:mimicking-criterion} and
Lemma~\ref{decodingtoporules}, the remaining check is that each non-structural
axiom of Figure~\ref{fig:axiom_lopp} is preserved by decoding
at an arbitrary offset \(t\).

Let \(t,n\) be such that the relevant local LOPP circuit sits inside \(d^n\) modes.
When \(G_n^d(t)=uiv\) and \(G_n^d(t+1)=ujv\) with \(|i-j|=1\), write
\(\Lambda^u_v(G^{(i,j)})\) for the qudit circuit obtained by placing the
one-qudit gate \(G^{(i,j)}\) on the changed coordinate and controlling on the
surrounding word \(uv\).  The notation abbreviates the decoding clause in
Definition~\ref{def:decoding}; if the Gray edge is traversed in the opposite
orientation, it is read with the two levels exchanged, using \eqref{RxSym} and
\(X^{(i,j)}=X^{(j,i)}\).

\begin{itemize}
  \item \emph{Phase addition}~\eqref{lopp-eq:sum}.
    Decoding a phase shifter on mode \(t\) yields
    \(\ctrl_{G_n^d(t)}(\input{./figures/phase.tikz})\) (Definition~\ref{def:decoding}).
    The decoded equation is therefore an instance of the \(\QCeq\) phase
    arithmetic axiom~\eqref{sum}, closed under iterated control.

  \item \emph{\(2\pi\)-periodicity}~\eqref{lopp-eq:2pi}.
    Similarly, decoding \(\input{./figures/lov-phase-2pi.tikz}\) yields a multi-controlled
    \(2\pi\)-phase, which is equal to \(\id_n\) by the \(\QCeq\) axiom~\eqref{2pi}
    (again closed under iterated control). The right-hand side decodes to \(\id_n\).

  \item \emph{Swap decomposition}~\eqref{lopp-eq:swap}.
    The swap decomposition is local to the two modes.  If
    \(G_n^d(t)=uiv\) and \(G_n^d(t+1)=ujv\), then the left-hand side decodes to
    \(\Lambda^u_v(X^{(i,j)})\).  The right-hand side is the optical
    phase--beam-splitter--phase decomposition of the same two-mode swap; after
    decoding, the beam splitter is \(\Lambda^u_v(R_x^{(i,j)}(\pi/2))\) and the
    phases are the corresponding controlled level phases.  Expanding
    \(R_x^{(i,j)}\) and \(X^{(i,j)}\) by Appendix~\ref{app:derived-notation},
    and using \eqref{sum}, \eqref{2pi}, \eqref{Htotal}, and
    \eqref{Hij-Hij} under the surrounding controls, gives the same controlled
    swap \(\Lambda^u_v(X^{(i,j)})\).  Thus the decoded swap axiom is a
    controlled instance of the two-level swap decomposition already available
    in \(\QCeq\); the global block-permutation version used for contexts is
    proved separately in Appendix~\ref{app:swap-encoding}.

  \item \emph{Euler identity}~\eqref{lopp-eq:euler}.
    Again let the two modes have Gray labels \(uiv\) and \(ujv\).
    Phase shifters on either optical mode decode to level phases on \(i\) or
    \(j\), and beam splitters decode to \(R_x^{(i,j)}\) on the same
    two-level subspace, all under the common surrounding control \(uv\).
    After this replacement, the decoded optical Euler axiom is the
    surrounding-control instance \(\Lambda^u_v(\eqref{eulerH})\), with the
    deterministic parameter convention fixed in
    Appendix~\ref{appendix:relations_angles}.  The orientation choice is
    whether the edge is read as \(i\to j\) or \(j\to i\);
    \eqref{RxSym} and the symmetry of the derived \(H^{(i,j)}\) notation reduce
    the reflected case to the displayed one.

  \item \emph{Three-beamsplitter identity}~\eqref{lopp-eq:3bs}.
    Put \(w_r=G_n^d(t+r)\) for \(r=0,1,2\).  By
    Lemma~\ref{lem:gray-neighbours}, \(w_0\) and \(w_1\) differ in one digit,
    and \(w_1\) and \(w_2\) differ in one digit.  Two cases remain.
    If the same digit changes twice, then, after fixing the surrounding word,
    the three labels are \(uiv,u(i+\varepsilon)v,u(i+2\varepsilon)v\) for
    \(\varepsilon=\pm1\).  The decoded equation is the controlled instance of
    the one-qudit three-rotation axiom \eqref{3Rx}, with the angle convention
    of Appendix~\ref{appendix:relations_angles} and with \eqref{RxSym} used
    when the reflected block gives the decreasing orientation.

    If the two changed digits are different, the three labels form a corner of
    a two-coordinate square: after a wire permutation and renaming of adjacent
    levels they have the shape
    \(u\,i\,a\,v,\ u\,j\,a\,v,\ u\,j\,b\,v\), with \(|i-j|=|a-b|=1\).
    The first and third rotations act on different qudit coordinates, while
    the middle one carries the fixed value of the other coordinate as a
    value-control.  The resulting equation is the surrounding-control instance
    of the derived mixed rule \eqref{3CRx}, whose parameters are the same as
    for \eqref{3Rx}.
\end{itemize}

Decoding preserves each generator of the LOPP equational theory and is stable
under contexts, so \(\mathrm{LOPP}\vdash C_1=C_2\) implies
\(\QCeq\vdash D(C_1)=D(C_2)\).
\end{proof}

%% file: appendix_dec_swap.tex
\section{Explicit encoding and decoding of swaps}\label{app:swap-encoding}

This appendix shows that the optical block-swaps used by the encoding decode
to the intended qudit-level swaps in \(\CQC\).

\subsection{Gray-code indices for words}

A mode in the \(d\)-ary Gray code for \(n\) qudits is labelled by a word
\(w \in [d]^n\).  The inverse map below returns its position in the Gray
ordering.

\begin{definition}
For a base-\(d\) word \(w \in [d]^n\) we write \(G^{-1}_{d}(w)\) for its index
in the \(d\)-ary Gray ordering.  It is defined recursively as follows.
\begin{itemize}
  \item For the empty word \(\epsilon\) we set \(G^{-1}_{d}(\epsilon) = 0\).
  \item For a non-empty word \(w = h \cdot t\) of length \(n+1\), with
        \(h \in \{0,\dots,d-1\}\) the first digit and \(t \in [d]^n\) the suffix,
        we put
        \[
          G^{-1}_{d}(h\cdot t) =
          \begin{cases}
            h \, d^{n} + G^{-1}_{d}(t) & \text{if }h\text{ is even},\\[0.4ex]
            (h+1) \, d^{n} - 1 - G^{-1}_{d}(t) & \text{if }h\text{ is odd}.
          \end{cases}
        \]
\end{itemize}
\end{definition}

\subsection{Elementary LOPP swap gadgets}

The elementary gadget swaps two basis states.

\begin{definition}
Let \(w_1,w_2\in[d]^n\) be two distinct words of the same length.  The notation
\(C_{d,w_1,w_2}\) denotes the raw \(\LOPP\) word-swap constructed in
Definition~\ref{def:adjacent-word-swaps}.  Its intended action is the mode
transposition that swaps the two modes labelled by \(w_1\) and \(w_2\), and
acts as the identity on all other modes.
\end{definition}

We also use this notation in structured instances such as
\(C_{d,\beta \cdot q_1 \cdot \alpha,\ \beta \cdot p_1 \cdot \alpha}\), where
\(\beta\) and \(\alpha\) are fixed prefixes and suffixes.  Decoding then
produces the expected two-level transposition with surrounding controls
determined by \(\beta\) and \(\alpha\).  The saturated version below ranges
over all such prefixes and suffixes so that commutativity and exhaustivity can
later remove those extra controls.

\begin{definition}
For integers \(k,\ell \ge 0\) and digits \(q_1,q_2,p_1,p_2 \in \{0,\dots,d-1\}\)
we define
\[
  C^{\star}_{d,k,\ell,[q_1,q_2],[p_1,p_2]}
  =
  \prod_{\beta\in[d]^k}
  \prod_{\alpha\in[d]^{\ell}}
  C_{d,\ \beta \cdot q_1 \cdot q_2 \cdot \alpha,\ \beta \cdot p_1 \cdot p_2 \cdot \alpha}.
\]
Here \(\beta\) and \(\alpha\) range over words of lengths \(k\) and \(\ell\),
respectively, both in lexicographic increasing order with \(\beta\) as the
outer index; for \(k=0\) or \(\ell=0\) the corresponding word is empty.
\end{definition}

\(C^{\star}\) applies the same three-CNOT pattern for every choice of prefix
\(\beta\) and suffix \(\alpha\).  The decoding calculation below first
identifies one two-position factor and then uses commutativity and exhaustivity
of controls to remove the saturated prefix and suffix controls.

\subsection{Adjacent construction of word-swaps}

For \(n\ge 1\) and \(0\le r<d^n-1\), write
\(A^n_r:=id_r\otimes\input{./figures/swap.tikz}\otimes id_{d^n-r-2}\) for the adjacent
optical swap of modes \(r\) and \(r+1\).

\begin{definition}\label{def:adjacent-word-swaps}
Let \(w_1,w_2\in[d]^n\) be distinct.  Put \(i=G^{-1}_{d}(w_1)\) and
\(j=G^{-1}_{d}(w_2)\).  If \(i>j\), define
\(C_{d,w_1,w_2}:=C_{d,w_2,w_1}\).  If \(i<j\), let
\(r_1,\ldots,r_{2(j-i)-1}\) be the concatenation of
\(i,i+1,\ldots,j-1\) and \(j-2,j-3,\ldots,i\), with the second list empty
when \(j=i+1\), and set
\(C_{d,w_1,w_2}:=A^n_{r_{2(j-i)-1}}\circ\cdots\circ A^n_{r_1}\).
\end{definition}

This finite product is well-typed at arity \(d^n\).  As a permutation of
optical modes, the first \(j-i\) adjacent swaps move mode \(i\) to mode \(j\),
and the remaining adjacent swaps return every intermediate mode to its original
position.  Thus \(C_{d,w_1,w_2}\) implements the transposition of the two modes
labelled by \(w_1\) and \(w_2\).

\subsection{Decoding of the adjacent swap construction}

The next lemmas record the effect of decoding these gadgets.

\begin{lemma}\label{lemma-dec-C}
For any distinct digits \(p_1,q_1\) and words \(\alpha,\beta\), put
\(w_q=\beta\cdot q_1\cdot\alpha\) and
\(w_p=\beta\cdot p_1\cdot\alpha\).  Then
\(D(C_{d,w_q,w_p})=\input{./figures/oneCNOTab.tikz}\).  It is a single two-level
\(X^{(p_1,q_1)}\) on the appropriate qudit, with the surrounding controls
determined by \(\beta\) and \(\alpha\).
\end{lemma}

\begin{proof}
By Definition~\ref{def:adjacent-word-swaps}, \(C_{d,w_q,w_p}\) is a
factorisation of the optical transposition between the two Gray modes \(w_q\)
and \(w_p\) into adjacent optical swaps.  Definition~\ref{def:decoding} and
Lemma~\ref{lem:gray-neighbours} decode each adjacent optical swap as the
controlled adjacent two-level transposition on the corresponding neighbouring
Gray words.  The relations \eqref{Xij-Xij}, \eqref{Xij-Xjkbis}, and
\eqref{Xij-Xjk} give the usual adjacent-transposition calculus, so the decoded
palindromic factorisation reduces to the transposition of \(w_q\) and \(w_p\)
and fixes all other Gray basis words.  Since these two words differ only in the
displayed digit, the resulting qudit circuit is the controlled
\(X^{(p_1,q_1)}\) shown as \(\input{./figures/oneCNOTab.tikz}\).  The calculation is
internal to \(\QCeq\), using the decoded adjacent swaps and the derived
transposition rules; it does not invoke the final completeness theorem.
\end{proof}

\begin{lemma}\label{lem:dec-C-two-digits}
For digits \(p_1,p_2,q_1,q_2\) with \(p_1\ne q_1\) and \(p_2\ne q_2\), and for
words \(\alpha,\beta\), put
\(w_q=\beta\cdot q_1\cdot q_2\cdot\alpha\) and
\(w_p=\beta\cdot p_1\cdot p_2\cdot\alpha\).  Then
\(D(C_{d,w_q,w_p})=\input{./figures/threeCNOTab.tikz}\).
\end{lemma}

\begin{proof}
By Definition~\ref{def:adjacent-word-swaps} and the same decoded
adjacent-transposition calculation used in Lemma~\ref{lemma-dec-C},
\(D(C_{d,w_q,w_p})\) is the basis-state transposition exchanging \(w_q\) and
\(w_p\).  Let \(w_m:=\beta\cdot q_1\cdot p_2\cdot\alpha\).  The three
single-position transpositions \((w_q\,w_m)\), \((w_m\,w_p)\), and
\((w_q\,w_m)\) compose to \((w_q\,w_p)\); in \(\QCeq\) this is the controlled
instance of the transposition rules \eqref{Xij-Xjk} and \eqref{Xij-Xij}.
Lemma~\ref{lemma-dec-C} identifies these three single-position transpositions
with the three factors in \(\input{./figures/threeCNOTab.tikz}\), giving the claimed
decoding.
\end{proof}

\begin{lemma}
\label{lem:dec-Cstar-concrete}
For all \(k,\ell\ge 0\) and digits \(p_1,p_2,q_1,q_2\) with
\(p_1\ne q_1\) and \(p_2\ne q_2\), we have
\(D\bigl(C^{\star}_{d,k,\ell,[q_1,q_2],[p_1,p_2]}\bigr)
=id_k \otimes \input{./figures/threeCNOT.tikz} \otimes id_\ell\).
\end{lemma}

\begin{proof}
By definition, \(D(C^{\star}_{d,k,\ell,[q_1,q_2],[p_1,p_2]})\) is the product,
over all prefixes \(\beta\) and suffixes \(\alpha\), of the decoded factors
\(D(C_{d,\beta \cdot q_1 \cdot q_2 \cdot \alpha,\,
\beta \cdot p_1 \cdot p_2 \cdot \alpha})\).  Lemma~\ref{lem:dec-C-two-digits}
identifies each factor as \(\input{./figures/threeCNOTab.tikz}\).  These factors act on the
same pair of target levels and differ only by their surrounding control words; by
commutation of distinct controls (Thm.~\ref{thm:commuting}) we may group the
first CNOT-style component of every factor, then the middle component of every
factor, and finally the last component of every factor.  For each of these three
groups, the product ranges over every prefix \(\beta\in[d]^k\) and suffix
\(\alpha\in[d]^\ell\), so iterated exhaustivity
(Corollary~\ref{cor:iterated-exhaustivity}) removes all surrounding controls.
The three collapsed groups form the displayed three-CNOT pattern, tensored with
identities on the untouched wires.
This yields the claimed form.
\end{proof}

\subsection{Decoding block-swaps}

Block-swaps are assembled from these saturated gadgets.

\begin{definition}
For integers \(k,\ell \ge 0\) we define the LOPP block-swap
\(\sigma^{d}_{k,\ell,1}\) that exchanges a block of \(\ell\) qudit wires with a
single qudit wire to its right by
\[
  \sigma^{d}_{k,\ell,1}
  =
  \prod_{j=0}^{\ell-1}
    \prod_{x=0}^{d-2}
      \prod_{y=x+1}^{d-1}
        C^{\star}_{d,k+j,\ell-j-1,\,[x,y],[y,x]}.
\]
The products are read in lexicographic increasing order of the displayed
indices \(j,x,y\), with \(j\) outermost.
\end{definition}

\begin{lemma}\label{lem:dec-swap-1}
We have
\(
  D\left(\sigma^{d}_{k,\ell,1}\right)
  =
  id_k \otimes \sigma_{\ell,1},
\)
where \(\sigma_{\ell,1}\) is the usual PROP symmetry swapping an \(\ell\)-wire
block with a single wire.
\end{lemma}

\begin{proof}
Using Lemma~\ref{lem:dec-Cstar-concrete} we obtain
\[
\begin{aligned}
  D\left(\sigma^{d}_{k,\ell,1}\right)
  &=
  \prod_{j=0}^{\ell-1}\ \prod_{x=0}^{d-2}\ \prod_{y=x+1}^{d-1}
  D\left(
    C^{\star}_{d,k+j,\ell-j-1,[x,y],[y,x]}
  \right)\\[1ex]
  &=
  \prod_{j=0}^{\ell-1}\ \prod_{x=0}^{d-2}\ \prod_{y=x+1}^{d-1}
  \left(
    id_{k+j} \otimes \input{./figures/threeCNOTi.tikz} \otimes id_{\ell-j-1}
  \right)\\[1ex]
  &=
  \prod_{j=0}^{\ell-1}
  \Bigl(
    id_{k+j}
    \otimes
    \underbrace{\Bigl(
      \prod_{x<y} \input{./figures/threeCNOTi.tikz}
    \Bigr)}_{\overset{\eqref{swap-decomp}}{=} \sigma_{1,1}}
    \otimes
    id_{\ell-j-1}
  \Bigr).
\end{aligned}
\]
The product over \(x<y\) is the adjacent swap \(\sigma_{1,1}\) by the swap
decomposition rule~\eqref{swap-decomp}.  Factoring out the leftmost \(k\)
wires, we obtain
\[
\begin{aligned}
  D\left(\sigma^{d}_{k,\ell,1}\right)
  &= id_k \otimes
     \Bigl(
       \prod_{j=0}^{\ell-1}
         \bigl(id_j \otimes \sigma_{1,1} \otimes id_{\ell-j-1}\bigr)
     \Bigr)\\
  &= id_k \otimes \sigma_{\ell,1},
\end{aligned}
\]
since the product on the right is the standard factorisation of a
block-swap into \(\ell\) adjacent transpositions.
\end{proof}

The next lemma passes from swapping one wire with an \(\ell\)-block to swapping
two arbitrary blocks.

\begin{lemma}\label{lem:dec-swap}
For \(c=1\), take \(\sigma^{d}_{a,b,1}\) as defined above.  For \(c>1\), define
inductively
\(
  \sigma^{d}_{a,b,c}
  :=
  \sigma^{d}_{a,b+c-1,1} \circ \sigma^{d}_{a,b+1,c-1}.
\)
Then
\(
  D\left(\sigma^{d}_{a,b,c}\right)
  =
  id_a \otimes \sigma_{b,c}.
\)
\end{lemma}

\begin{proof}
By definition,
\[
\begin{aligned}
  D\left(\sigma^{d}_{a,b,c}\right)
  &= D\left(\sigma^{d}_{a,b+c-1,1} \circ \sigma^{d}_{a,b+1,c-1}\right)\\
  &= D\left(\sigma^{d}_{a,b+c-1,1}\right)
     \circ
     D\left(\sigma^{d}_{a,b+1,c-1}\right)\\
  &\overset{\cref{lem:dec-swap-1}}{=}
     \bigl(id_a \otimes \sigma_{b+c-1,1}\bigr)
     \circ
     D\left(\sigma^{d}_{a,b+1,c-1}\right).
\end{aligned}
\]
By the induction hypothesis on \(c\) we have
\(D(\sigma^{d}_{a,b+1,c-1}) = id_a \otimes \sigma_{b+1,c-1}\), so
\[
\begin{aligned}
  D\left(\sigma^{d}_{a,b,c}\right)
  &= id_a \otimes
     \bigl(
       \sigma_{b+c-1,1} \circ \sigma_{b+1,c-1}
     \bigr)\\
  &= id_a \otimes \sigma_{b,c},
\end{aligned}
\]
since \(\sigma_{b,c}\) is given by first moving the rightmost wire across the
block of size \(b+c{-}1\) and then swapping the resulting \((b+1)\)-block with
the remaining \(c{-}1\) wires.  This gives the composition
\(\sigma_{b+c-1,1} \circ \sigma_{b+1,c-1}\).
\end{proof}

%% file: appendix_soundness.tex
\section{Axiom soundness checks}\label{app:axiom-soundness}

This appendix gives the semantic verification used in the soundness direction
of Theorem~\ref{thm:main-completeness}.  Tables~\ref{tab:core-axiom-soundness}
and~\ref{tab:control-axiom-soundness} list the case-by-case checks: the first
treats the core gate-algebra schemata of Figure~\ref{fig:axioms_QC_part1}, and
the second treats the control-support schemata of
Figure~\ref{fig:axioms_QC_part2}.
Closure of \(\QCeq\) under sequential composition, tensor product, structural
symmetries, and value-control then preserves soundness for every contextual
instance of the schemata.

\begin{table}[htbp]
\small
\renewcommand{\arraystretch}{1.18}
\begin{tabularx}{\textwidth}{@{}p{0.25\textwidth}X@{}}
\hline
\textbf{Schema} & \textbf{Essential soundness check} \\
\hline
\eqref{sum}, \eqref{2pi} &
Immediate from \(e^{i\alpha}e^{i\beta}=e^{i(\alpha+\beta)}\) and
\(e^{2\pi i}=1\); any value-control only adds orthogonal branch projectors. \\
\eqref{XH} &
Since \(X^{(i,j)}\) is the transposition of levels \(i\) and \(j\), conjugating
\(H^{(j,k)}\) by \(X^{(i,j)}\) relabels its active pair to \(H^{(i,k)}\).
Hence both sides agree on \(\mathrm{span}\{\ket i,\ket j,\ket k\}\) and are
identity elsewhere. \\
\eqref{Hiip-Hiip} &
The adjacent Hadamard squares to the identity on
\(\mathrm{span}\{\ket i,\ket{i+1}\}\) and is identity on the orthogonal
complement. \\
\eqref{eulerH} &
Lemma~\ref{lem:EH-angle-soundness} proves the required \(2\times2\) matrix
identity for the admissible angle parameters; both sides are identity on the
orthogonal complement. \\
\eqref{3Rx} &
Lemma~\ref{lem:3Rx-angle-soundness} proves the corresponding \(2\times2\)
identity, including the degenerate angle cases; outside the selected two-level
subspace both sides are identity. \\
\eqref{CX-XC-CX} &
Both sides induce the same permutation of computational basis states. \\
\eqref{swap-decomp} &
On \(\ket p\ket q\), every factor is identity when \(p=q\), while for
\(p\neq q\) the factor indexed by
\((\min\{p,q\},\max\{p,q\})\) acts and sends \(\ket p\ket q\) to
\(\ket q\ket p\). \\
\hline
\end{tabularx}
\caption{Semantic verification of the core gate-algebra schemata of
Figure~\ref{fig:axioms_QC_part1}.}
\label{tab:core-axiom-soundness}
\end{table}

\begin{table}[htbp]
\small
\renewcommand{\arraystretch}{1.18}
\begin{tabularx}{\textwidth}{@{}p{0.25\textwidth}X@{}}
\hline
\textbf{Schema} & \textbf{Essential soundness check} \\
\hline
\eqref{axiom-total-phase}, \eqref{axiom-total-hadamard} &
If \(U\) denotes the scalar phase or the adjacent Hadamard, the left-hand side
is \(\bigoplus_{k\in[d]} U = I_d\otimes U\), the denotation of the
right-hand side. \\
\eqref{HHcomm}, \eqref{Hphasecomm}, \eqref{Hcnot} &
The side conditions in Figure~\ref{fig:axioms_QC_part2} ensure that the
nontrivial basis states moved by the two gates are disjoint.
Proposition~\ref{prop:qudit-support-sensitive} then implies that the
corresponding block-diagonal operators commute. \\
\eqref{cp-cp}, \eqref{cp-cnot}, \eqref{cnot-cnot-diff},
\eqref{ch-ccnot}, \eqref{ch-ch}, \eqref{ch-ccp}, \eqref{cnot-cnot} &
Writing the two denotations with control projectors \(P_k\) and \(P_\ell\) for
\(k\neq\ell\), their non-identity blocks are supported on orthogonal direct
summands because \(P_kP_\ell=0\).  Therefore the operators commute. \\
\eqref{axiom-compat}, \eqref{axiom-compat-pi} &
Decompose with respect to the basis \(\ket a\ket b\) of the two control wires.
On every block except the designated branch pair, both sides are identity; on
that block, after the swap of the control wires, both sides apply the same
target unitary.  Hence the block-diagonal matrices coincide. \\
\hline
\end{tabularx}
\caption{Semantic verification of the control-support schemata of
Figure~\ref{fig:axioms_QC_part2}.}
\label{tab:control-axiom-soundness}
\end{table}